\documentclass[runningheads]{llncs}%
\usepackage{graphicx}
\usepackage[figuresright]{rotating}
\usepackage{epstopdf}
\usepackage{amsmath}
\usepackage{cite} 
\usepackage{algorithm}
\usepackage{algorithmic}
\usepackage{booktabs}
\usepackage{float}
\usepackage{rotating}
\usepackage{array}
\graphicspath{{Pictures/}}
\newcommand\pN{\mathcal{N}}

\usepackage{amsmath,amssymb,graphicx,adjustbox,multirow,hhline}
\usepackage{hyperref}
\usepackage{lscape}
\usepackage{wrapfig}

\usepackage[utf8]{inputenc}
\usepackage[english]{babel}

\long\def\comment#1{}

\newcommand{\beq}{\begin{equation*}}
\newcommand{\eeq}{\end{equation*}}

\newfont{\bbb}{msbm10 scaled 700}

\newfont{\bb}{msbm10 scaled 1100}

\newcommand{\nv}{{\bf n}}

\newcommand{\vv}{{\bf v}}
\newcommand{\xv}{{\bf x}}
\newcommand{\yv}{{\bf y}}

\newcommand{\Xm}{{\bf X}}
\newcommand{\Ym}{{\bf Y}}

\newcommand{\Nc}{{\cal N}}
\newcommand{\Oc}{{\cal O}}

\newcommand{\thetav}{\hbox{\boldmath$\theta$}}

\newcommand{\Phim}{\hbox{\boldmath$\Phi$}}

\newcommand{\Psim}{\hbox{\boldmath$\Psi$}}

\newcommand{\Thetam}{\hbox{\boldmath$\Theta$}}

\renewcommand{\arg}{{\hbox{arg}}}

\usepackage{hyperref}

\usepackage{rotating}

\begin{document}

\title{Toward Performance Optimization in IoT-based Next-Gen Wireless Sensor Networks\thanks{This research work was funded in part by the Higher Education Commission of Pakistan under the research grant number 288.67/TG/R$\&$D/HEC/2018/25181.}}
\titlerunning{Toward Performance Optimization in IoT-based Next-Gen WSNs}
%
\author{Muzammil Behzad\inst{1,2} \and
Manal Abdullah\inst{3}\and
Muhammad Talal Hassan\inst{2} \and \\
Yao Ge\inst{4} \and
Mahmood Ashraf Khan\inst{2}}

%

%
\authorrunning{\href{http://muzammilbehzad.com/}{M. Behzad et al.}}
%
\institute{Pukyong National University (PKNU), Busan 48513, South Korea \and
COMSATS University Islamabad (CUI), Islamabad 44000, Pakistan \and
King Abdulaziz University (KAU), Jeddah 21589, Kingdom of Saudi Arabia \and
The Chinese University of Hong Kong (CUHK), Shatin 999077, Hong Kong \\
Email: mbehzad@pukyong.ac.kr, maaabdullah@kau.edu.sa, talal@vcomsats.edu.pk, yge@ee.cuhk.edu.hk, mahmoodashraf@comsats.edu.pk
}

\maketitle              
\setcounter{footnote}{0} 
\begin{abstract}
In this paper, we propose a novel framework for performance optimization in Internet of Things (IoT)-based next-generation wireless sensor networks. In particular, a computationally-convenient system is presented to combat two major research problems in sensor networks. First is the conventionally-tackled resource optimization problem which triggers the drainage of battery at a faster rate within a network. Such drainage promotes inefficient resource usage thereby causing sudden death of the network. The second main bottleneck for such networks is that of data degradation. This is because the nodes in such networks communicate via a wireless channel, where the inevitable presence of noise corrupts the data making it unsuitable for practical applications. Therefore, we present a layer-adaptive method via 3-tier communication mechanism to ensure the efficient use of resources. This is supported with a mathematical coverage model that deals with the formation of coverage holes. We also present a transform-domain based robust algorithm to effectively remove the unwanted components from the data. Our proposed framework offers a handy algorithm that enjoys desirable complexity for real-time applications as shown by the extensive simulation results.

\end{abstract}
\section{Introduction}
Recent technological-accelerations for surging advancements regarding industrial applications in Internet-of-Things (IoT) based wireless communication have significantly aided major scientific and research platforms, where the main focus is to propose exceptionally elegant and convenient systems in terms of computational cost, design and practical execution. With these tremendous efforts available at hand, the consumer electronics industry have been made confident to manufacture wireless devices with economical value, tiny structure, and the capability of effectively utilizing the in-hand battery resources. Toward this end, sensors-based wireless networks have earned significant importance in the unlimited development of information and communication technologies \cite{7110295}. However, since many of these devices are restricted by the resources available to them at hand, the communication overhead and power consumption are, hence, critical areas of research for analysis, manufacturing and development of such wireless networks in order to achieve efficient management in IoT.

Wireless sensor networks (WSNs) are made up of small, portable and energy-restricted sensor nodes deployed in an observation venue. These nodes carry the baggage of transmitting vital information using wireless radio links. Such information can have the form of multi-dimensional signals, and are of critical importance in many world-wide applications. The development of such networks demands extensive planning strategy along with superior tactical approaches for~its working capabilities. This effective development motivates~the existence of many real-time application scenarios such as environmental control \cite{7887698}, under-water networks \cite{umar2014enhancing}, battlefield surveillance \cite{7528266}, medical and health-care systems \cite{6883890, 10.1007/978-3-662-48145-5_9}, and many more \cite{behzad2014design, sandhu2014mobility, sandhu2014reec, behzad2018m, behzad2018technology}.

\subsection{Underlying Structure of WSNs}
As a function of the underlying transmission mode, WSNs can be categorized by two types of communication mechanisms: 1) direct or single-hop, and 2) multi-hop, as shown in Fig. \ref{fig:fig0}. In the former method, the nodes in a network transmit their data directly to the base station (BS), also known as the sink. This in turn drains-up the battery life of the nodes, hence, resulting in an early death of the network. Therefore, this type of communication is not recommended for efficient and practical approaches. On the other hand, the later approach suggests a much more promising deal. In this method, the nodes are not needed to communicate with the BS directly, and are, therefore, allowed to send their data to BS in multiple steps. This ultimately lessens the burden on each node, and allows the network to remain stable for a longer period of time. 
\begin{figure}[b!]
	\centering
	\includegraphics[width=0.9\linewidth]{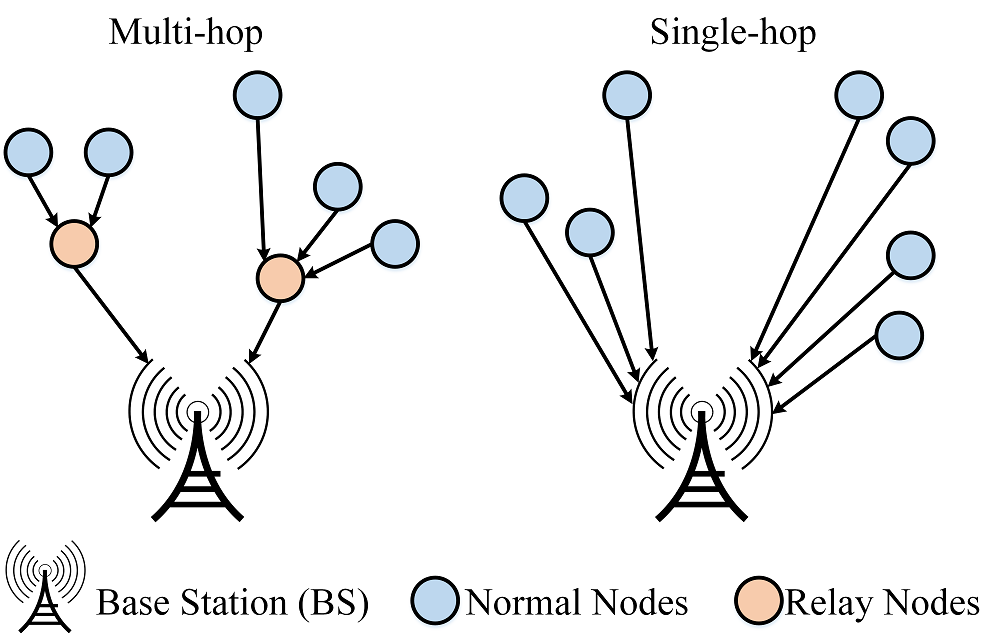}
	\caption{Multi-hop vs. Single-hop Communication Mechanism}
	\label{fig:fig0}
\end{figure}

Similarly, another distribution of WSNs is based on the type of the response that the nodes usually exhibit. Specifically, WSNs can be designed as either proactive or reactive. In a proactive mode, the nodes keep their transmitters continuously active, and periodically transmit the data independent of any parameters. Consequently, such power-hunger transmissions result in an inefficient energy utilization. On the contrary, in the reactive mode, the nodes respond only to the events that, for example, exceeds a certain threshold or when a specific event has been triggered. Since the nodes only respond to drastic changes and keep their transmitters turned-off otherwise, this yield a practically convenient system with an elongated network lifespan.\vspace{0.5cm}

\subsection{Research Developments in WSNs}
The first step in establishing a WSN is the initialization and distribution of sensor nodes around the observation field. Many researchers have advocated normal distribution of the statically deployed nodes as the optimal distribution (e.g., see \cite{7368292}). The deployment of nodes in the network field area is then followed by transmission of the required data. Since these nodes are left unattended, with limited resources at hand, efficient utilization of the available resources becomes a key operation factor to form a vigorous and standalone network.

To tackle this, many protocols recommended clustering of the network area, a pioneer contribution by W. B. Heinzelman \cite{heinzelman2000application}. Fundamentally, clustering divides the field into multiple smaller observation versions thereby making resource management a comparatively convenient task \cite{7016049, saleem2014iddr, 7920983, behzad2017distributed}. However, this requires free and fair election of cluster heads (CHs) in each cluster. These CHs are responsible for data fusion, i.e., they receive data from their respective clusters' normal nodes, and transmit them to the BS.

Traditionally proposed protocols for WSNs focus mainly on performance improvements via effective selection criterion for CHs, the choice of single-hop or multi-hop communication between nodes, and whether the clustering scheme should be static or dynamic. Even though an optimal combination of the above factors yield interesting results, however, impressive results can be achieved by looking into more exciting parts of the problem. Furthermore, another much concerned and barely discussed side of the problem is the degradation due to environment. A common form of such inevitable degradation affecting the data sent over radio links is that of additive white Gaussian noise (AWGN). Several designed protocols discard this important issue and just focus on minimizing the energy consumption by assuming that the data received by nodes have not experienced any noise addition due to the environment.

\subsection{Notations}
\label{notations}
In rest of the paper, we use the following notations. We represent all the vectors used in our work with small case and bold face letters (e.g., $\yv$), all the scalars with small case normal font letters (e.g., $y$). We reserve upper case and bold face letters (e.g., $\Ym$) for matrices. For sets, we use calligraphic notation (e.g., $\Nc$). We use $\yv_i$, $y(j)$ and $\Nc_k$ to denote $i$th column of matrix $\Ym$, $j$th element of vector $\yv$ and a subset of $\Nc$, respectively.

\subsection{Recoveries via Sparse Representations}
Contrary to the traditional Nyquist-Shannon sampling theorem, where one must sample at least double the signal bandwidth, compressive sensing (CS) has emerged out as a new framework for data acquisition and sensor design in an extremely competent way. The basic idea is that if the data signal is sparse in a known basis, a perfect recovery of the signal can be achieved leading to a significant reduction in the number of measurements that need to be stored.

According to CS, the following model can be used to recover an unknown vector $\vv$ from an under-determined system:
\begin{equation}
\xv = \Phim \vv = \Phim \Psim \thetav = \Thetam \thetav,
\end{equation}
where $\xv \in {\mathbb{C}^M}$, $\vv = \Psim \thetav \in {\mathbb{C}^N}$ are observed signal and unknown vector, $\thetav\in {\mathbb{C}^N}$ is unknown sparse signal which, for example, a node will collect, representing projection coefficients of $\vv$ on $\Psim$, $\Thetam = \Phim \Psim$ is an $M \times N$ reconstruction matrix ($M < N$). The measurement matrix $\Thetam$ is designed such that the dominant information of $\thetav$ can be captured into~$\xv$.

Reconstruction algorithms in CS exploit the fact that many signals are genetically sparse, therefore they proceed to minimize ${\ell_0},{\ell_1}$ or ${\ell_2}$-norm over the solution space. Among them, $\ell_1$-norm is the most accepted approach due to its tendency to successfully recover the sparse estimate $\hat {\thetav}$ of $\thetav$ as follows:
\begin{align}\label{opt.CS}
\hat {\thetav} = \mathop {\arg \min }\limits_{\thetav} {\left\| {\thetav} \right\|_{{\ell_1}}}, \quad \text{subject to} \quad {\xv} \approx \Phim \Psim \thetav.
\end{align}
\begin{figure}[t]
	\centering
	\includegraphics[width=0.8\linewidth]{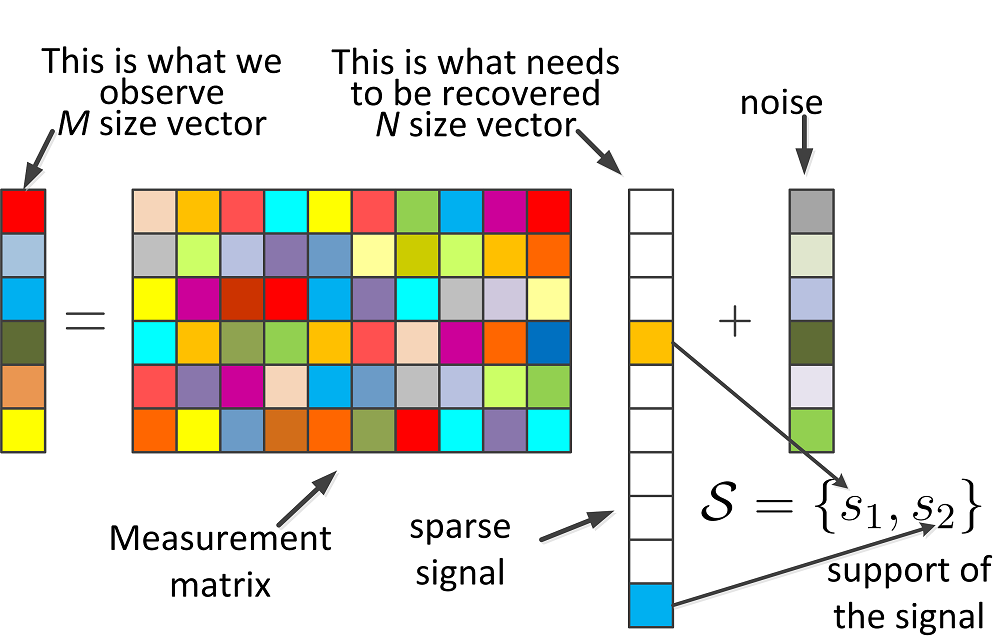}
	\caption{Sparse Model}
	\label{fig:fig1}
\end{figure}

In this regard, however, the inevitable presence of noise in wireless channels is always a challenging task to combat. Consequently, the system is modeled as
\begin{equation}
\yv = \xv + \nv = \Thetam \thetav + \nv,
\end{equation}
where $\yv \in {\mathbb{C}^M}$ is the noisy version of the clean signal $\xv \in {\mathbb{C}^M}$ which is corrupted by the noise vector $\nv \in {\mathbb{C}^M}$ with i.i.d. zero mean Gaussian entries having variance $\sigma_{\nv}^2$, i.e., \boldmath{$\nv(.)$} $\sim \pN(0, \sigma_{\nv}^2\textbf{1})$. A depiction of this model is shown in Fig.~\ref{fig:fig1}.

\subsection{Contribution}
A part of this work has already been published in \cite{behzad2018AINA}, and this is the extended version of our previous work. In \cite{behzad2018AINA}, we introduced a novel framework to tackle two major concerns in WSNs: 1) performance optimization via efficient energy utilization, and 2) combating the unavoidable presence of Gaussian noise, added as a result of multiple communications among the nodes via wireless channel. We proposed a fast and low-cost sparse representations based collaborative system enriched with layer-adaptive 3-tier communication mechanism. This is supported by an effective CHs election method and mathematically convenient coverage model guaranteeing minimization of energy and coverage holes. A computationally desired implementation of our framework is an added benefit that makes it a preferable choice for real-world applications. 

To tackle AWGN, the data is transmitted in spatial-domain form and its sparse estimates are later computed at the receiver end. For a much better denoising, we let the nodes situated at a single-hop to mutually negotiate with each other for better collaboration. The data denoising is further refined by a specially designed averaging filter. 

In this paper, we extend the concept of image denoising in wireless sensor networks. Specifically, we propose to use region growing based efficient denoising mechanism where we divide the entire image into various sub-regions based on their intensities, and apply smoothening filter. Motivated by this, we also extend our current framework for color images where we are especially interested in exploiting inter-channel correlation of each color image. This effective piece of information plays a crucial role in identifying the noisy components, and thereby helps discarding those components. Our proposed protocol lends itself the following salient~features:
\begin{itemize}
	\item[--] The implementation of a mathematically efficient coverage model along with an adaptive CHs election method help avoiding coverage holes to a greater extent.
	\item[--] Our proposed layer-adaptive 3-tier communication system greatly reduces energy holes.
	\item[--] To compute denoised data signals, we compute support-independent sparse estimates which relieves us from finding distribution of the sparse representations first, hence, giving it a support-agnostic nature.
	\item[--] Prior collaboration enjoyed by the nodes for communication yields an effectively significant energy minimization.
	\item[--] The use of a fast sparse recovery technique allows a desired computational complexity of our algorithm.
	\item[--] The use of inter-channel correlation among red, green and blue channels of color images not only makes it a suitable choice for denoising of color images, but it also provides a convenient solution for fast and practical image denoising applications.
\end{itemize}

Rest of the paper is structured as follows: Section \ref{Related_Work} presents an overview of the related work done in this area. We describe our proposed framework in Section \ref{Proposed_Framework}, while the results from various simulations are presented and discussed in Section \ref{Results_Discussions}. Finally, Section \ref{Conclusion} concludes the paper.

\section{Related Work}
\label{Related_Work}
Presently, researchers are fundamentally concentrating on the technologically-enriched tools for performance optimization of network structure, as a result of which the lifetime of WSNs is possible to increase. This possibility provides a roadway for scientists working in this research domain to propose low-cost, energy-efficient and optimized algorithms~\cite{8109448, behzad2017layer}. 

This includes a trend-setter work by W. B. Heinzelman, et al., proposing a multi-hop energy efficient communication protocol for WSNs, namely LEACH \cite{926982}. The objective was to reduce energy dissipation by introducing randomly elected CHs resulting in, however, unbalanced CHs distribution. Nevertheless, the partition of network area into different regions via clustering yielded a significant increase in system lifetime. As a principal competitor, a new reactive protocol, named as TEEN, was proposed by authors of \cite{925197} for event-driven applications. This protocol, even though constrained to temperature based scenarios only, proposed threshold aware transmissions thereby outperforming LEACH in terms of network lifespan.

In comparison with the aforementioned homogeneous WSNs, the authors of \cite{smaragdakis2004sep} and \cite{QING20062230} proposed SEP and DEEC, respectively, introducing heterogeneous versions of the WSNs by allowing a specific set of nodes, defined as advanced nodes, to carry higher initial energy than other normal nodes. SEP used energy based weighted election to appoint CHs in a two-level heterogeneous network ultimately improving network stability. As a stronger contestant to SEP, DEEC deployed multi-level heterogeneity and improvised CHs election measure to attain extended lifespan of the network than SEP.

A. Khadivi, et al., proposed a fault tolerant power aware protocol with static clustering (FTPASC) for WSNs in \cite{1696382}. The network was partitioned into static clusters, and energy load was distributed evenly over high-power nodes, resulting in minimization of power consumption, and increased network lifetime. Another static clustering based sparsity-aware energy efficient clustering (SEEC) protocol is proposed in \cite{SWQ28921546.OW12N}. This protocol used sparsity and density search algorithms to classify sparse and dense regions. A mobile sink is then exploited, specifically in sparse areas, to enhance network lifetime.

As opposed to static clustering, authors of \cite{5136647} presented centralized dynamic clustering (CDC) environment for WSNs. In this protocol, the clusters and number of nodes associated with each cluster remains fixed, and a new CH is chosen in each round of communication between clusters and BS. CDC showed better results than LEACH in terms of communication overhead and latency. In a similar fashion, G.~S.~Tomar, et al., proposed an adaptive dynamic clustering protocol for WSNs in \cite{4809826}, which creates a dynamic system that can change topology architecture as per traffic patterns. Mutual negotiation scheme is used between nodes of different energy levels to form energy efficient clusters. Periodic selection of CHs is done based on different characteristics of nodes. Another work proposed to use the cooperative and dynamic clustering to achieve energy efficiency \cite{7822956}. This framework ensured even distribution of energy, and optimization of number of nodes used for event reporting thereby showing promising results.

D. Jia, et al., tackled the problem of unreasonable CHs selection in clustering algorithms \cite{7365420}. The authors considered dynamic CH selection methods as the best remedy to avoid overlapping coverage regions. Their experimental results showed increased network lifetime than LEACH and DEEC. Another energy efficient cluster based routing protocol, termed as density controlled divide-and-rule (DDR), is proposed in \cite{S2RT4387429.OW12N}. The authors tried to take care of the coverage and energy holes problem in clustering scenarios. They presented density controlled uniform distribution of nodes and optimum selection of CHs in each round to solve this issue. Similarly, a cluster based energy efficient routing protocol (CBER) is proposed in \cite{6531736}. This protocol elects the CHs on the basis of optimal CH distance and nodes' residual energy. CBER reported to outperform LEACH in terms of energy consumption of the network, and its lifetime.

\section{Proposed Framework Design}
\label{Proposed_Framework}
In this section, we provide the readers with detailed understanding of our proposed routing protocol. Here, we broadly discuss the widely accepted radio model for communication among nodes. This is then followed by a comprehensive explanation of our adopted network configuration and its operation details for energy efficiency and denoising of the~data.

\subsection{Wireless Communication Model}
For transmission and reception of required data among sensor nodes via wireless medium, we assume the simple and most commonly used first order radio communication model as given in Fig. \ref{fig:fig2}. In this figure, we present the energy consumed by a node while transmitting and receiving data. We show that a packet of data traveling over radio waves has to combat against degrading factors such as noise, multi-path fading, etc. Thus, we also take into account the $d^2$ losses that almost all chunks of data has to face. This is mathematically explained in terms of the following expressions:
\begin{equation}
E_{Tx}(k,d) =
\begin{cases}
k\times (E_{elec}+ \epsilon_{fs}\times d^2), & d < d_o\\
k\times (E_{elec}+ \epsilon_{mp}\times d^4), & d \geq d_o
\end{cases}
\end{equation}
\begin{equation}
\hspace{-5cm} E_{Rx}(k) = E_{elec}\times k,
\end{equation}
where $d_o$ is a reference distance, $k$ is the number of bits in packet, $d$ is the transmission distance which varies every time for each node, $E_{elec}$ is the energy used for data processing, $\epsilon_{fs}$ and $\epsilon_{mp}$ are channel dependent loss factors\footnote{This is worth noting that over larger distances, such loss factors demand a higher amount of energy yielding sudden death of the network. This is often missed by traditional protocols assuming lossless channel. Therefore, avoiding these power-hungry transmissions significantly optimize resources.}, $E_{Tx}$ is the energy used by a node for transmission, and $E_{Rx}$ is the energy used by a node for data reception. As shown, the $d^r$ losses may change from $\left. d^r \right \vert_{r=2}$ to $\left. d^r \right \vert_{r=4}$ forcing a higher value of $E_{Tx}$. A similar increase is then observed in the $E_{Rx}$ values to process a highly corrupted data when assuming a noisy environment, as in our case. The generally used energy dissipation values  for a radio channel are presented in Table~\ref{tab1}.
\begin{figure}[t]
	\centering
	\includegraphics[width=1\linewidth]{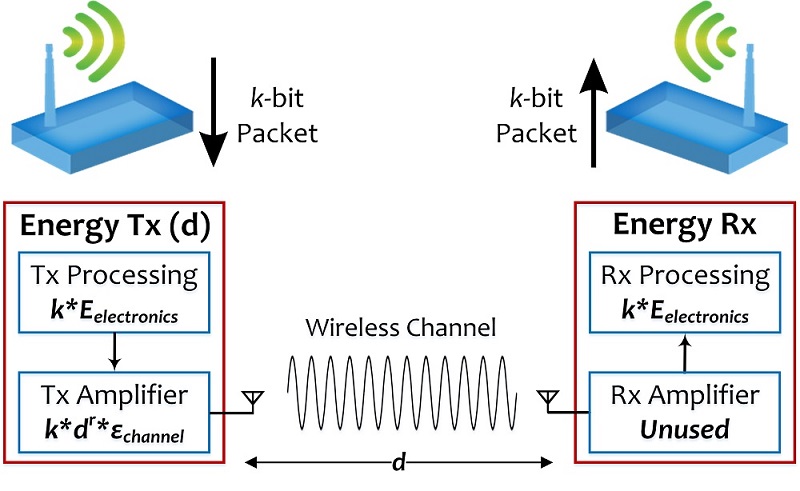}
	\caption{Radio Model}
	\label{fig:fig2}
\end{figure}
\begin{table}[b]
	\caption{\textbf{Energy Dissipation Measurements}}
	\centering
	\renewcommand{\arraystretch}{1}
	\begin{tabular}{|l|l|}
		\hline \hline
		\textbf{Dissipation Source} & \textbf{Amount Absorbed} \\ [0.5ex]
		\hline \hline
		$E_{elec}$ of Rx and Tx & 50 nJ/bit \\
		\hline
		Aggregation Energy  & 5 nJ/bit/signal\\
		\hline
		Tx Amplifier $\epsilon_{fs} \text{for} \left. d^r \right \vert_{r=2}$ & 10 pJ/bit/4m$^2$ \\
		\hline
		Tx Amplifier $\epsilon_{mp} \text{for} \left. d^r \right \vert_{r=4}$ & 0.0013 pJ/bit/m$^4$ \\
		\hline\hline
	\end{tabular}
	\label{tab1}
\end{table}

\subsection{Network Configuration}
\begin{figure*}[t]
	\centering
	\includegraphics[width=1\linewidth]{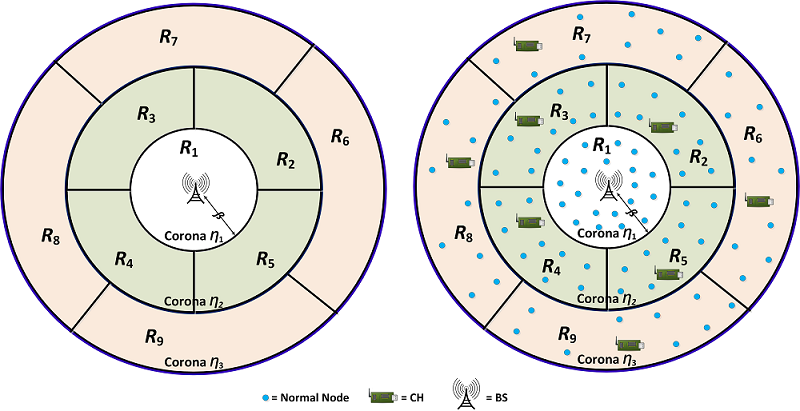}
	\caption{Network Model: (a) Network Configuration, (b) Nodes Deployment}
	\label{fig:fig3}
\end{figure*}
For the configuration model, we use a network consisting of $L$ number of nodes deployed randomly. Unlike the traditional models, we adopt a spherically-oriented field, and propose to use an optimized version of area division via adaptive clustering. For a much better understanding, see the network model shown in Fig.~\ref{fig:fig3}. Here, for the sake of simplicity and understanding, we use the field area $A = \pi150^2\text{ m}^2$, i.e., the diameter $D = 300$, with a total of $L = 100$ deployed nodes. To avoid formation of energy holes, and thus the death of network, we place the BS in the center of the network at coordinates $<i,j> = (0,0)$. This is followed by the clustering of network field into various coronas which are then further classified into different sensing-based reporting regions. The prior computation of number of coronas, represented by $\eta$, is a function of field area $A$, which itself is depending on $D$, and the number of nodes $L$. As a sound approximate, we propose $\eta = D/L$. Hence, in our case we use $\eta$ = 3 coronas, denoted accordingly by $\eta_1, \eta_2 \text{ and } \eta_3$.

Once the $\eta$ number of coronas are formed, the next step is to divide each corona into various sensing regions as shown in the figure. However, for a much better network performance, the distribution is such that each sensing region in the upper level corona $\eta_\alpha$ surrounds two sensing regions in lower level corona $\eta_{\alpha-1}$. This is shown in Fig. \ref{fig:fig3}(a), where for example region $R_7 \text{ in } \eta_3$ covers both $R_2 \text{ and } R_3 \text{ in } \eta_2$, hence, avoiding coverage holes by satisfying the following expression for a general network configuration\footnote{Here, we presented calculations for $A = \pi150^2\text{ m}^2 \text{ and } L = 100$ merely for the ease of understanding. However, for any other small or large scale network configuration, the computations can be done in a similar fashion using the proposed expressions.}:
\begin{equation}
\label{eq:area}
A = \pi(D/2)^2 = \sum_{\alpha=1}^{\eta} A_{\eta_{\alpha}} = \sum_{\alpha} A_{R_{\alpha}}, \hspace{0.365cm}\alpha \in \mathbb{Z}^+,
\end{equation}
where $A_{\eta_{\alpha}}$ and $A_{R_{\alpha}}$ represented area of each corona and sensing region, respectively. This is worth noting that we do not divide the corona surrounding BS further, $\eta_1$ in our case, to avoid unneeded and poor use of available resources. Thus, we can safely write:
\begin{equation}
A_{\eta_{1}} = A_{R_{1}} = \pi\beta^2, \text{ where }\beta \in \mathbb{R}^+.
\end{equation}

\subsection{Nodes Deployment and Layer-Controlled CHs Nomination}
As soon as the network is clustered out into various coronas and sensing regions, the next step is to distribute the nodes randomly over these regions. To optimize resources, a sensible decision is to deploy an equal percentage of nodes over different regions to ensure minimization of coverage holes, and elongation of network lifetime. Therefore, in this scenario, we propose to deploy 20$\%$ of the nodes in region $R_1$ and the rest 80$\%$ of the nodes to be distributed evenly over $R_{2,3,...,8,9}$ regions as shown in Fig. \ref{fig:fig3}(b). This nodes' deployment always depend upon the network field area and number of nodes. Hence, for any other network configuration, an adjusted percentage can be calculated to optimize communication among nodes, and to avoid energy and coverage holes.

Following the deployment of nodes and prior network initialization, the election of CHs is carried out in all $R_{2,3,...,8,9}$ regions. Since the use of CHs in clustering techniques plays an important role to improve network lifespan, effective criterion for CHs election is equally necessary for further improving performance of the network. The most commonly used measures for electing CHs are residual energy and distance from BS. We propose a blend of both to increase the life of each node. Furthermore, we introduce a layering-based election of CHs. This means that the election will take place in lower level coronas $\eta_{\alpha}$ first, and will then move to high level coronas $\eta_{\alpha+1}$ for higher level CHs. The reason to adopt this is the effectiveness noted in CHs election. Thus, in each round, all the nodes are assessed based on their residual energies and top 5$\%$ of the nodes having highest residual energies in their respective regions are shortlisted. These shortlisted nodes then contest against each other where the node with smallest distance to the CHs of both associated regions in lower level corona is elected as CH. The nodes in $\eta_2$ are evaluated in a similar fashion based on the residual energy but having minimum distance with the center of their respective region.

\subsection{Layer-Adaptive 3-Tier Communication Mechanism}
For transfer of data among various nodes, we propose a layer-adaptive 3-tier architecture. Our communication mechanism is enriched with distance-optimized transmissions to avoid wastage of energy. The nodes use a multi-hop scheme instead of directly transmitting the data of interest to BS. In tier-1 phase, all the normal nodes gather data, and send it to the nearest CH. This CH may not necessarily be the same region CH. Here, we allow nodes in $\eta_\alpha$ to transmit the data to CHs of even $\eta_{\alpha-1}$. This is the reason why we distributed the sensing regions in such a way that each region in upper level corona is bordered with two regions in lower level corona. However, the nodes of a region in $\eta_\alpha$ cannot transmit to another region on the same corona, i.e., it must either send data to its own CH in $\eta_\alpha$, or any other nearest CH in the two bordered regions on lower level corona $\eta_{\alpha-1}$ as explained in Fig. \ref{fig:figcommmodel}.
\begin{figure*}[t]
	\centering
	\includegraphics[width=1\linewidth]{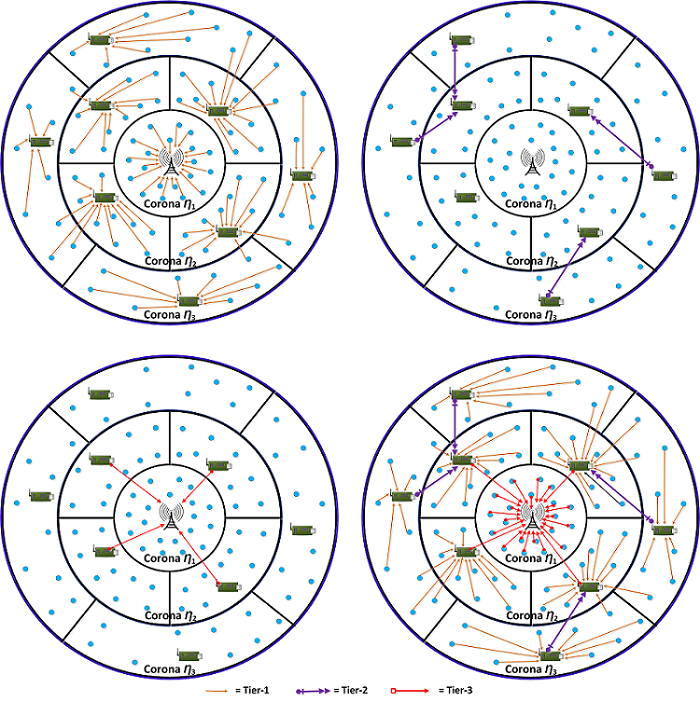}
	\caption{3-Tier Communication Architecture}
	\label{fig:figcommmodel}
\end{figure*}

In the next tier-2 phase of communication, the CHs of $\eta_\alpha$ aggregate their data and then send it to the CHs of $\eta_{\alpha-1}$. Note that even though the CH of $R_3$ is receiving data from CHs of both $R_7 \text{ and } R_8$, this is blessing in disguise. This is because, as shown in the figure, the CHs of both $R_7 \text{ and } R_8$ have not received data from all the nodes in its region, since some nodes find another nearest CH, so these CHs are aggregating and then forwarding a comparatively smaller amount of load thereby not overburdening themselves. Also, the CHs change in each round based on the election criterion, so it ultimately saves energy. Finally in tier-3 phase, all the CHs in lower level coronas send their data to the BS, hence, completing the data transmission process.

\subsection{Coverage Model}
For reduction in coverage holes, we express the coverage scenario of nodes by a mathematical model. All the deployed sensor nodes are represented in set notation as $\kappa = \{\mu_1, \mu_2, \mu_3, ..., \mu_L\}$. The coverage model of one alive node $\mu_\alpha$ belonging to the set $\kappa$ can be expressed as a sphere centered at $<i_\alpha,j_\alpha>$ with radius $h_\alpha$. We let a random variable $\aleph_\alpha$ define an event when a data pixel $<a,b>$ is within the coverage range of any node $\mu_\alpha$. As a result, the equivalent of likelihood of the event $\aleph_\alpha$ to happen, as denoted by $P\{\aleph_\alpha\}$, is represented as $P_{cov}\{a,b,\mu_\alpha\}$. A decomposed version of the above is given as follows:
\begin{multline}
P\{\aleph_\alpha\} = P_{cov}\{a,b,\mu_\alpha\}=\\
\begin{cases}
1, & (a-i_\alpha)^2 + (b-j_\alpha)^2 \le h_\alpha^2\\
0, & \text{otherwise}
\end{cases}
\end{multline}
where the equation translates that a data pixel $<a,b>$ is surrounded by the coverage range of any random node $\mu_\alpha$ if the distance between them is smaller than the threshold radius $h_\alpha$. However, since the event $\aleph_\alpha$ is stochastically independent from others, this means $\left. h_\alpha \text{ and } h_\gamma \text{ are not related} \implies \alpha, \gamma \in [1,L] \right.$ and $\left. \alpha \ne\gamma\right.$. This gives us the following conclusive equations:
\begin{equation}
P\{\overline{\aleph_\alpha}\} = 1-P\{\aleph_\alpha\} = 1-P_{cov}\{a,b,\mu_\alpha\},
\end{equation}
\begin{multline}
P\{\aleph_\alpha \cup \aleph_\gamma \} = \\ 1-P\{\overline{\aleph_\alpha}\cap\overline{\aleph_\gamma}\}= 1-P\{\overline{\aleph_\alpha}\}.P\{\overline{\aleph_\gamma}\},
\end{multline}
where $P\{\overline{\aleph_\alpha}\}$ denotes the statistical complement of $P\{\aleph_\alpha\}$ which means that $\mu_\alpha$ failed to assist data pixel $<a,b>$. Importantly, this data pixel is given coverage if any of the nodes in the set is covering it otherwise a coverage hole would form. Hence, the following expressions denote the probability such that data pixels would be within the coverage range of at least one of the nodes in the set to minimize coverage holes:
\begin{multline}
P_{cov}\{a,b,\kappa\} = P\{\bigcup\limits_{\alpha=1}^{L}\aleph_\alpha\} = 1 - P\{\bigcap\limits_{\alpha=1}^{L}\overline{\aleph_\alpha}\},\\
=  1 - \prod_{\alpha=1}^{L}(1-P_{cov}\{a,b,\mu_\alpha\}).
\end{multline}
For further facilitation, we present the coverage rate as fraction of area under coverage, denoted by $Q$, and the overall area of the observation field as follows:
\begin{equation}
P_{cov}\{\kappa\}=\sum_{\alpha=1}^{L}\sum_{\gamma=1}^{L}\frac{P_{cov}\{a,b,Q\}}{A}
\end{equation}

\subsection{Data Denoising}
After taking care of the energy efficiency, second major problem is retrieving the original data back. This is because the received data is generally degraded by AWGN so it is of no use unless denoised. For this purpose, we propose denoising of the data samples via Bayesian analysis based sparse recovery techniques. To do so, we take into account the data correlation of various adjacent nodes, and use this as an important piece of information for collaboration among nodes. 

We use three stages for CS based sparse recovery technique to denoise the data. In doing so, received data is converted to sparse domain first (e.g., wavelet transform for images data). This is followed by computing similar and correlated data by adjacent nodes, giving them weights based on the similarity extent. Using equivalent sparse representations of data samples, probability of active taps is computed giving us the location of undesired corrupted support locations \cite{7952375}. With the help of correlation information, an averaging based collaborative step is performed to remove the unwanted noisy components as shown via flowchart in Fig \ref{fig:fig4}. Here, we denote the initially denoised image by $\bar{\Xm}_{d}$. 

Finally, we apply a specially developed averaging filter to further smooth out the data as discussed in the later sections. This filter fundamentally works on finding similar data samples, and then averaging those samples to provide a clean estimate of the data. Using a CS based pre-determined dictionary, a reverse transform is applied to give back the denoised data in spatial-domain representation as $\hat{\xv} = \Thetam \hat{\thetav}$.

\subsubsection{Similarity via Distance vs. Correlation:}
\label{Patches_Grouping}
For the similar and correlated data, we first compute samples from the data, for example, overlapping patches or blocks in images. Once the overlapping patches are formed, the next step is to find a certain number of similar patches, for each patch, that would be used during collaboration. The grouping of patches in such a way using a similarity measure has led to a number significant improvements in a wide range of application like signal/image/bio-medical processing, computer vision, machine intelligence, etc. (e.g., see \cite{krause2015shape, 6853394, 7051524, 7351075, 7457822, bahrami2016reconstruction, 7952375, behzad2018image}).	

A number of techniques for similarity based grouping of patches have been proposed in the literature. Some of those include self-organizing maps \cite{VanHulle2012}, vector quantization \cite{1056457}, fuzzy clustering \cite{hoppner1999fuzzy} and a review on these \cite{jain1999data}. The recently developed denoising algorithms use a distance based measure where similarity between different signals are realized in terms of the inverse of the point-wise distance between them. Therefore, a smaller distance between the signals would imply a higher similarity and vice versa. The generally used distance based similarity measure is the Euclidean distance as used by the state-of-the-art denoising image algorithms like NL-means \cite{1467423}, BM3D \cite{4271520}, etc.

However, despite being an effective way of finding similarity, Euclidean distance based similar-intensity grouping has a limitation; it limits the search for number of similar patches. For instance, even though natural images have some similarity in their structure, the number of similar patches vary. Consequently, in an image having a smaller number of similar patches, the collaboration is not that effective thereby disturbing the performance of denoising, especially in case of high noise. This creates a bottleneck specifically for lower resolution images where finding similar-intensity patches becomes a difficult task.

To tackle this case and have a similarity measure that can be used globally even in lower resolution images or images having a smaller number of similar-intensity patches, novel methods are being proposed to find better ways of collaboration by using efficient grouping of similar patches. For example, the authors in \cite{7005524} search the similar patches by using not only a patch itself but the noise too where they propose the concept of noise similarity, while the authors in \cite{7444121} propose sequence-to-sequence similarity (SSS) which is an essential way of preserving the edge information.

In our case, we take care of the aforementioned problem by introducing intensity-invariant grouping. The idea is to stack all the patches that have a similar inherent structure without relying on the intensity values as shown in the Stage 01 of Fig. \ref{fig:fig4}. The correlation coefficient serves as the best tool to be utilized for the said purpose. For two random signals $\yv_k$ and $\yv_i$, the correlation coefficient is given as,
\begin{align}\label{eq:corr_coeff}
r(\yv_\alpha,\yv_\gamma) = \frac{ cov(\yv_\alpha,\yv_\gamma) }{\sigma_{\yv_\alpha}\sigma_{\yv_\gamma}},
\end{align}
where $-1 \le r(\yv_\alpha,\yv_\gamma) \le 1$. A value close to $1$ or $-1$ means larger positive and negative correlation, respectively, while a value close to $0$ means smaller correlation.

\subsubsection{Selection of the Measurement Matrix:}
Since we will be denoising the image patches by using the sparse estimates from collaborative filtering in the transform domain, the use of an appropriate measurement matrix or dictionary also serves as a key step. Generally, the dictionary mainly consist of basis vectors through which any random patch can be represented as a linear combination of the basis elements. In our case, we are representing any patch using the obtained sparse vector and the dictionary as already shown in Fig. \ref{fig:fig1}.

\subsubsection{Decorrelation of the Measurement Matrix:}
Each patch can be written as linear combination of basis elements from the dictionary. The columns of this dictionary are derived from wavelet basis and are normalized to have unit norms. Prior finding support sets of $\widehat{\thetav}_\alpha$ via sparse estimation of patches, we will reduce the correlation between dictionary columns for a robust computational and performance ability. Consequently, we remove weak supports by rejecting highly correlated columns as the information they encode could easily be encoded by other columns which correlate with them. We denote it by the decorrelator operator as follows
\begin{align}\label{eq:dictionary_decorrelation}
	\Thetam = \Gamma_{\tau}(\Thetam')
\end{align}
where $\Gamma_{\tau}(.)$ is the de-correlation operator that removes all the columns of $\Thetam'$ with correlation greater than $\tau$.

\subsubsection{Gaussianity Property:}
This should be noted that the Gaussianity property of the noisy data received and then aggregated at the receiver (e.g., CHs or BS) should remain intact. This is because, even though our proposed Bayesian analysis based denoising algorithm is agnostic to support distribution of the sparse coefficients, it does need the data samples to be corrupted by Gaussian noise collectively. A concise version of this is provided in the following Lemma~\ref{lemma:lemma1} to support the accuracy of our denoising algorithm.
\begin{lemma}
	The aggregated data samples received at either CHs or BS keep the Gaussianity property intact, hence, we can denoise the cumulative version of the AWGN corrupted data.
	\label{lemma:lemma1}
\end{lemma}
\begin{proof}
	To show this, we consider two independent Gaussian random data samples $P \text{ and } Q$ sent by nodes $\left.\mu_\alpha \text{ and } \mu_\gamma, \text{ both } \in \kappa\right.$. For data aggregated by CH, we let $\left.Z=\rho P+\delta Q\right.$. Without loss of generality, let $\rho \text{ and }\delta$ be positive real numbers because for $\rho<0$, $P$ would be replaced by $-P$, and then we would write $\mid \rho\mid$ instead of $\rho$. The commutative probability function can be written as:
	\begin{multline}
	F_Z(z)=P\{Z\le z\}=P\{\rho P+\delta Q\le z\}\\
	=\int\int_{\rho P+\delta Q\le z}\varphi(p)\varphi(q)
	dpdq
	\end{multline}
	where $\varphi(.)$ represents the unit Gaussian density function. However, as the integrand $(2\pi)^{-1}\exp(-(p^2+q^2)/2)$ possesses circular symmetry, the numerical property of this integral is a function of length of the origin from $\left.\rho p+\delta q=z\right.$. Consequently using coordinates rotation, we can conclude
	\begin{multline}
	F_Z(z)=\int_{p=-\infty}^{\zeta}\int_{q=-\infty}^{q=\infty}\varphi(p)\varphi(q)dpdq=\Delta(\zeta)
	\end{multline}
	where $\zeta=\frac{z}{\sqrt{\rho^2+\delta^2}},\text{ and }\Delta(.)$ shows standard Gaussian CDF. Hence, the CDF of $Z\vert_{L=2}$ is a zero-mean Gaussian random variable having total variance equal to $\left.\rho^2+\delta^2\right.$.
\end{proof}

\begin{figure}[t]
	\centering
	\includegraphics[width=0.85\linewidth]{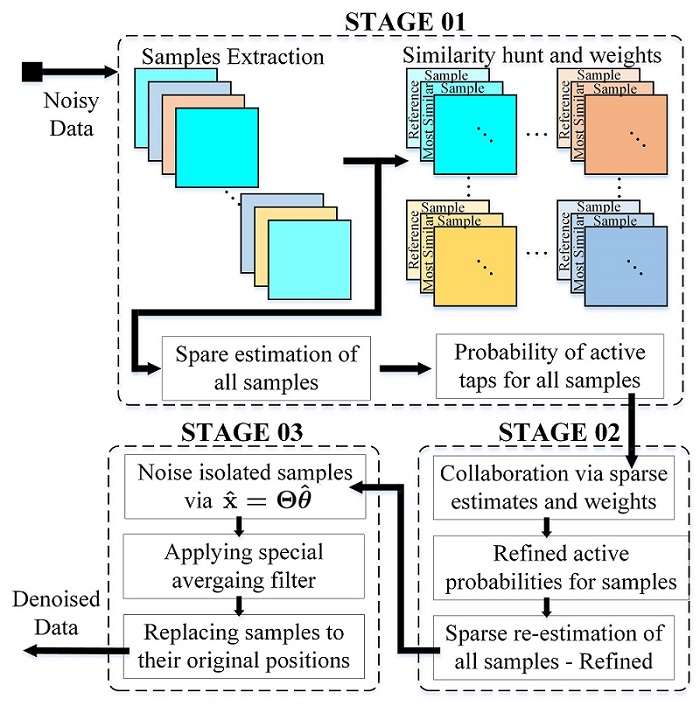}
	\caption{Flowchart of Data Denoising}
	\label{fig:fig4}
\end{figure}

\begin{figure}[t!]
	\centering
	\includegraphics[width=1\linewidth]{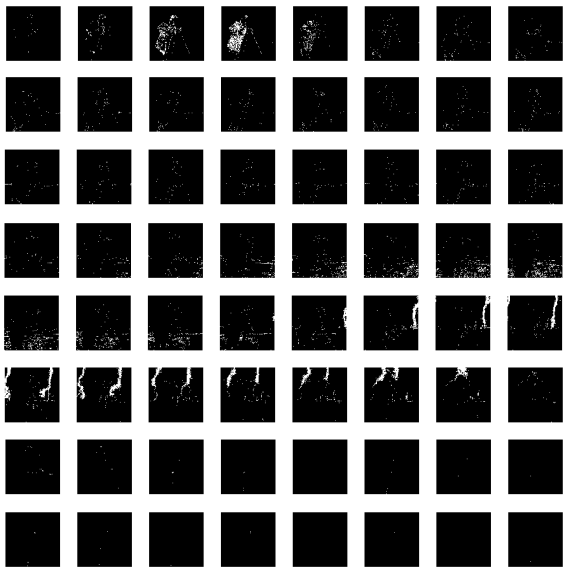}
	\caption{An example of dividing the Cameraman image into 64 different groups/bins (left to right): first row; group 1-8, second row; group 2-16, third row; group 17-24, fourth row; group 25-32, fifth row; group 33-40, sixth row; group 41-48, seventh row; group 49-56, 8th row; group 57-64}
	\label{fig:Reg_Growing}
\end{figure}
\subsection{Region Growing based Smoothening Filter}
As a final step for removing out the noisy components from the image, we perform region growing method on the output image resulted from the previous process. For this image, we store the pixels in different number of bins based on their intensity levels. For instance, we assign group 1 to the pixels that have, for example, intensity range from 0-3, group 2 to pixel intensities from 4-7, and so on. We do this for all the pixels and as a result we create different bins with pixels and their locations stored within those bin groups. We show an example of applying such intensity-leveling on the \textit{Cameraman} image in Fig. \ref{fig:Reg_Growing}. In this figure, we display all the intensity groups/bins as binary images where the white pixels correspond to the pixels of the \textit{Cameraman} image belonging to the relevant group.

For each bin, we apply the region growing algorithm to find the connected pixels within that bin. This means that the local similar intensity pixels are identified first. Afterwards, if the number of connected pixels in each bin exceed a certain threshold, then we replace those connected pixels by their mean. Similarly, we repeat this process for all the bins which ultimately provides us with the region growing based processed image that we denote by $\bar{\Xm}_{r}$. Finally, we get our final denoised image $\bar{\Xm}$ using the weighted average of the image $\bar{\Xm}_{d}$ from denoiser and the region growing processed image $\bar{\Xm}_{r}$ as follows
\begin{align}\label{eq:post_processed}
\bar{\Xm} =  \varrho_1 \bar{\Xm}_{d} + \varrho_2 \bar{\Xm}_{r},
\end{align}
where $\varrho_1$ and $\varrho_2$ are the weights which are a function of the noise variance. 

\subsection{Effective Collaboration via RGB Channels of Color Images}
As opposed to the case of grayscale single channel images, color images having three R, G and B channels that provide a more advanced way through which the patches can collaborate. Since finding similar patches using more effective approaches is the key for such collaboration, the three channels of a color images supply an important piece of information in the form of the channel correlation that can be used to identify similar patches.

To understand this, consider the three R, G and B channels of the standard Mandrill image as shown in Fig. \ref{fig:colorblock} as separate images. Since the additive white Gaussian noise is independent in all three channels of the image, we denoise the color image by denoising each channel separately. This results in formation of rectangular patches for all three channels. To denoise a patch in a specific channel of the observed color image, once the patches are extracted, similar patches are grouped together by taking into account information from both reference channel and the other two channels. 
\begin{figure}[t!]
	\centering
	\includegraphics[width=1\linewidth]{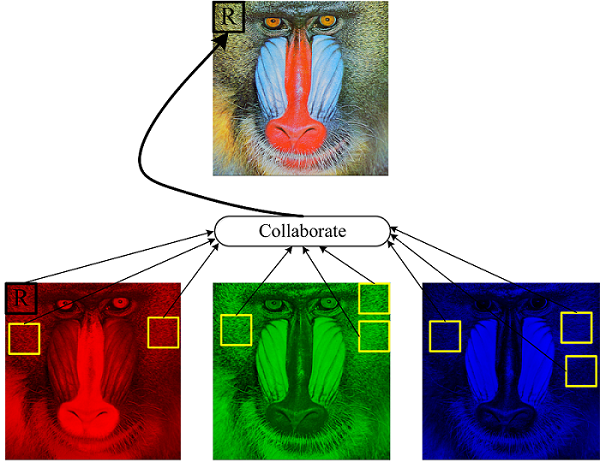}
	\caption{A depiction of collaboration among patches across all three channels}%
	\label{fig:colorblock}%
\end{figure}

For example in Fig. \ref{fig:colorblock}, to denoise the reference patch, denoted by `R', from the red channel, similar patches are grouped together from the red channel firstly. This ensures the identification of patches as similar and gives a set containing the information of similar patch numbers. Using this set from the red channel, the similar patches from other channels, for this specific patch, are also identified. Then, the reference patch in the red channel may collaborate with the patches from all channels. Since the idea is to refine the probabilities of active taps by using the sparse vectors that may share the same support, finding similar patches using all three channels can be very effective. These grouped patches for all channels can then ultimately be used to effectively estimate the sparse vectors that are in turn used to obtain denoised patches. These steps are performed for all the patches in all the three channels which ultimately provide us with a denoised color image.

\section{Computational Complexity}
\label{Computational_Complexity}
The computational complexity of our proposed framework is dominated by that of the sparse recovery algorithm that we use, which fortunately has a low computational complexity when compared to other similar existing algorithms for sparse recovery. With the dimensions of our problem at hand, the complexity for estimating one $\thetav_\alpha$ via the sparse recovery algorithm is of order $\Oc(MN^2\Upsilon)$ where $\Upsilon$ is the expected number of non-zeros that is generally a very small number.

\section{Results and Discussions}
\label{Results_Discussions}
In this section, we compare our proposed scheme with the state-of-the-art and traditional routing protocols such as LEACH \cite{926982}, TEEN \cite{925197}, SEP \cite{smaragdakis2004sep}, DEEC \cite{QING20062230} and DDR \cite{S2RT4387429.OW12N}. We use the values given in Table \ref{tab1}, and our experimentation is divided into two main scenarios: 1) efficient resource utilization, and 2) data denoising. The comparison is carried out over $\left.L = 100, 1000 \text{ and } 10000\right.$ nodes with following metrics: stability and instability period, network lifetime, energy consumption, computational complexity, peak signal-to-noise ratio (PSNR) and structural similarity (SSIM)~index.
\begin{figure}[t!]
	\centering
	\includegraphics[width=1\linewidth]{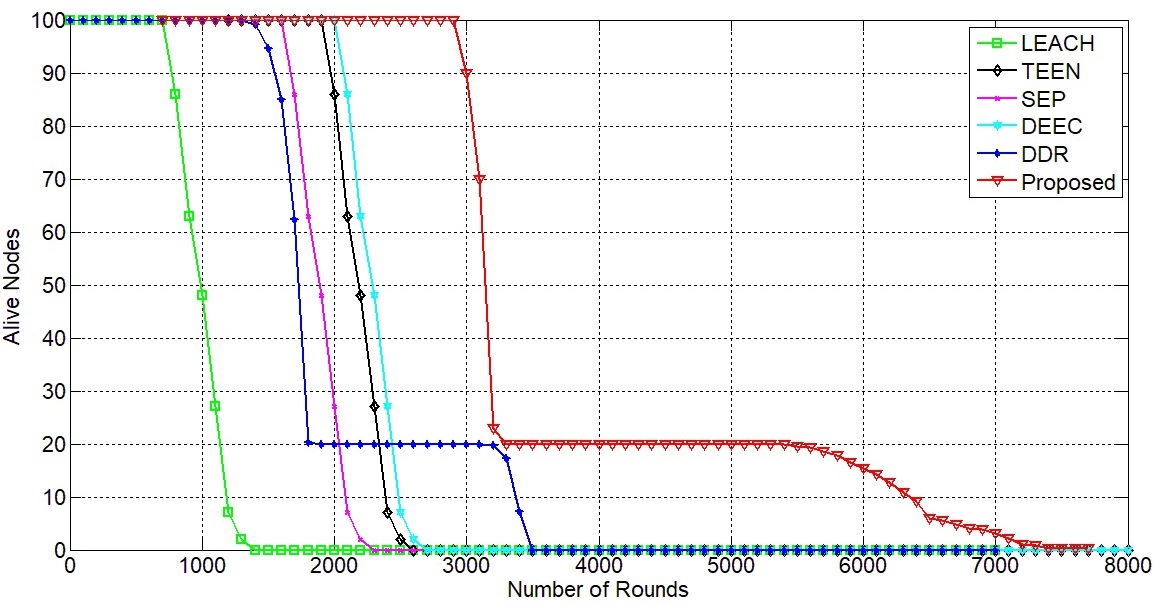}
	\caption{Stability Period}
	\label{fig:fig5}
\end{figure}
\begin{figure}[h!]
	\centering
	\includegraphics[width=1\linewidth]{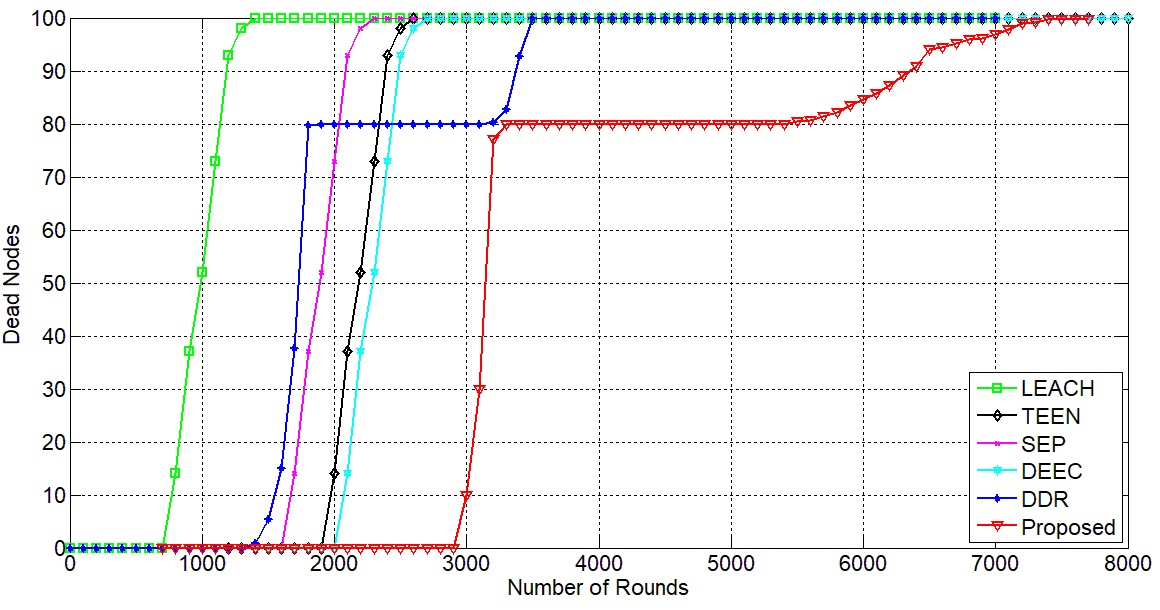}
	\caption{Instability Period}
	\label{fig:fig6}
\end{figure}
\begin{figure}[t!]
	\centering
	\includegraphics[width=1\linewidth]{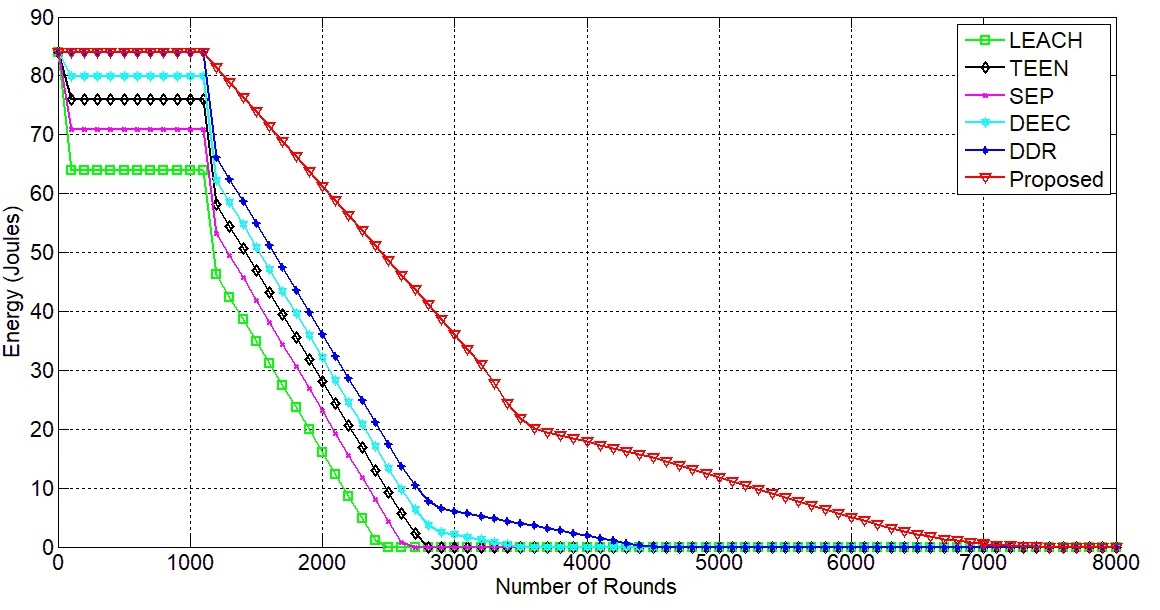}
	\caption{Energy Utilization Comparison}
	\label{fig:fig7}
\end{figure}
\begin{figure}[t!]
	\centering
	\includegraphics[width=1\linewidth]{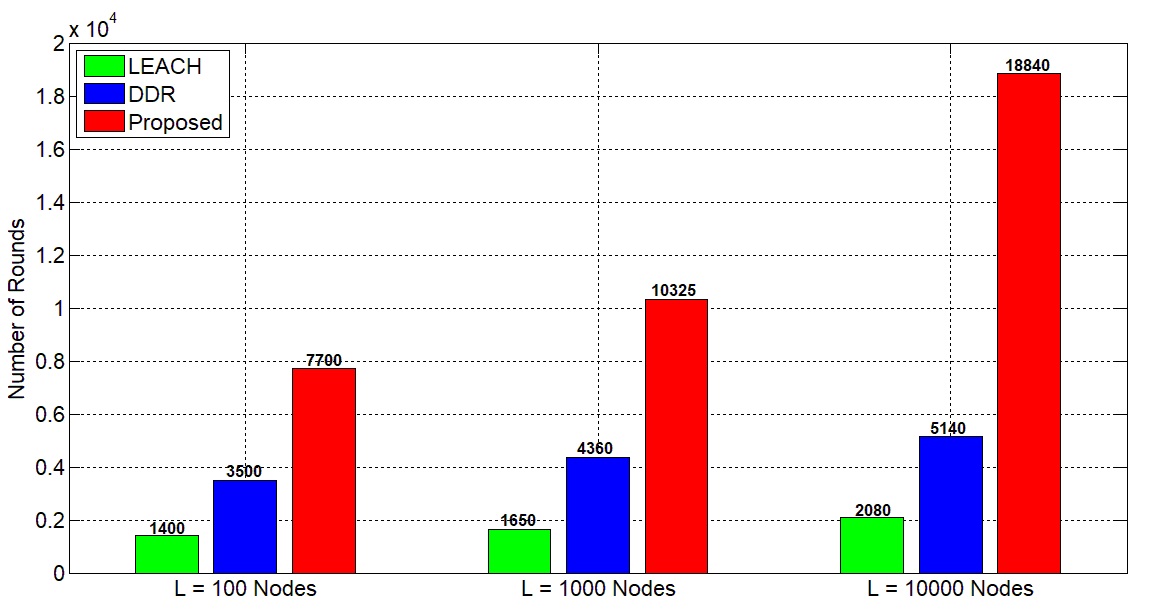}
	\caption{Network Lifetime}
	\label{fig:fig8}
\end{figure}
\begin{figure}[b!]
\centering
\includegraphics[width=1\linewidth]{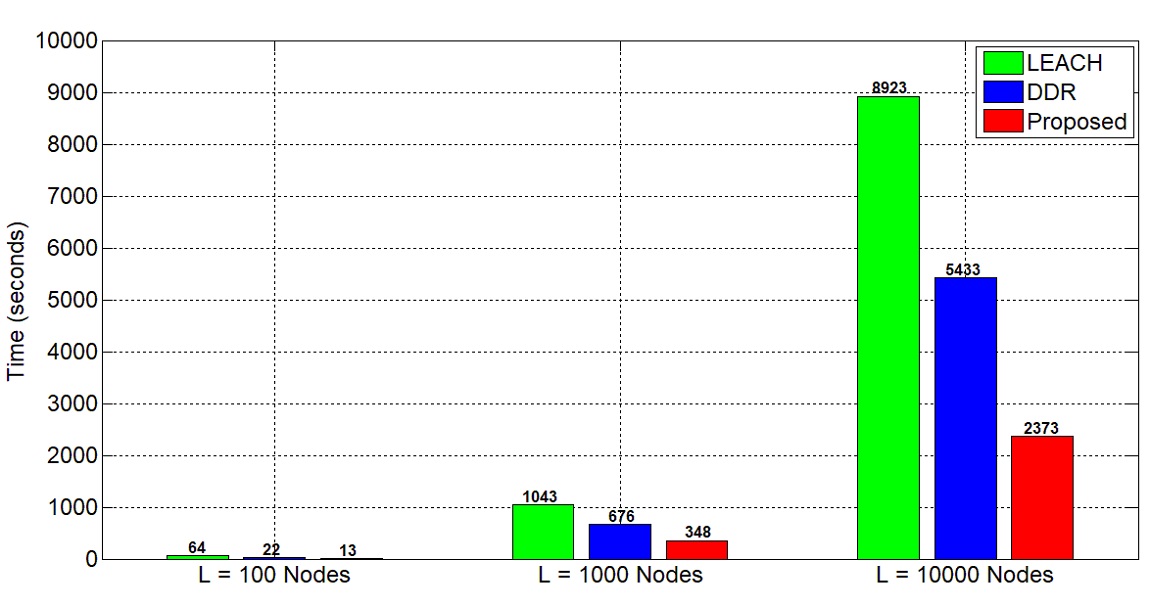}
\caption{Computational Overload Comparison}
\label{fig:fig9}
\end{figure}

A comparison of stability period for $L = 100$ is shown in Fig. \ref{fig:fig5}. This figure demonstrates the number of alive nodes over 8000 sensing rounds. It is evident from the figure that our proposed scheme significantly outperforms all the protocols, and shows promising results. The first node die time of our approach is around 2900, while that of LEACH, TEEN, SEP, DEEC and DDR is around 800, 1900, 1600, 2000, and 1400, respectively. Similarly, Fig. \ref{fig:fig6} illustrates the all node die time (ADT) of these protocols for $L =$ 100. It can be clearly seen that the ADT of our method is $\sim$6390, $\sim$5290, $\sim$5490, $\sim$5190 and $\sim$4290 better than LEACH, TEEN, SEP, DEEC and DDR, respectively. We show that our scheme provides the best ADT, and hence, is the most suitable candidate for practical applications.

We provide a comparison of energy efficient resource utilization in Fig. \ref{fig:fig7}. Here, we show that all protocols start with same energy levels. However, based on the optimized communication method, our scheme demonstrates outstanding results beating all the contestants. In Fig. \ref{fig:fig8}, we compare the network lifetime of our proposed method with LEACH and DDR for $\left.L = 100, 1000 \text{ and } 10000\right.$. It is validated that our protocol is equally competitive on large scale network scenarios outperforming each of the traditional methods.

The complexity of our approach is dominated by the communication yielding a convenient implementation of our method as compared with other protocols as shown in Fig. \ref{fig:fig9}. We compare the computational time consumed by the contestant methods using a $\left.\text{2.20 GHz Intel Core i7-3632QM}\right.$ machine for different number of nodes. This figure proves the robustness of our protocol by showing superior performance, hence, lending itself the most preferable choice for real-time applications.
\begin{table*}[t!]
	\caption{Comparison of Denoising Image Data Samples in Terms of PSNR/SSIM}
	\centering
			\begin{tabular}{|c|c|c|c|c|c|c|c|}\hline\hline
				\multicolumn{2}{|c|}{Noise Level $\sigma_n$} & 10 & 15 & 20 & 50 & 100 \\ \cline{1-7}
				\hline\hline
				\multirow{3}{*}{Lena}				
				& Noisy & 28.03/0.76 & 24.63/0.66 & 22.16/0.58 & 20.13/0.50 & 08.13/0.11 \\ \cline{2-7}
				& Denoised \cite{behzad2018AINA} & 32.64/0.90 & 30.44/0.86 & 28.77/0.82 & 23.97/0.62 & 20.50/0.48 \\ \cline{2-7}
				& Proposed & \textbf{35.54/0.96} & \textbf{33.65/0.94} & \textbf{32.65/0.91} & \textbf{26.21/0.76} & \textbf{23.01/0.61} \\ \hline
				
				\multirow{3}{*}{Barbara}
				& Noisy & 28.18/0.87 & 24.59/0.77 & 22.09/0.69 & 14.10/0.33 & 08.21/0.13 \\ \cline{2-7}
				& Denoised \cite{behzad2018AINA} & 31.88/0.94 & 29.56/0.91 & 27.93/0.88 & 22.75/0.68 & 20.11/0.50 \\ \cline{2-7}
				& Proposed & \textbf{35.36/0.97} & \textbf{32.34/0.93} & \textbf{29.09/0.91} & \textbf{25.67/0.79} & \textbf{23.21/0.61} \\ \hline
				
				\multirow{3}{*}{House}
				& Noisy & 28.07/0.51 & 24.57/0.44 & 22.02/0.38 & 14.03/0.19 & 08.09/0.07 \\ \cline{2-7}
				& Denoised \cite{behzad2018AINA} & 35.28/0.67 & 32.63/0.61 & 31.33/0.58 & 25.80/0.45 & 22.14/0.27 \\ \cline{2-7}
				& Proposed & \textbf{38.34/0.79} & \textbf{35.74/0.69} & \textbf{33.90/0.64} & \textbf{29.04/0.53} & \textbf{24.03/0.29} \\ \hline
				
				\multirow{3}{*}{Peppers}
				& Noisy & 28.08/0.81 & 24.72/0.72 & 22.13/0.63 & 14.17/0.33 & 08.16/0.13 \\ \cline{2-7}
				& Denoised \cite{behzad2018AINA} & 32.00/0.92 & 29.74/0.89 & 28.11/0.85 & 23.09/0.68 & 19.62/0.49 \\ \cline{2-7}
				& Proposed & \textbf{35.03/0.94} & \textbf{32.44/0.90} & \textbf{20.95/0.88} & \textbf{27.34/0.73} & \textbf{21.34/0.54} \\ \hline
				 
				\multirow{3}{*}{Boat}
				& Noisy & 28.08/0.75 & 24.59/0.63 & 22.05/0.53 & 14.17/0.24 & 08.10/0.09 \\ \cline{2-7}
				& Denoised \cite{behzad2018AINA} & 31.59/0.86 & 29.11/0.76 & 27.48/0.69 & 23.42/0.45 & 20.64/0.26 \\ \cline{2-7}
				& Proposed & \textbf{33.97/0.89} & \textbf{31.34/0.84} & \textbf{28.98/0.75} & \textbf{26.34/0.59} & \textbf{22.34/0.42} \\ \hline 
				
				\multirow{3}{*}{C-man}
				& Noisy & 28.07/0.53 & 24.56/0.45 & 22.09/0.40 & 14.13/0.21 & 08.18/0.10 \\ \cline{2-7}
				& Denoised \cite{behzad2018AINA} & 33.28/0.75 & 31.21/0.69 & 29.23/0.63 & 24.19/0.45 & 20.67/0.25 \\ \cline{2-7}
				& Proposed & \textbf{35.53/0.86} & \textbf{34.34/0.74} & \textbf{32.34/0.71} & \textbf{26.53/0.55} & \textbf{22.24/0.30} \\ \hline
								
				\multirow{3}{*}{Room}
				& Noisy & 28.21/0.80 & 24.62/0.68 & 22.07/0.58 & 14.19/0.25 & 08.10/0.09 \\ \cline{2-7}
				& Denoised \cite{behzad2018AINA} & 31.59/0.86 & 29.11/0.76 & 27.48/0.69 & 23.42/0.45 & 20.64/0.26 \\ \cline{2-7}
				& Proposed & \textbf{33.53/0.92} & \textbf{31.64/0.89} & \textbf{29.53/0.85} & \textbf{25.30/0.63} & \textbf{22.11/0.41} \\ \hline
						
				\multirow{3}{*}{Mandrill}
				& Noisy & 27.99/0.80 & 24.56/0.66 & 21.98/0.54 & 14.16/0.20 & 08.18/0.06 \\ \cline{2-7}
				& Denoised \cite{behzad2018AINA} & 30.88/0.85 & 28.51/0.75 & 27.09/0.67 & 24.17/0.47 & 21.30/0.28 \\ \cline{2-7}
				& Proposed & \textbf{34.87/0.91} & \textbf{31.51/0.76} & \textbf{29.93/0.73} & \textbf{26.54/0.53} & \textbf{23.87/0.30} \\ \hline\hline
			\end{tabular}
			\label{tab2}
\end{table*}


Finally, the detailed denoising results of various standard images are shown in Fig. \ref*{fig:denoised_start}-\ref*{fig:denoised_end}, and summarized in Table \ref{tab2}. We opt globally adopted PSNR and SSIM as evaluation metrics to prove that the denoising section of our proposed framework produces equally promising outcomes. The provided table summarizes denoising results of a number of images, as PSNR/SSIM, over a range of noise levels, i.e., $\left.\sigma_n = [10,15,20,50,100]\right.$. Similarly, we also present the extensive denoising results of color images in Fig. \ref*{fig:color_denoised}. As can be seen by in these figures and table, the recovered images are a very good approximation of original images thereby verifying the effectivenesses of our proposed framework.

For experimentation, we transmitted various images among deployed nodes and showed that the resultant images received at the receiver suffers from Gaussian noise. The PSNR and SSIM values of the corresponding received noisy images are shown in the table. In comparison with our denoised images, we show that a significant amount of improvement is achieved in terms of the noise being removed, and the actual data is recovered to a greater extent. Consequently, these results confirm that our proposed framework is indeed an effective and robust model for real-time scenarios in WSNs which outperforms many traditionally proposed routing protocols.

\newpage
\begin{figure*}[t]
	\centering
	\includegraphics[width=1\linewidth]{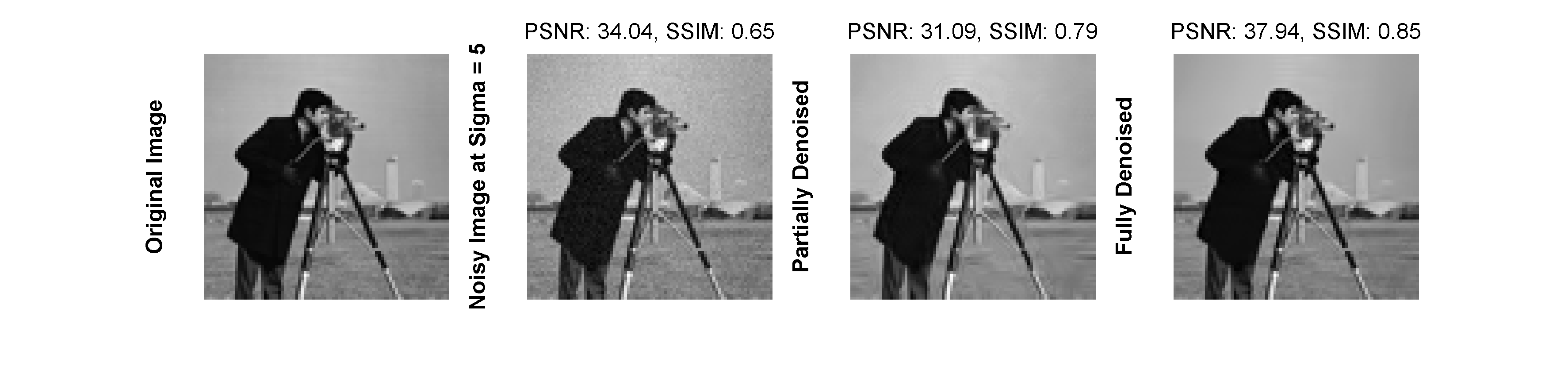}
	\includegraphics[width=1\linewidth]{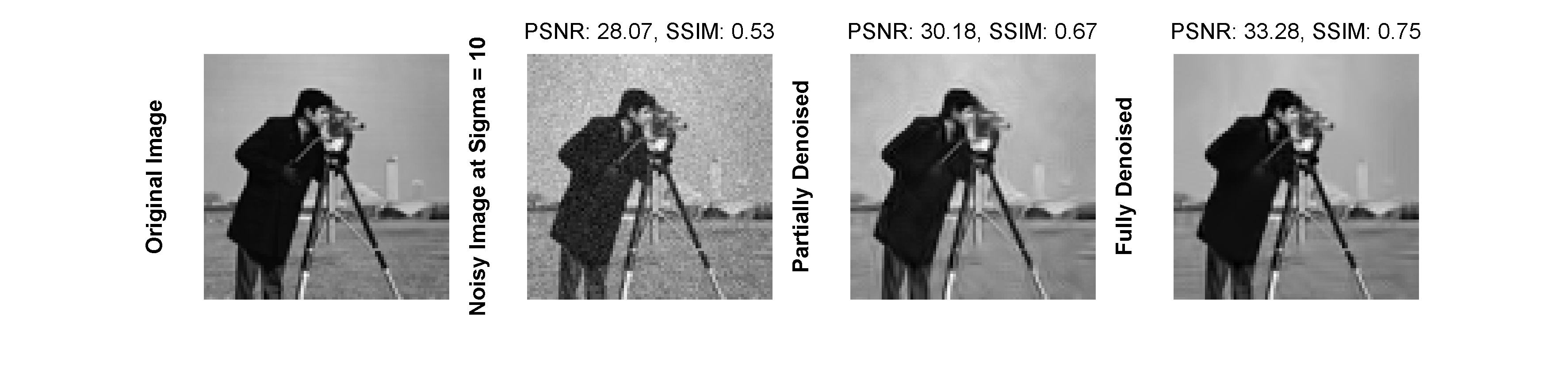}
	\includegraphics[width=1\linewidth]{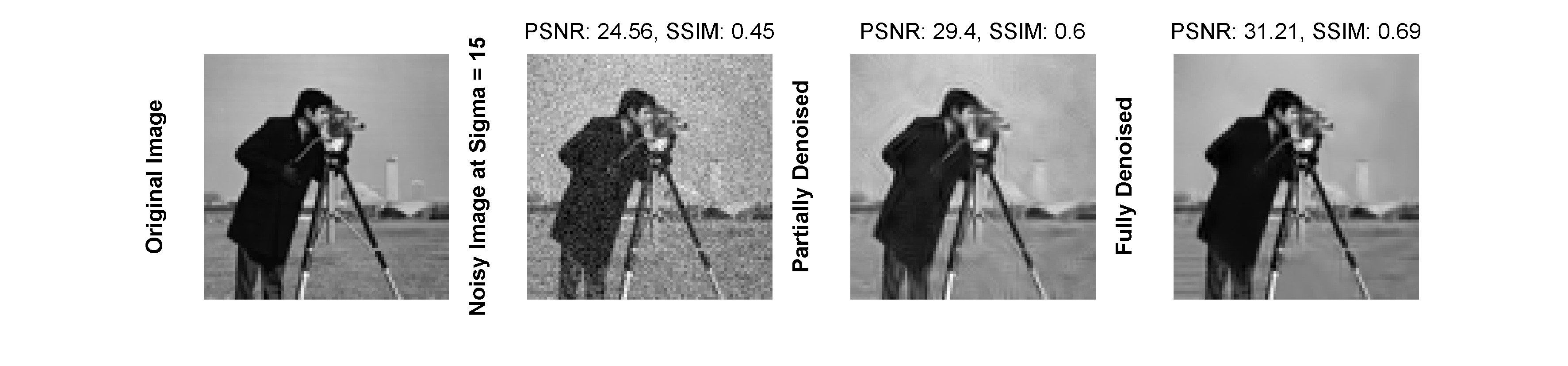}
	\includegraphics[width=1\linewidth]{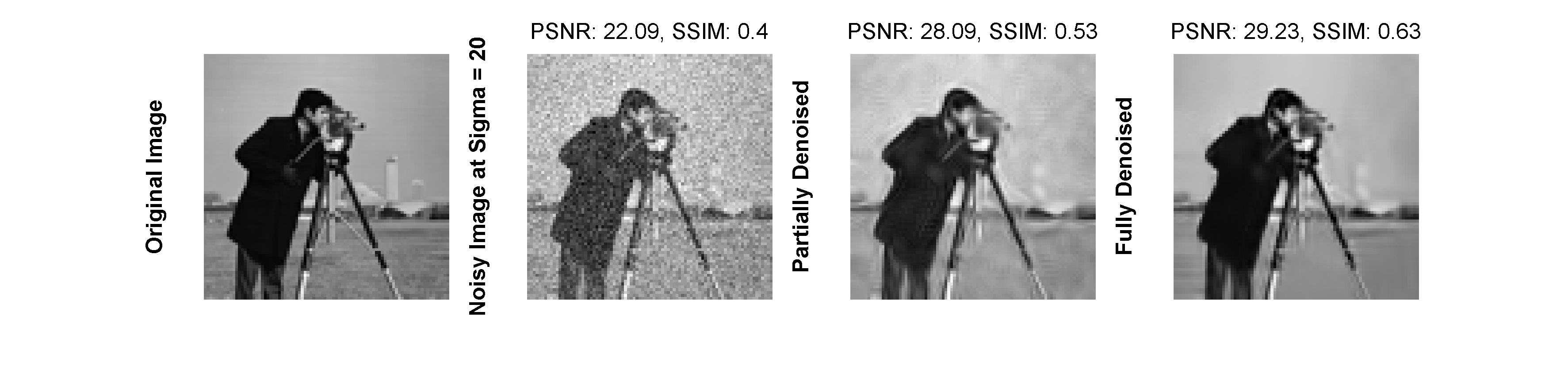}
	\includegraphics[width=1\linewidth]{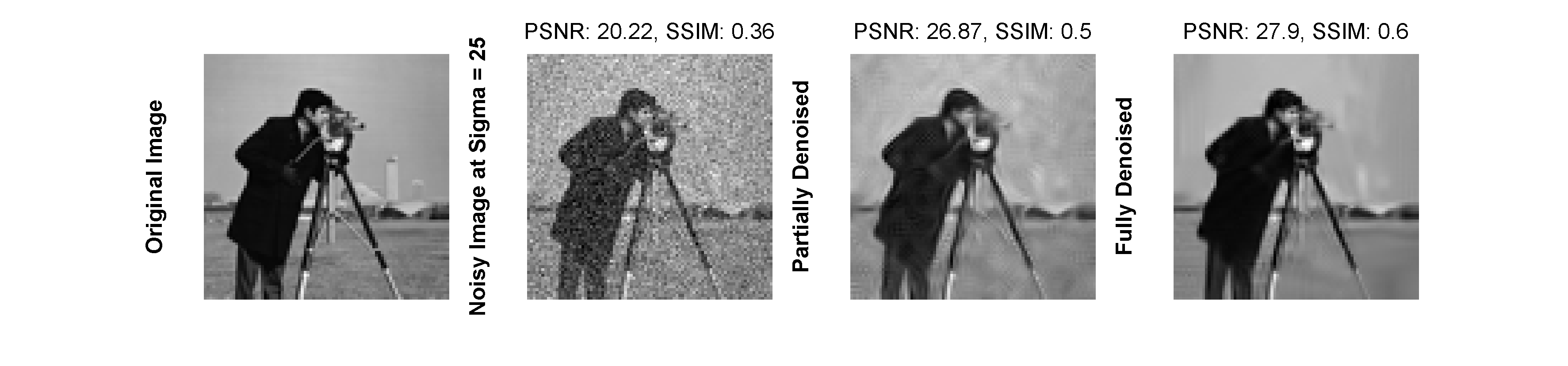}
\end{figure*}
\newpage
\begin{figure*}[t]
	\centering
	\includegraphics[width=1\linewidth]{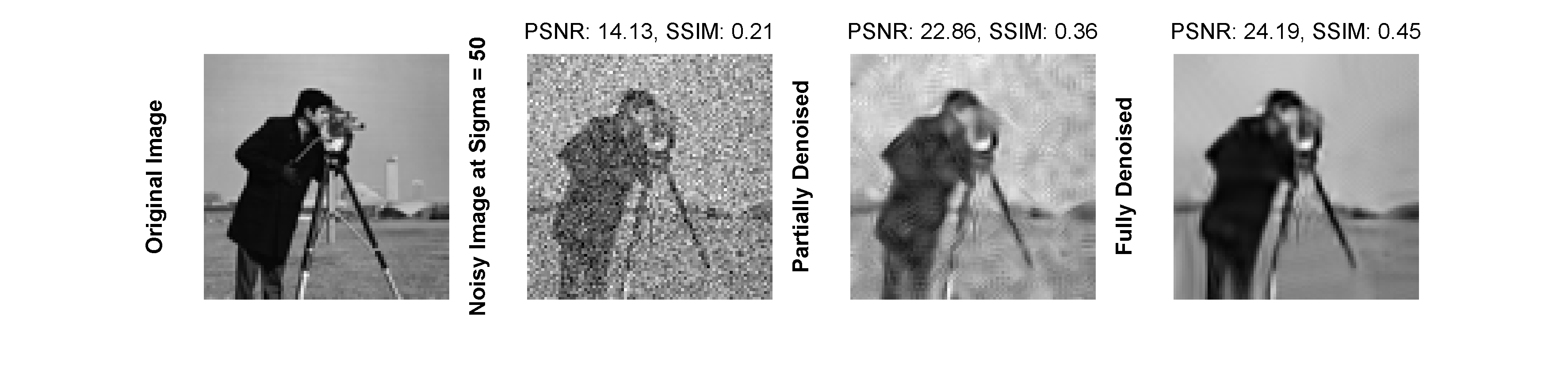}
	\includegraphics[width=1\linewidth]{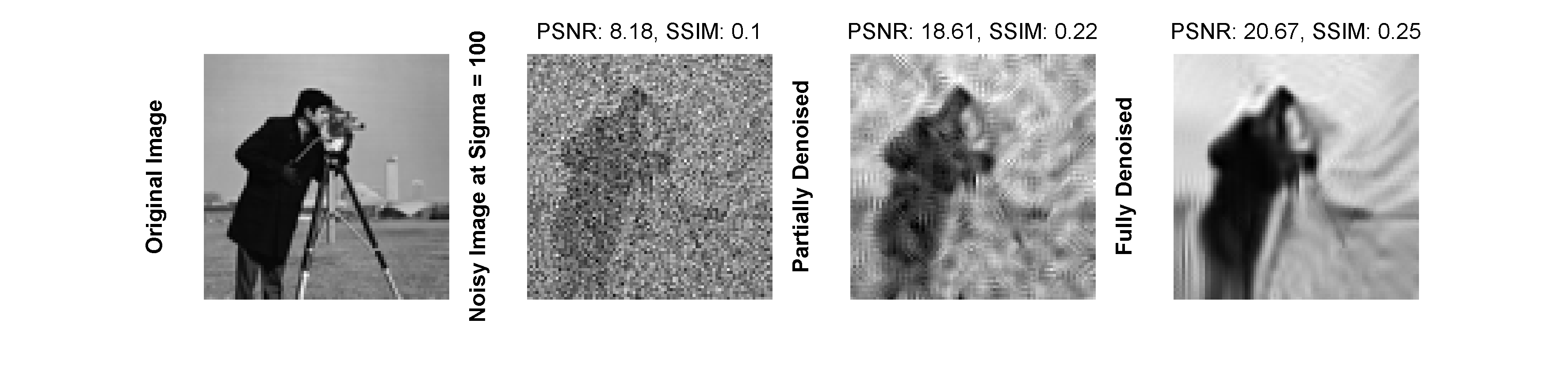}
	\includegraphics[width=1\linewidth]{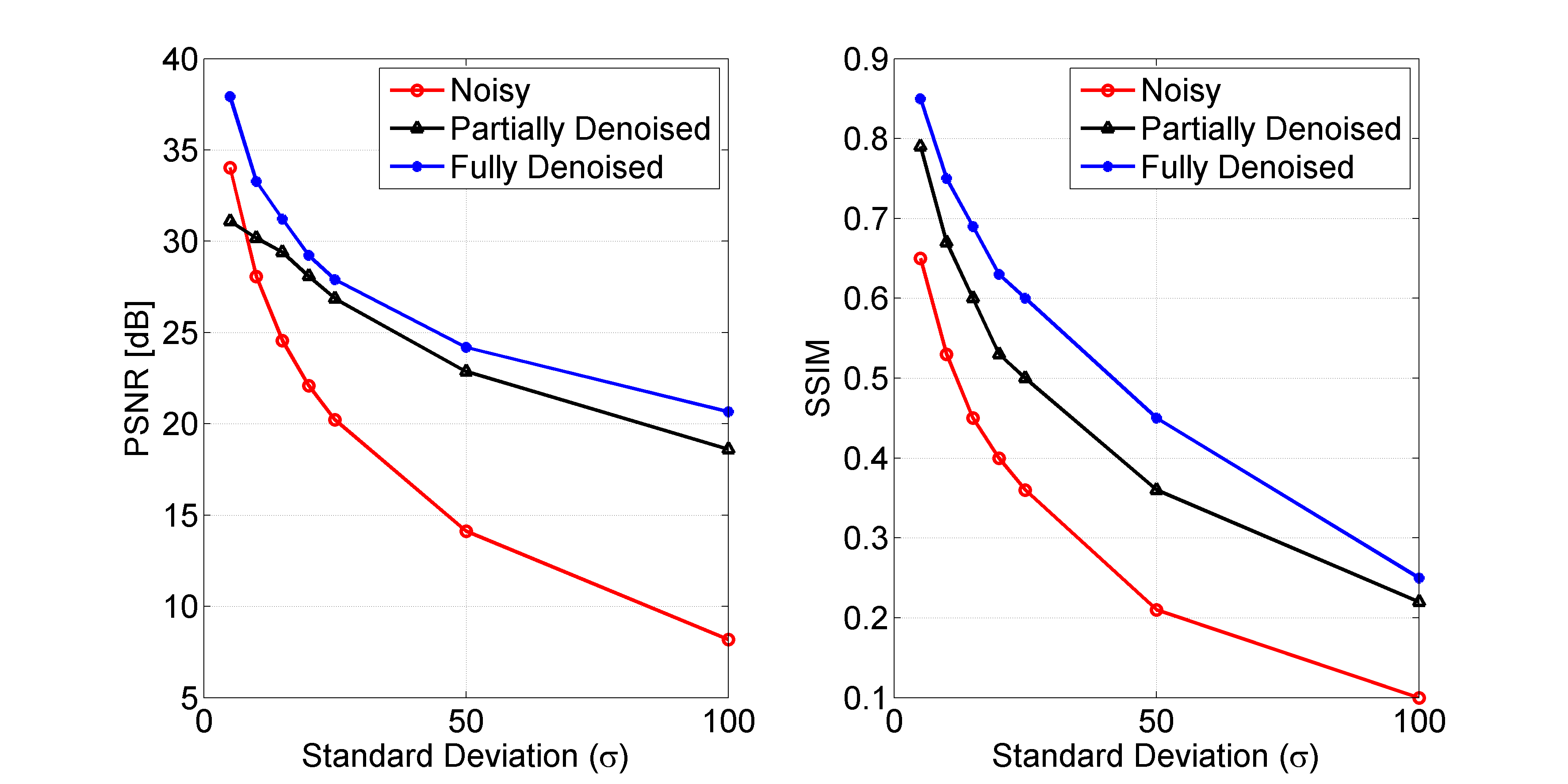}
	\caption{Denoising $256 \times256$ grayscale \textit{Cameraman} standard test data images over noise $\sigma =  [5,10,15,20,25,50,100]$ when received at a node $\mu_\alpha$. Each row represent an original image, a noisy image, a partially denoised, and a fully denoised image, respectively, corrupted by a specific level of additive white Gaussian noise (AWGN). The graphical results in the end show PSNR [dB] and SSIM results in the form of graphs.}
	\label{fig:denoised_start}
\end{figure*}

\newpage
\begin{figure*}[t]
	\centering
	\includegraphics[width=1\linewidth]{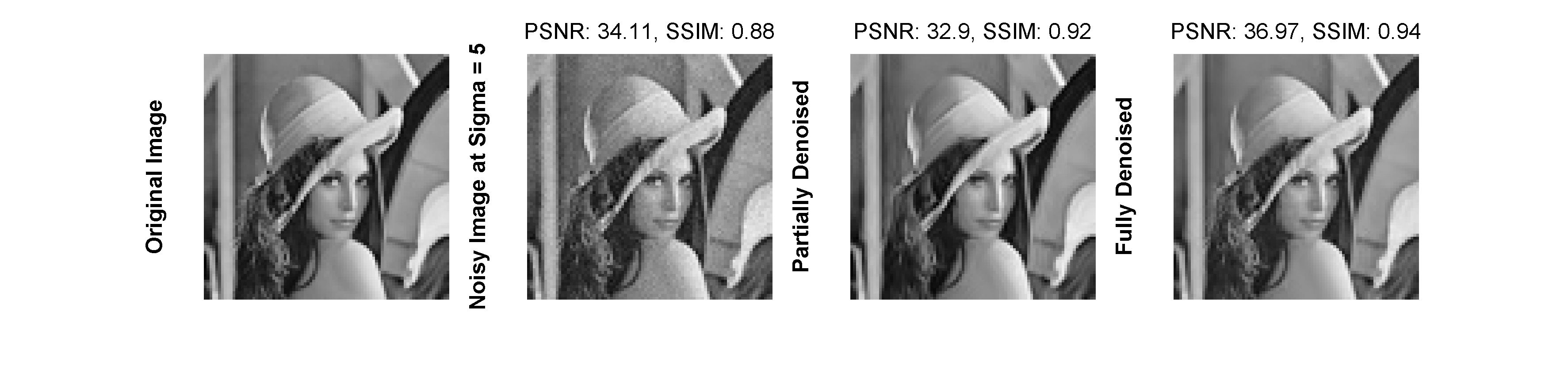}
	\includegraphics[width=1\linewidth]{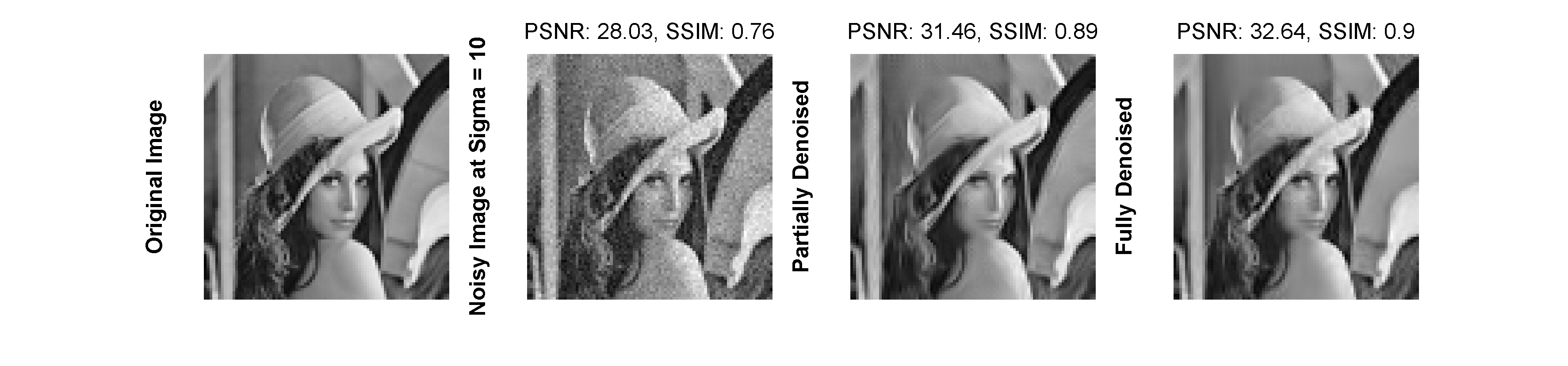}
	\includegraphics[width=1\linewidth]{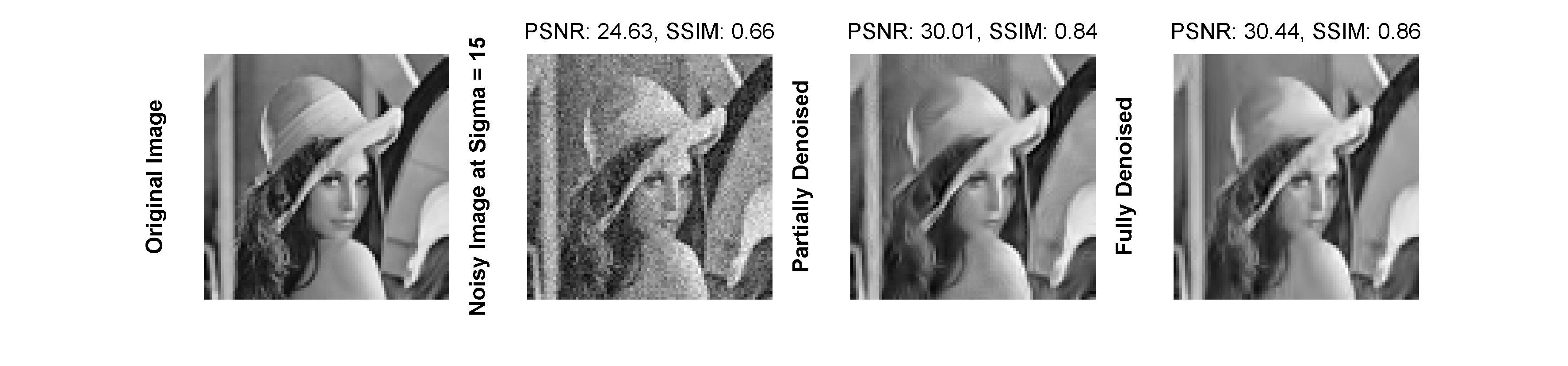}
	\includegraphics[width=1\linewidth]{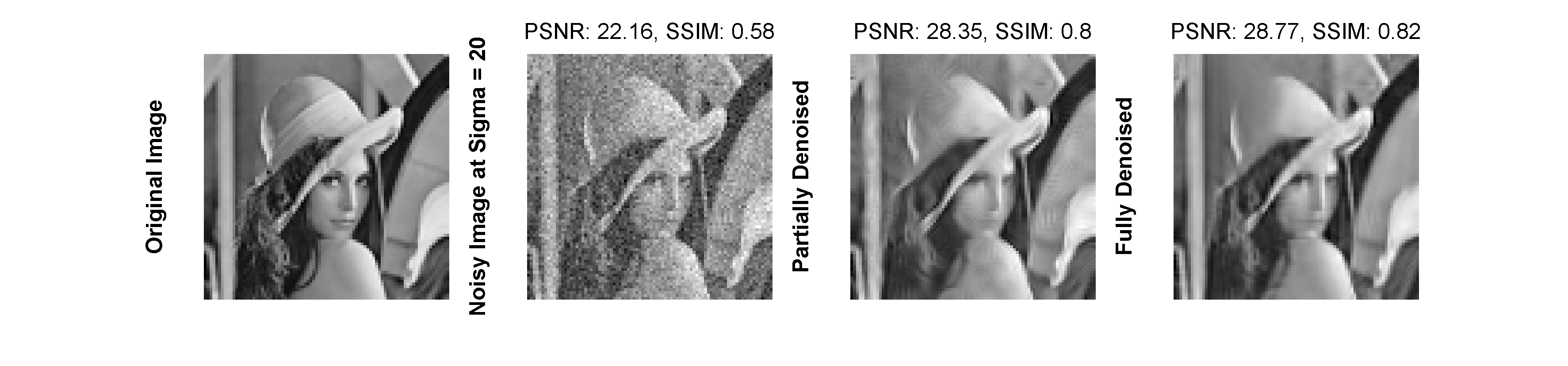}
	\includegraphics[width=1\linewidth]{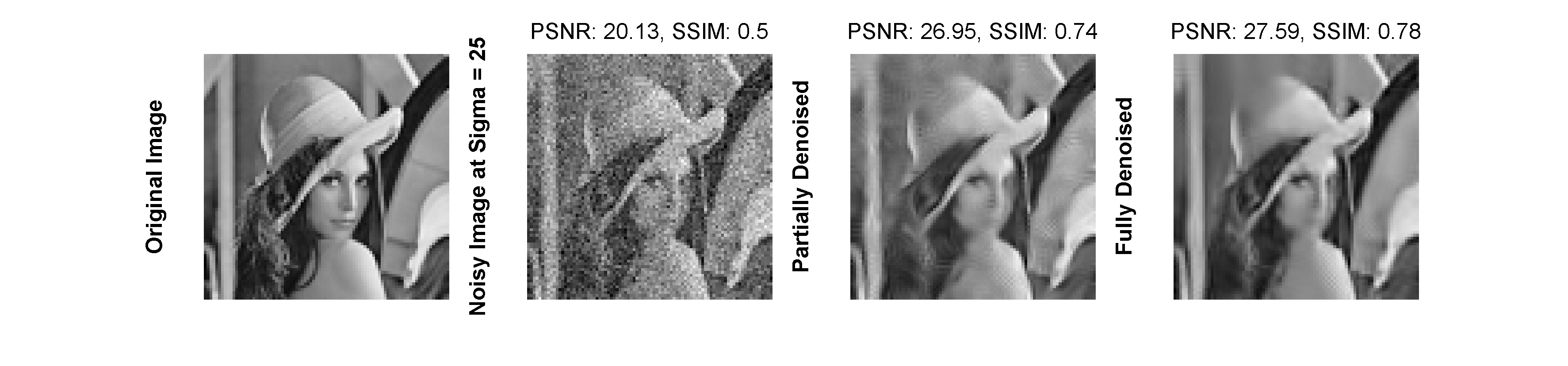}
\end{figure*}
\newpage
\begin{figure*}[t]
	\centering
	\includegraphics[width=1\linewidth]{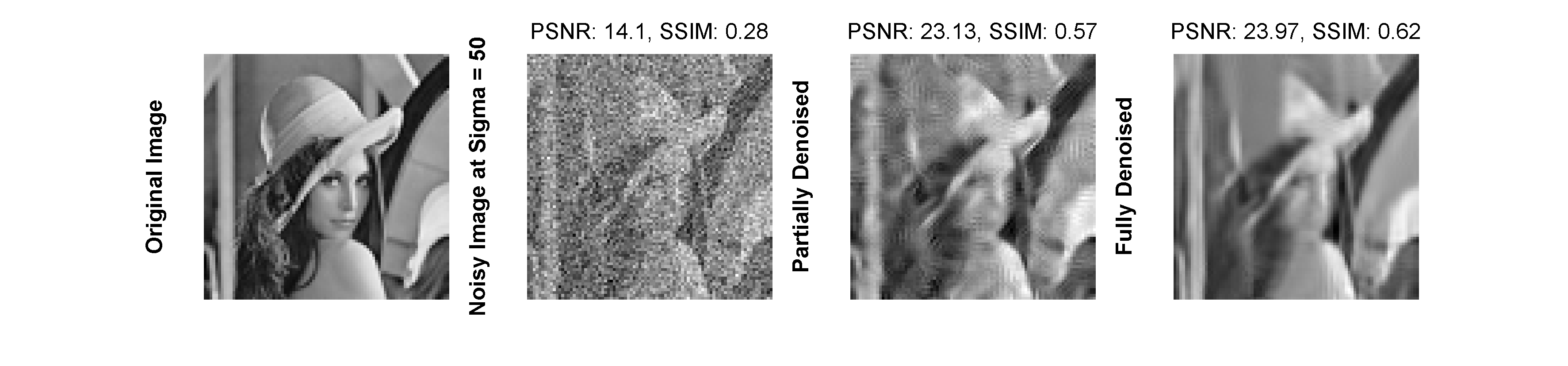}
	\includegraphics[width=1\linewidth]{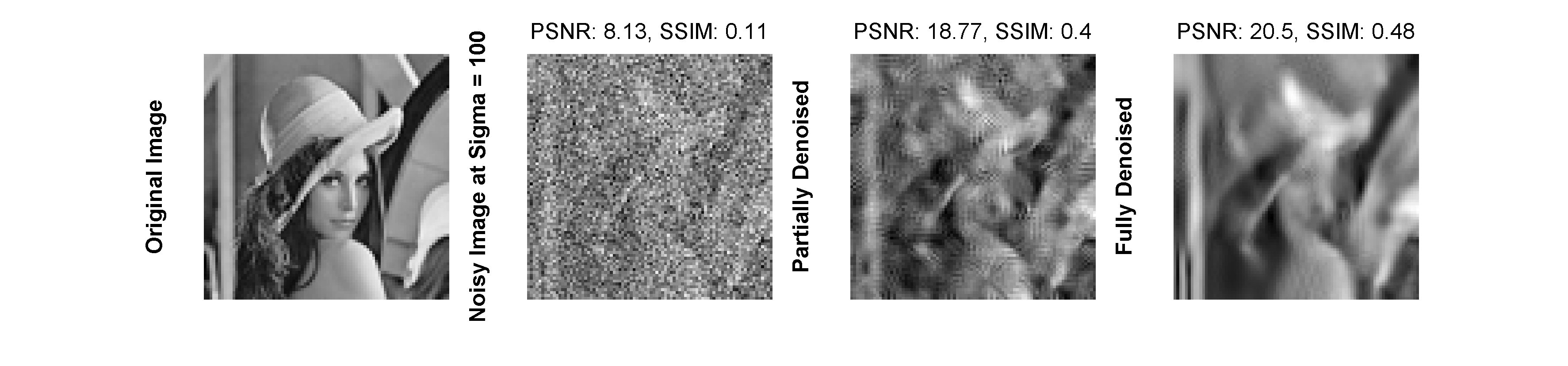}
	\includegraphics[width=1\linewidth]{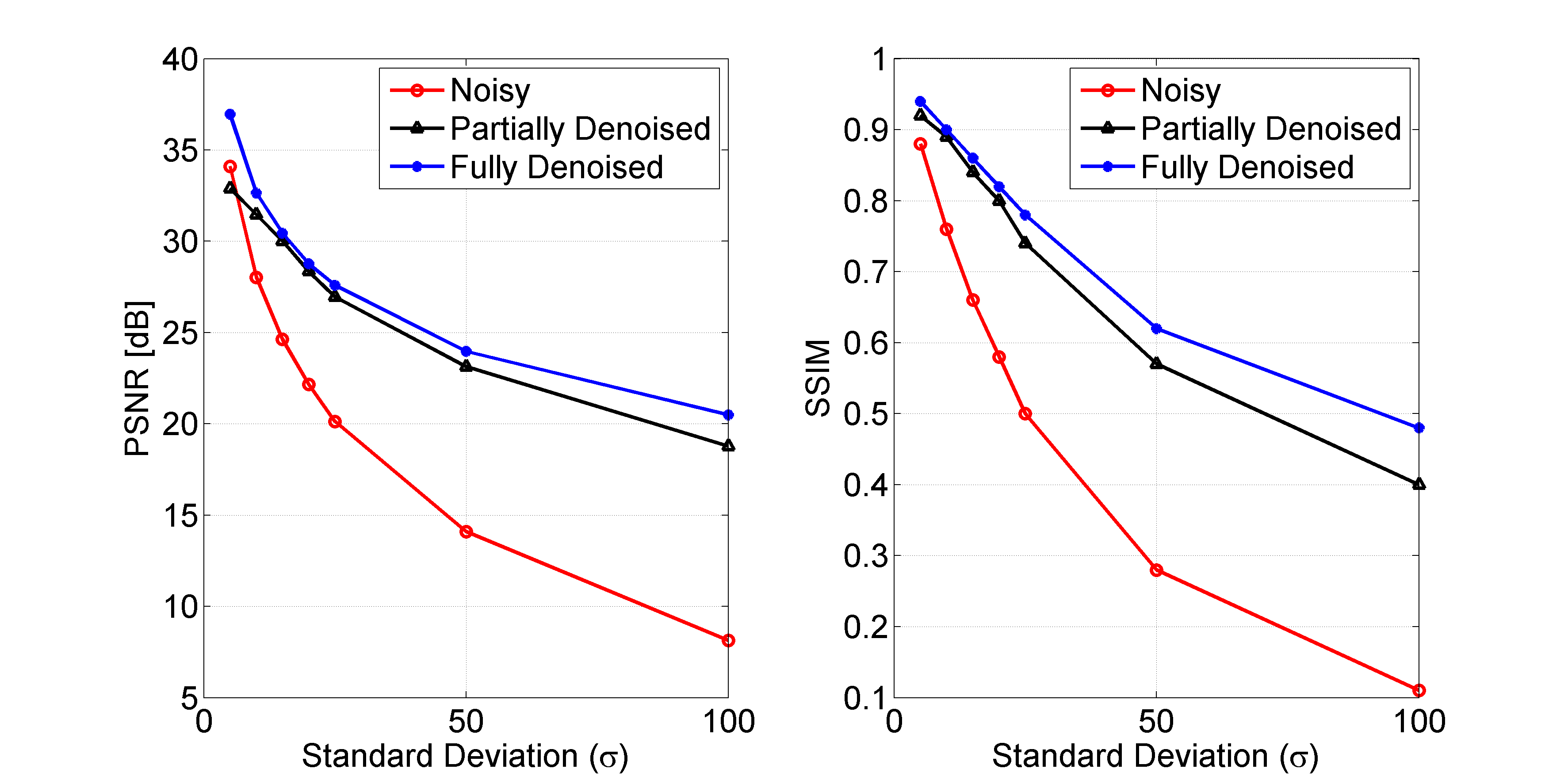}
	\caption{Denoising $256 \times256$ grayscale \textit{Lena} standard test data images over noise $\sigma =  [5,10,15,20,25,50,100]$ when received at a node $\mu_\alpha$. Each row represent an original image, a noisy image, a partially denoised, and a fully denoised image, respectively, corrupted by a specific level of additive white Gaussian noise (AWGN). The graphical results in the end show PSNR [dB] and SSIM results in the form of graphs.}
\end{figure*}

\newpage
\begin{figure*}[t]
	\centering
	\includegraphics[width=1\linewidth]{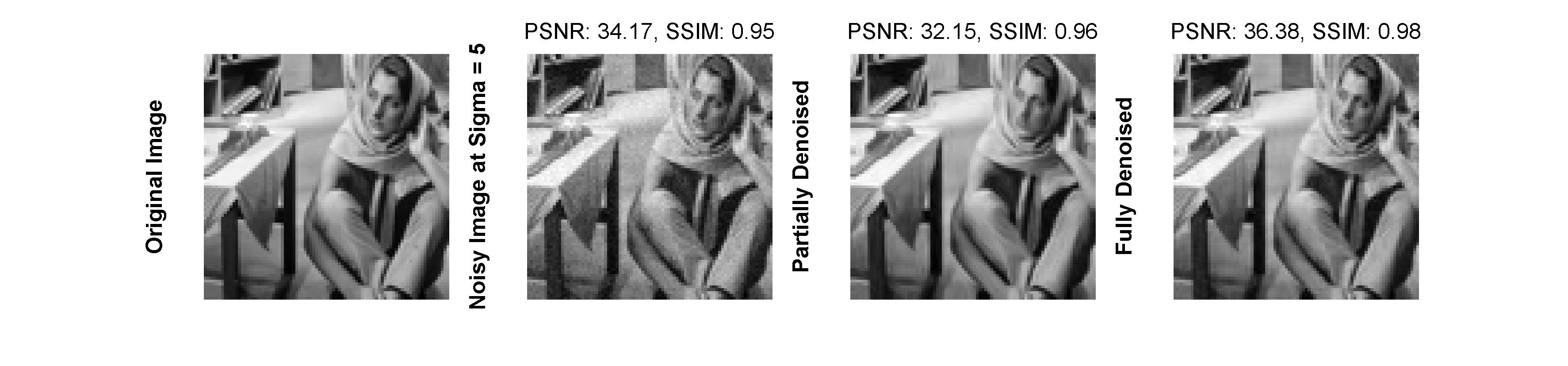}
	\includegraphics[width=1\linewidth]{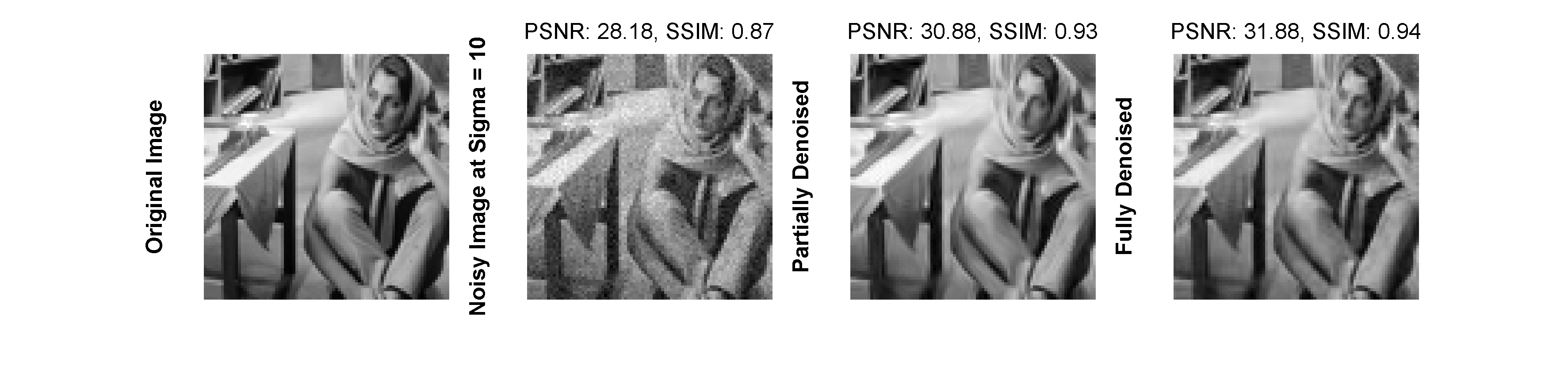}
	\includegraphics[width=1\linewidth]{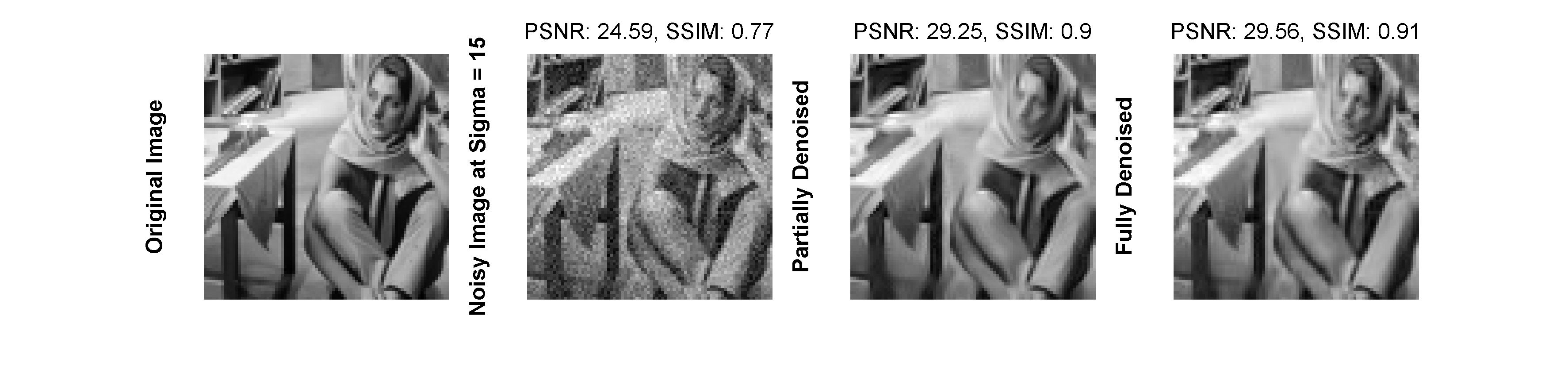}
	\includegraphics[width=1\linewidth]{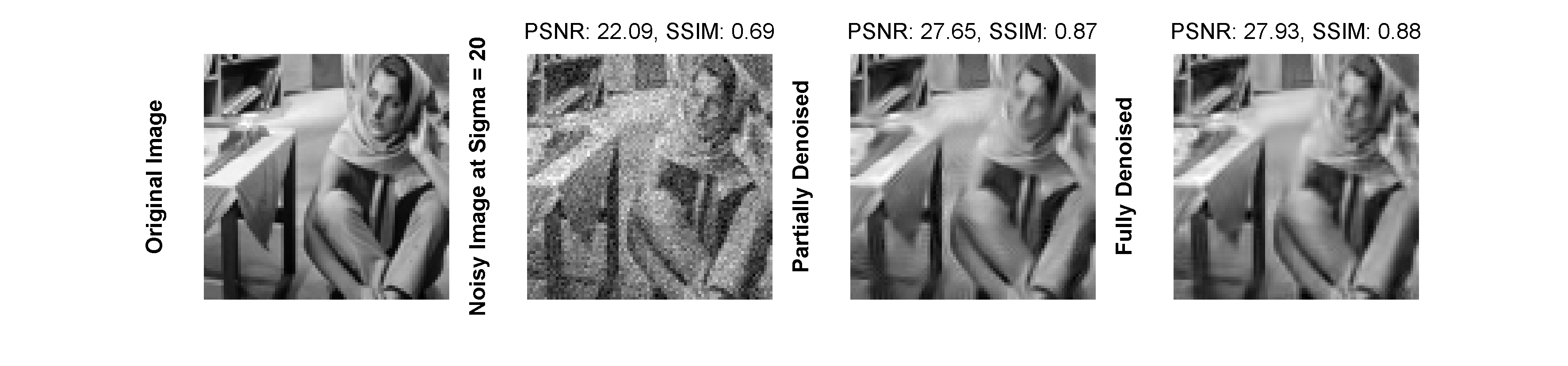}
	\includegraphics[width=1\linewidth]{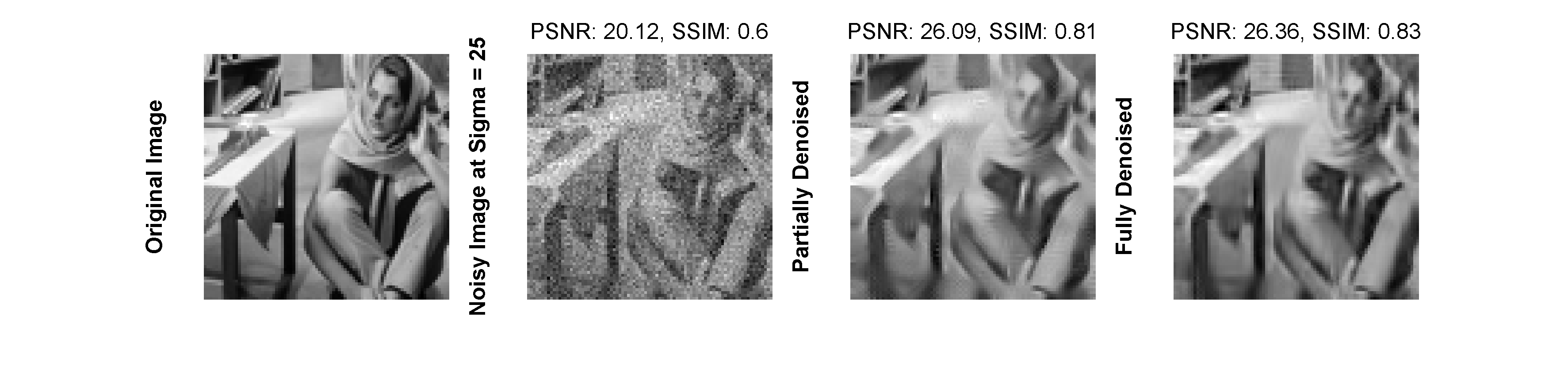}
\end{figure*}
\newpage
\begin{figure*}[t]
	\centering
	\includegraphics[width=1\linewidth]{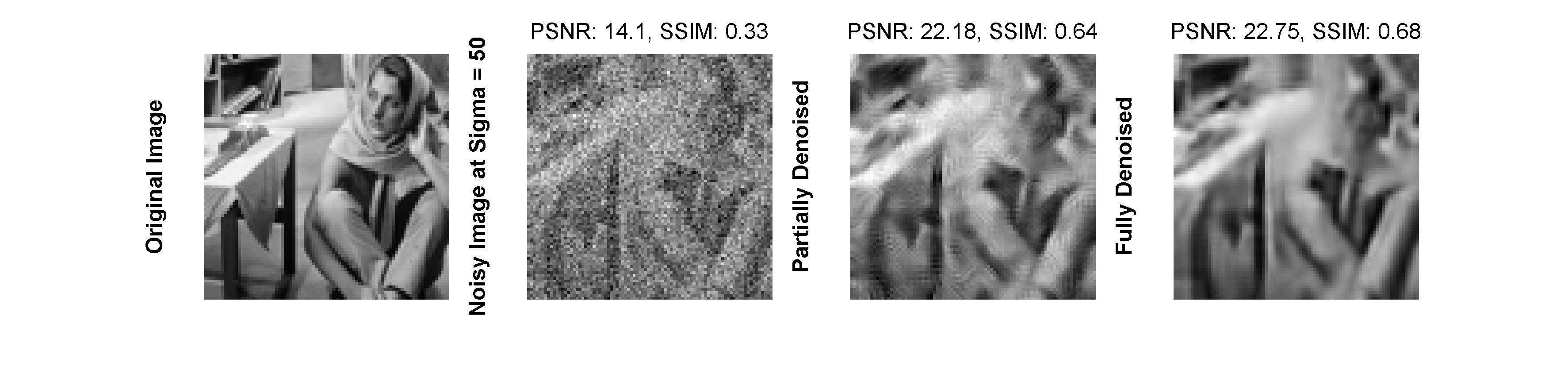}
	\includegraphics[width=1\linewidth]{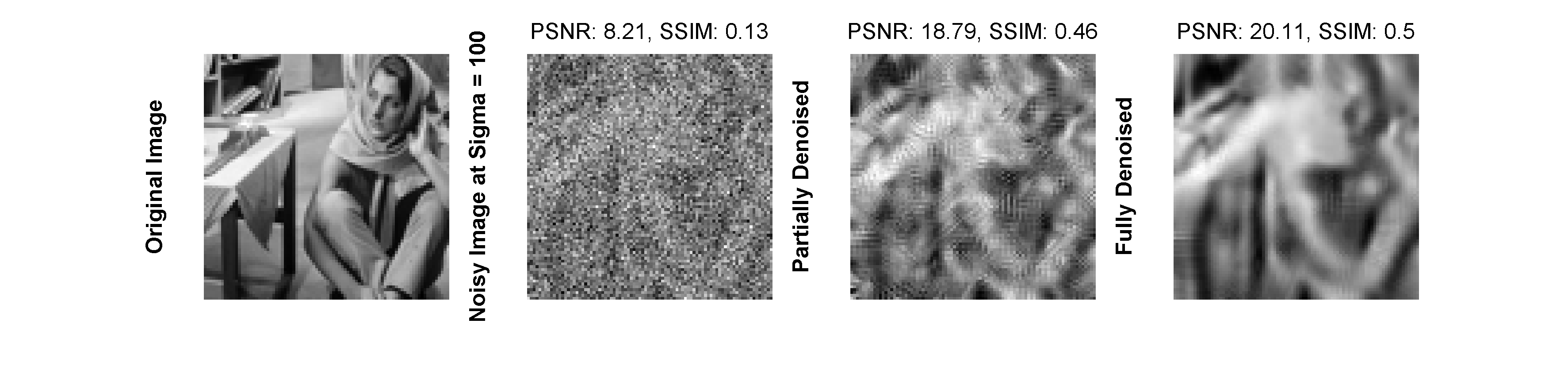}
	\includegraphics[width=1\linewidth]{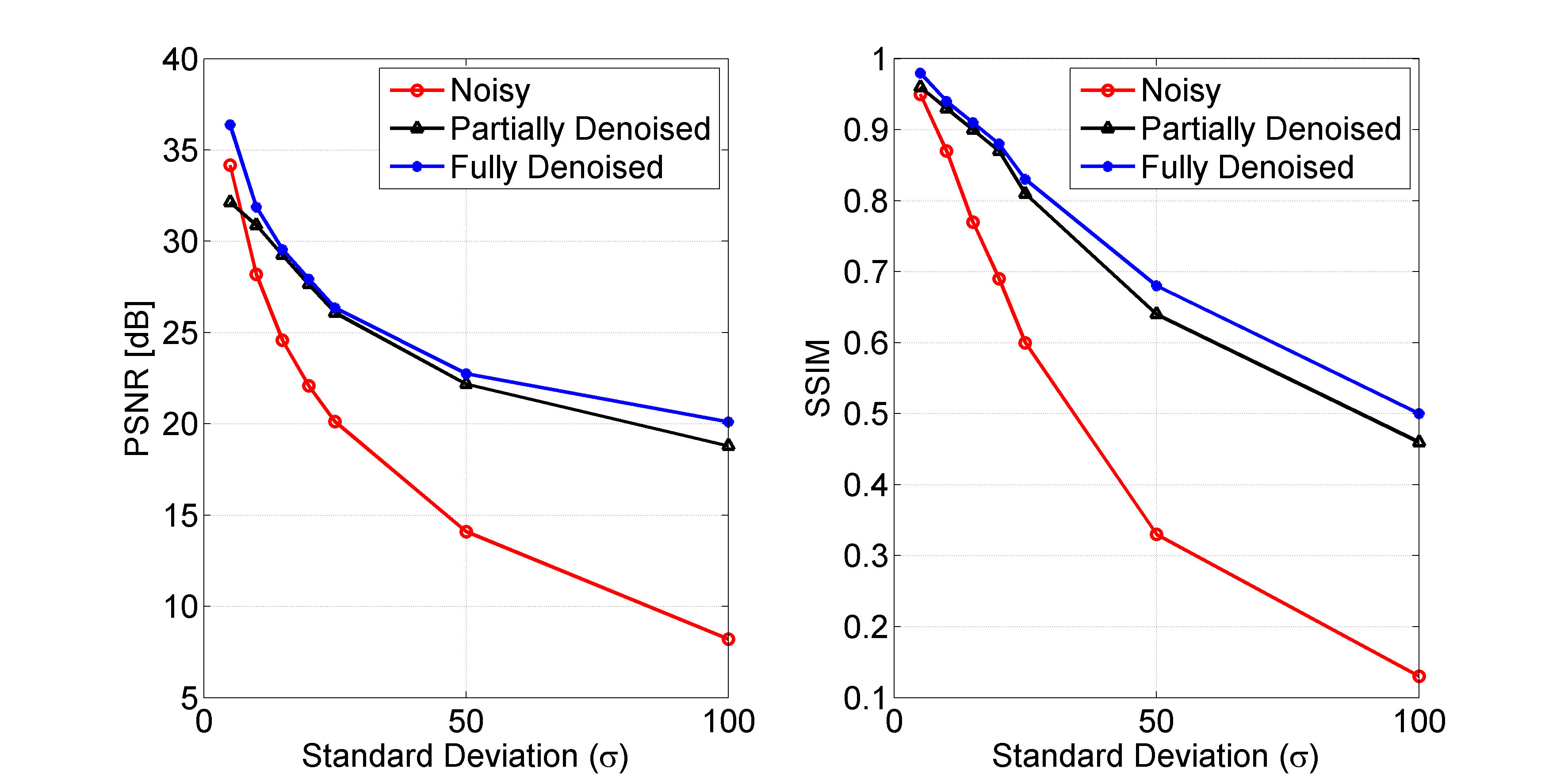}
	\caption{Denoising $256 \times256$ grayscale \textit{Barbara} standard test data images over noise $\sigma =  [5,10,15,20,25,50,100]$ when received at a node $\mu_\alpha$. Each row represent an original image, a noisy image, a partially denoised, and a fully denoised image, respectively, corrupted by a specific level of additive white Gaussian noise (AWGN). The graphical results in the end show PSNR [dB] and SSIM results in the form of graphs.}
\end{figure*}

\newpage
\begin{figure*}[t]
	\centering
	\includegraphics[width=1\linewidth]{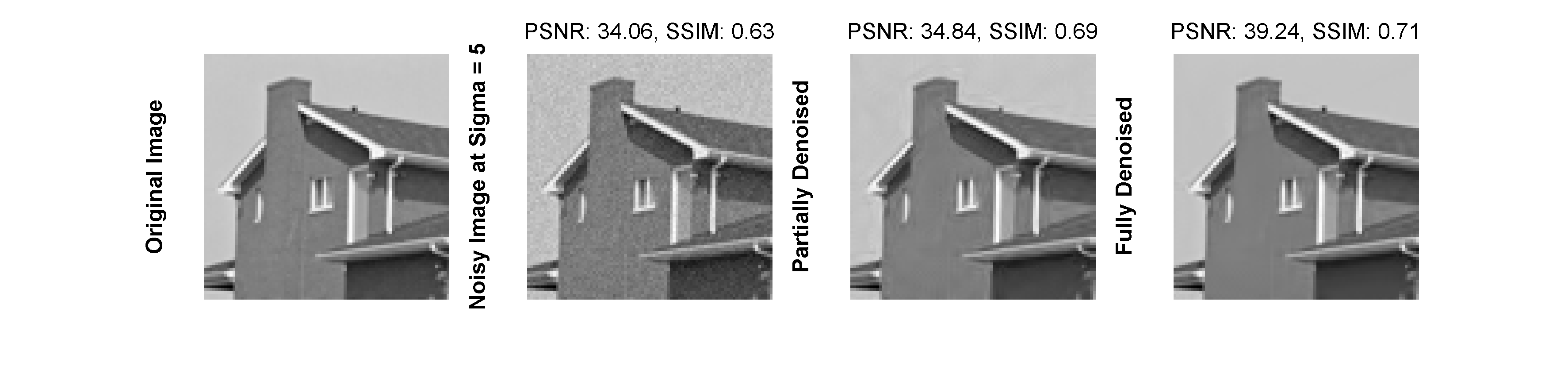}
	\includegraphics[width=1\linewidth]{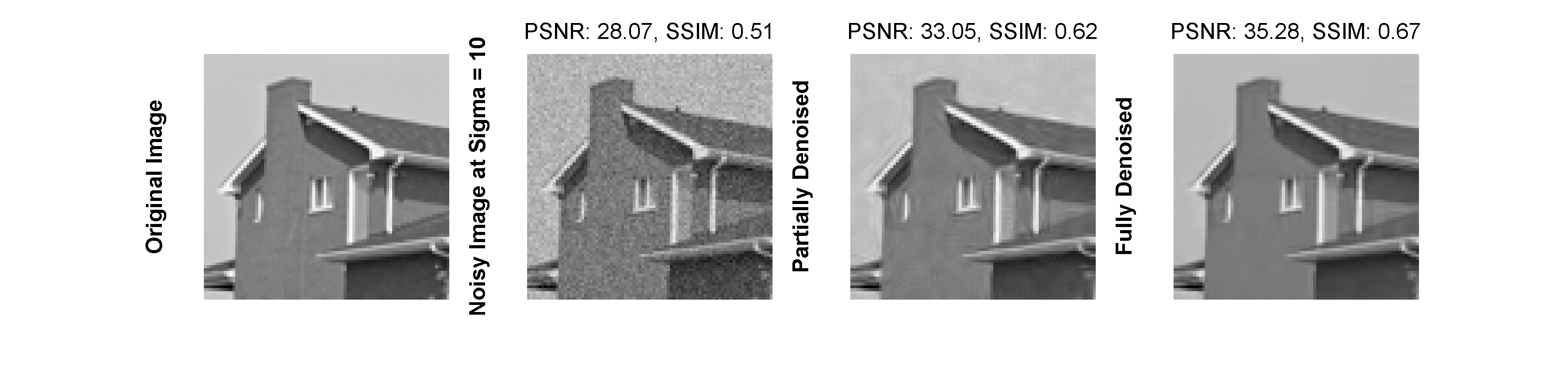}
	\includegraphics[width=1\linewidth]{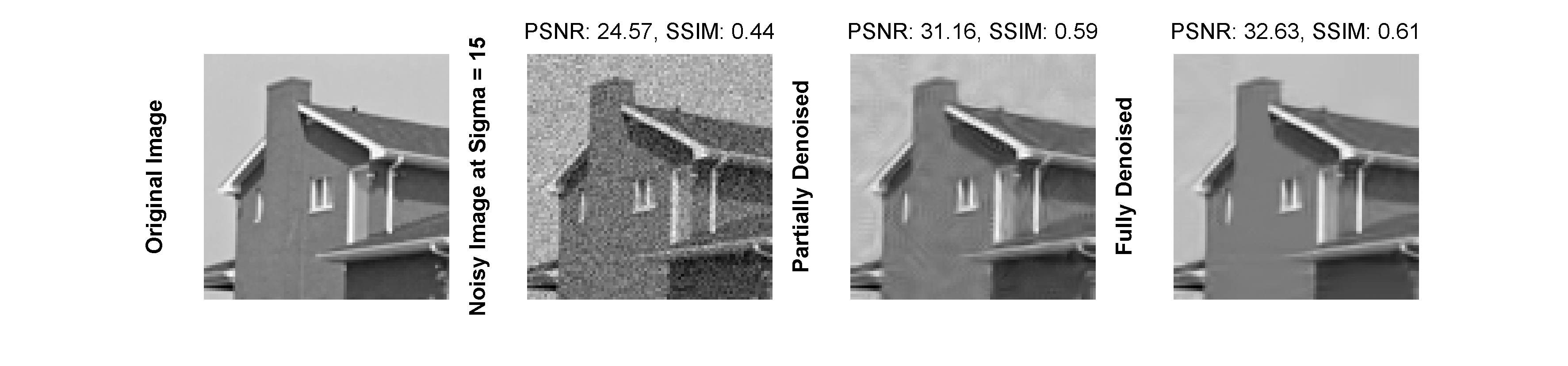}
	\includegraphics[width=1\linewidth]{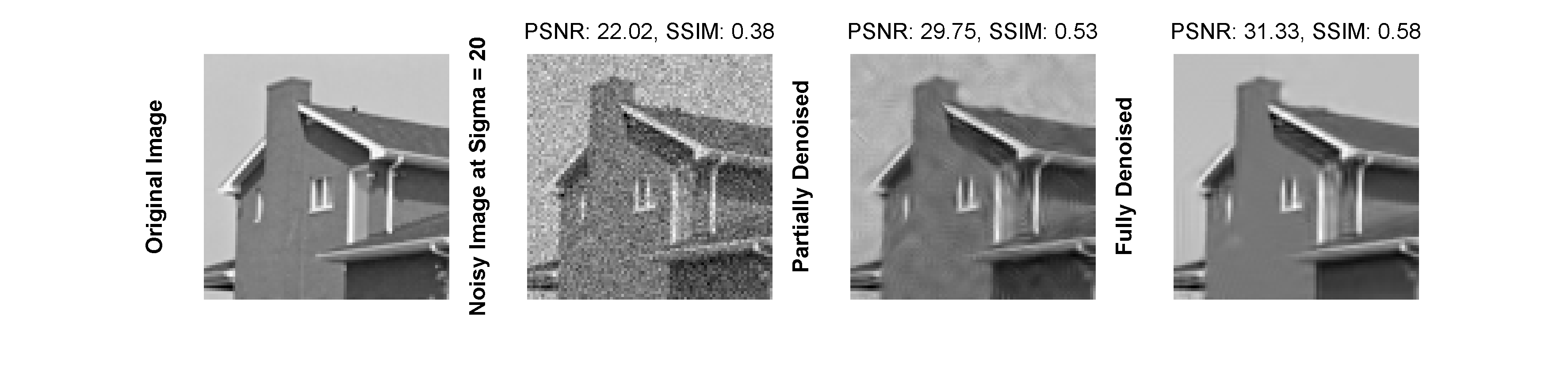}
	\includegraphics[width=1\linewidth]{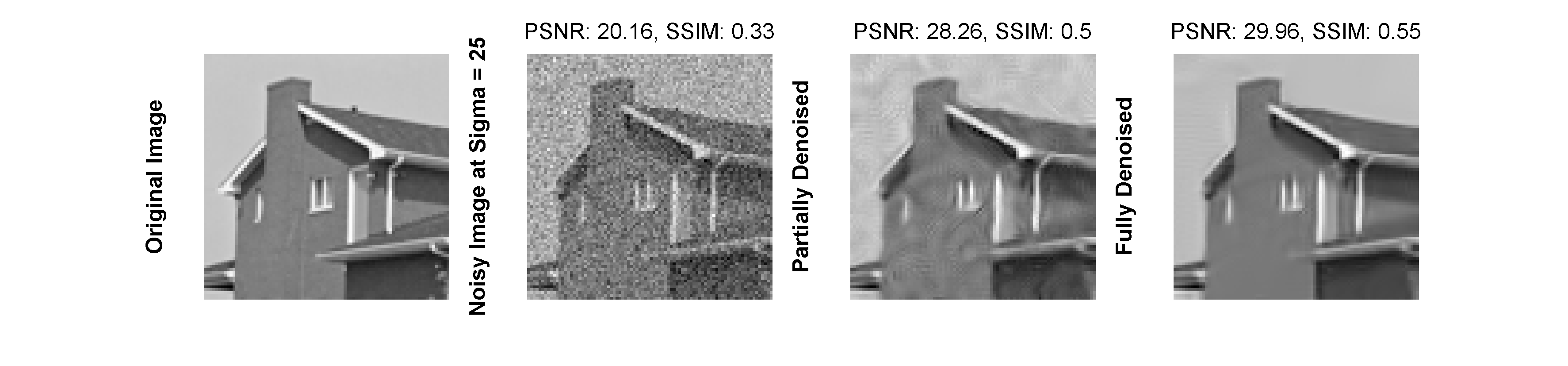}
\end{figure*}

\newpage
\begin{figure*}[t]
	\centering
	\includegraphics[width=1\linewidth]{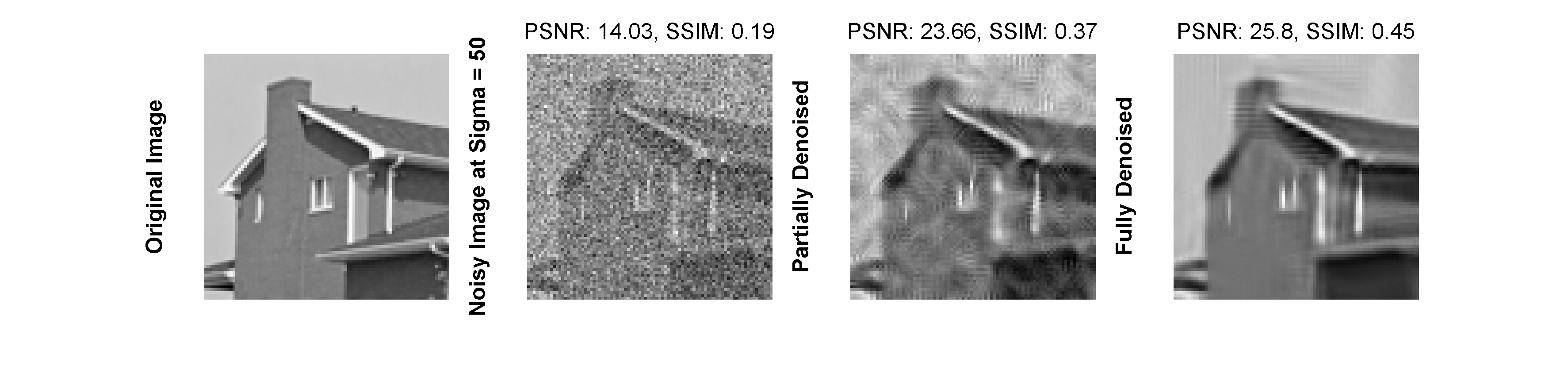}
	\includegraphics[width=1\linewidth]{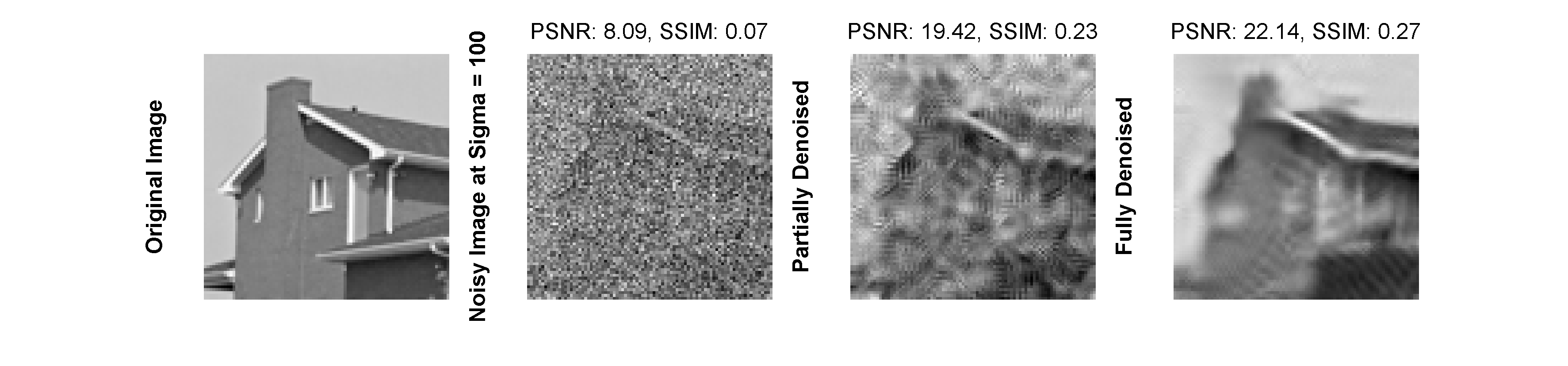}
	\includegraphics[width=1\linewidth]{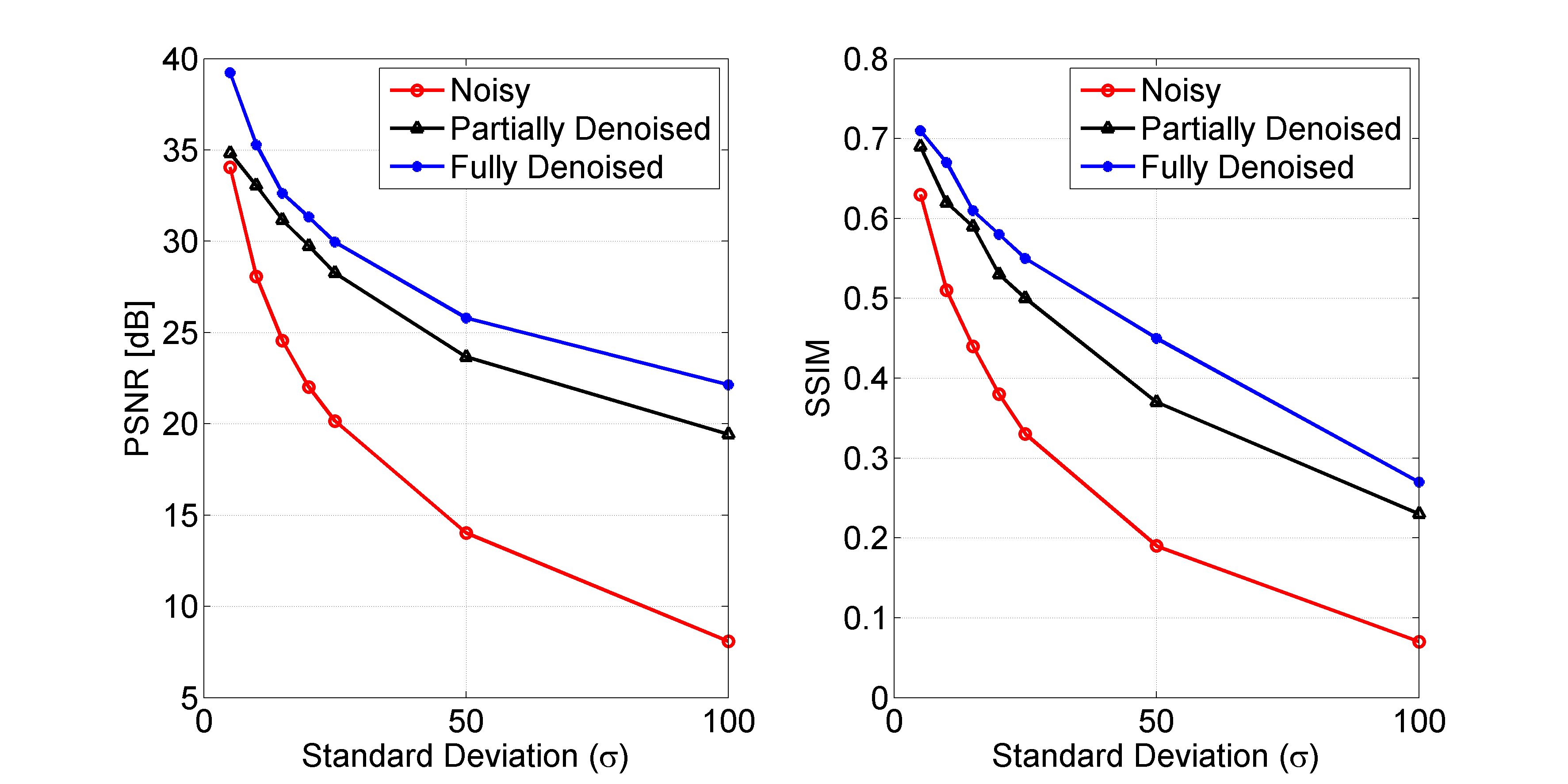}
	\caption{Denoising $256 \times256$ grayscale \textit{House} standard test data images over noise $\sigma =  [5,10,15,20,25,50,100]$ when received at a node $\mu_\alpha$. Each row represent an original image, a noisy image, a partially denoised, and a fully denoised image, respectively, corrupted by a specific level of additive white Gaussian noise (AWGN). The graphical results in the end show PSNR [dB] and SSIM results in the form of graphs.}
\end{figure*}

\newpage
\begin{figure*}[t]
	\centering
	\includegraphics[width=1\linewidth]{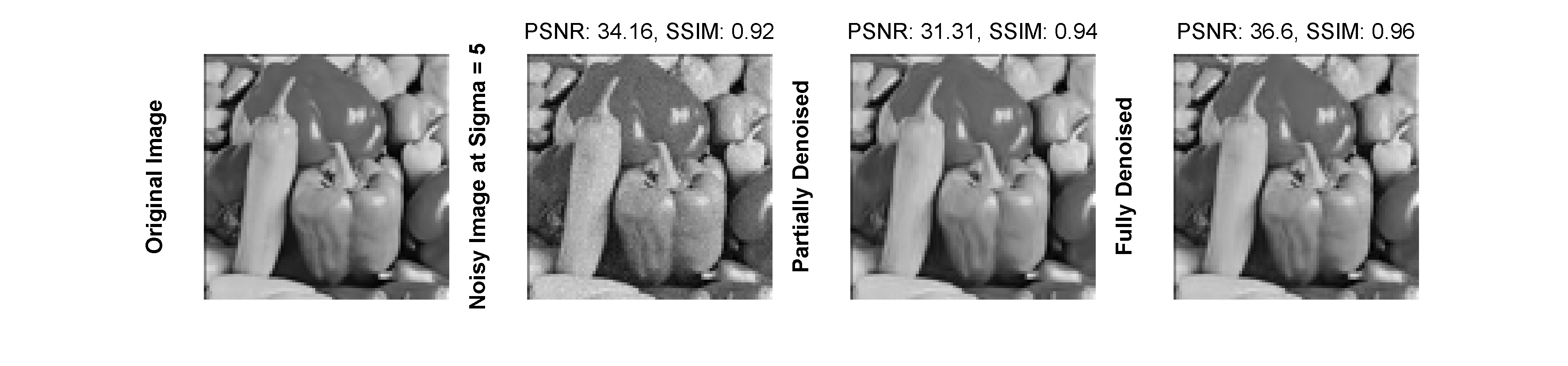}
	\includegraphics[width=1\linewidth]{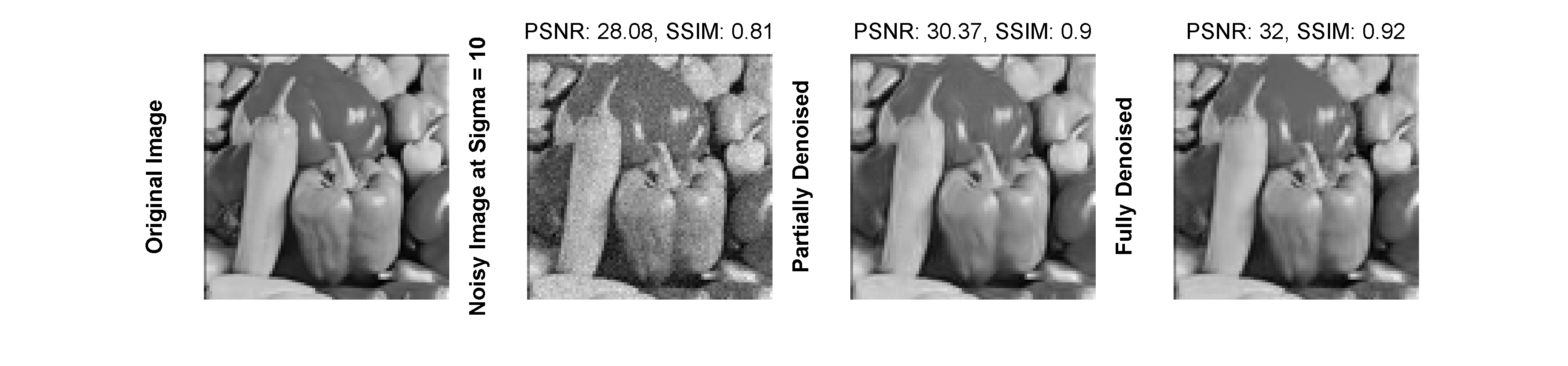}
	\includegraphics[width=1\linewidth]{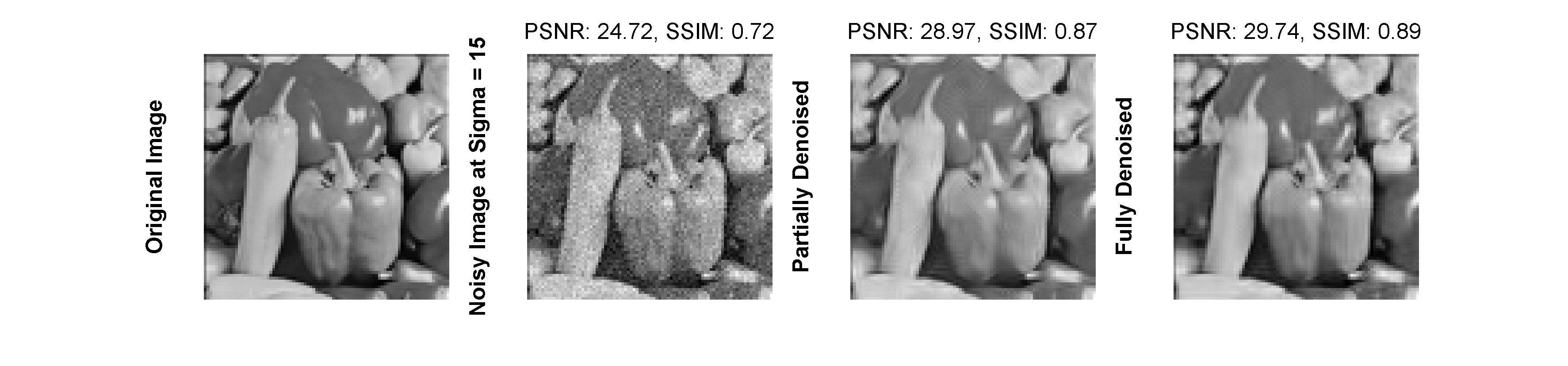}
	\includegraphics[width=1\linewidth]{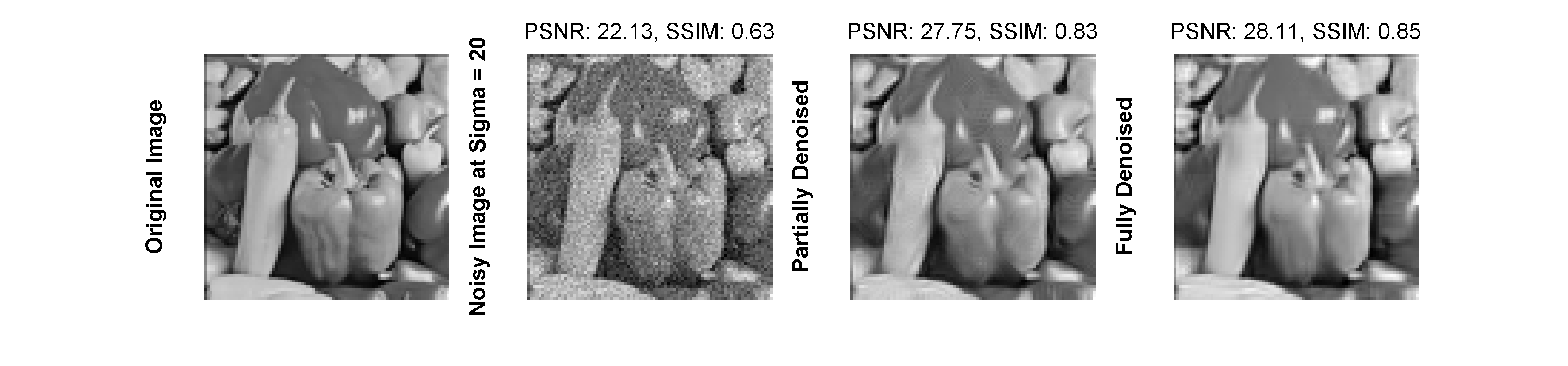}
	\includegraphics[width=1\linewidth]{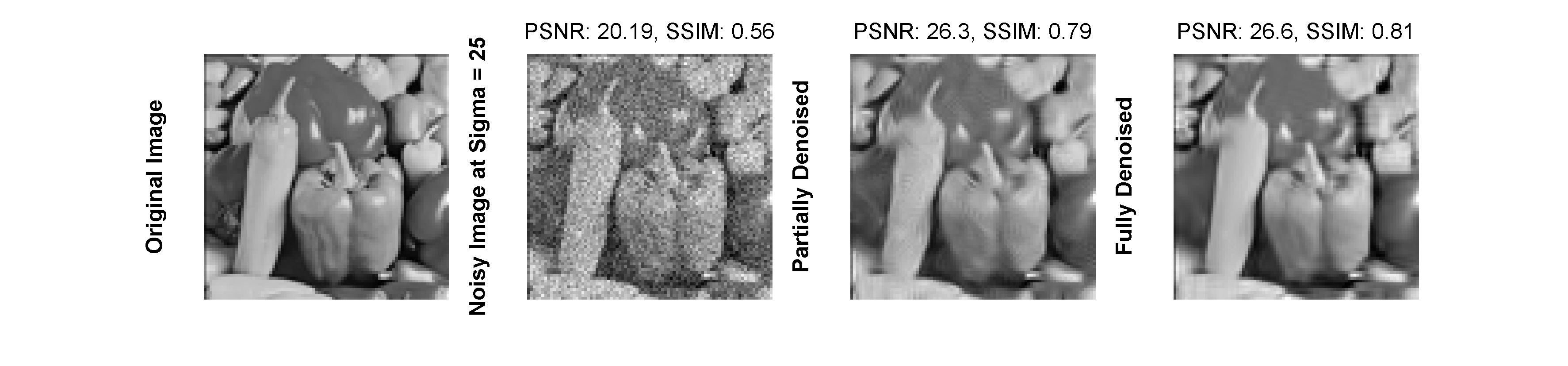}
\end{figure*}

\newpage
\begin{figure*}[t]
	\centering
	\includegraphics[width=1\linewidth]{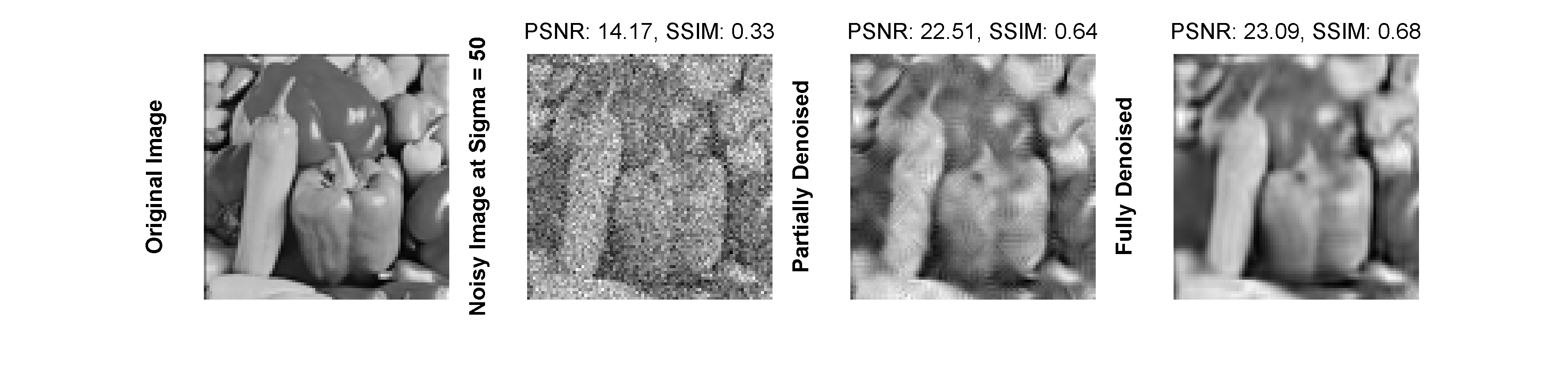}
	\includegraphics[width=1\linewidth]{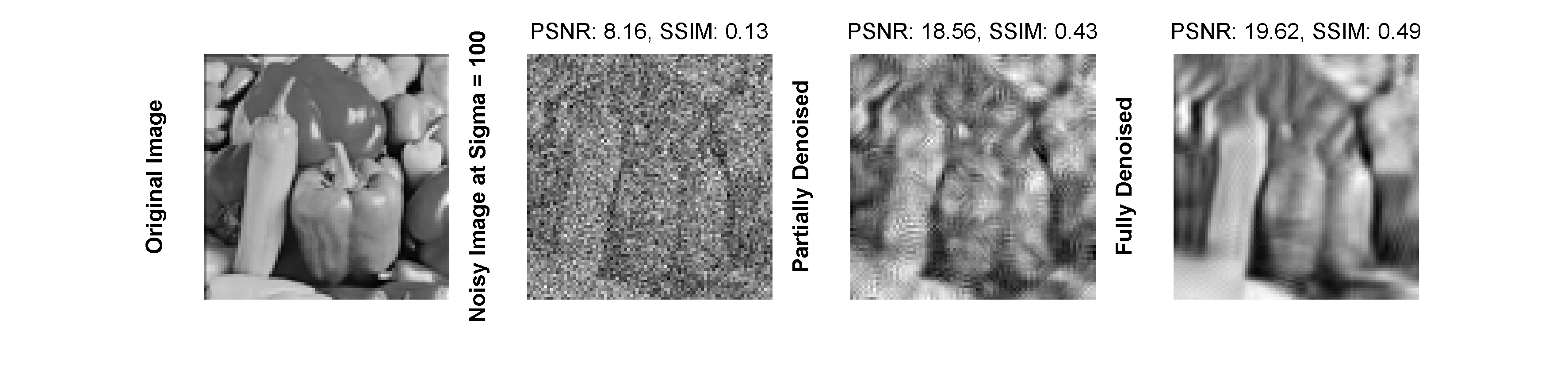}
	\includegraphics[width=1\linewidth]{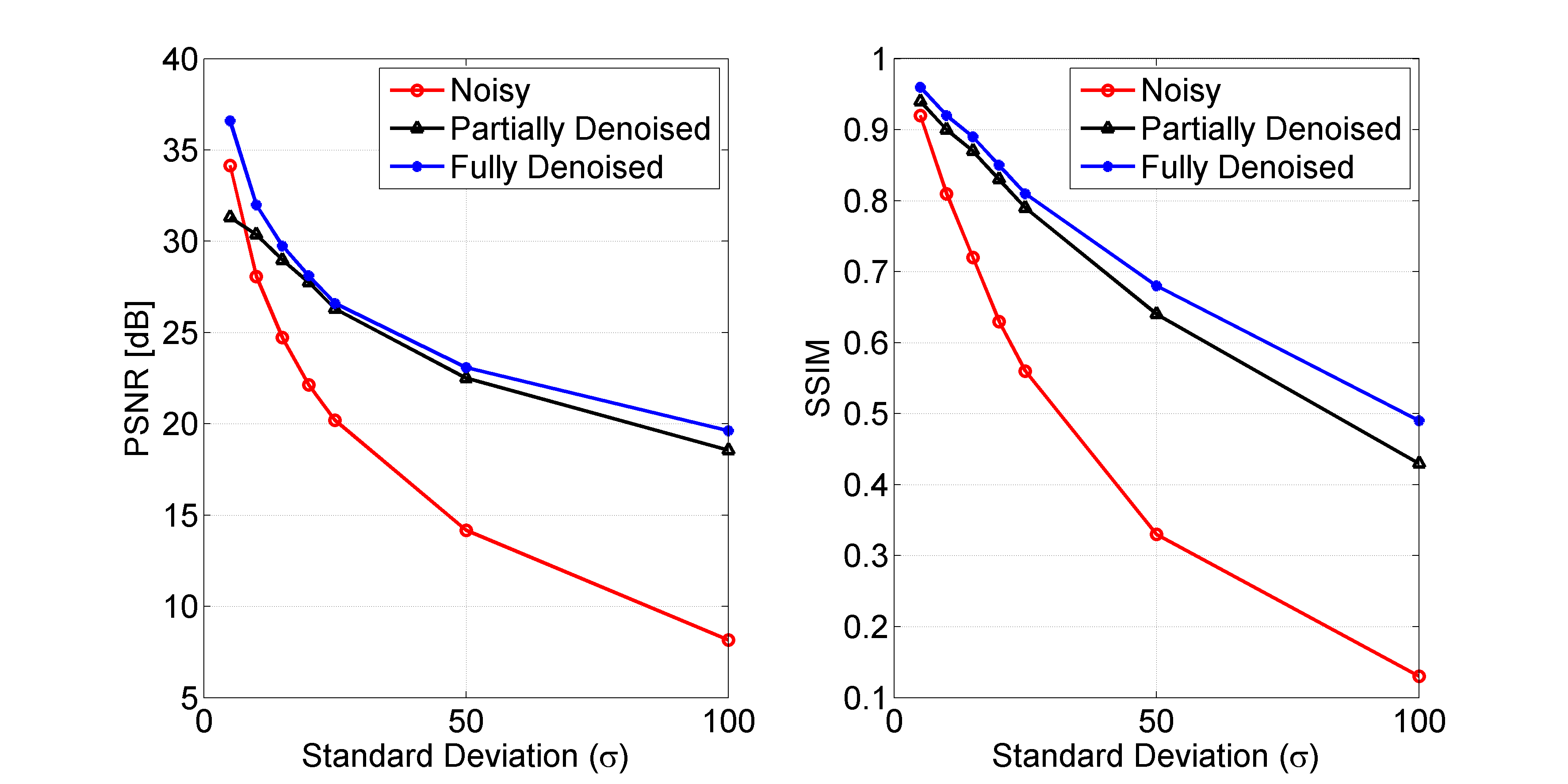}
	\caption{Denoising $256 \times256$ grayscale \textit{Peppers} standard test data images over noise $\sigma =  [5,10,15,20,25,50,100]$ when received at a node $\mu_\alpha$. Each row represent an original image, a noisy image, a partially denoised, and a fully denoised image, respectively, corrupted by a specific level of additive white Gaussian noise (AWGN). The graphical results in the end show PSNR [dB] and SSIM results in the form of graphs.}
\end{figure*}

\newpage
\begin{figure*}[t]
	\centering
	\includegraphics[width=1\linewidth]{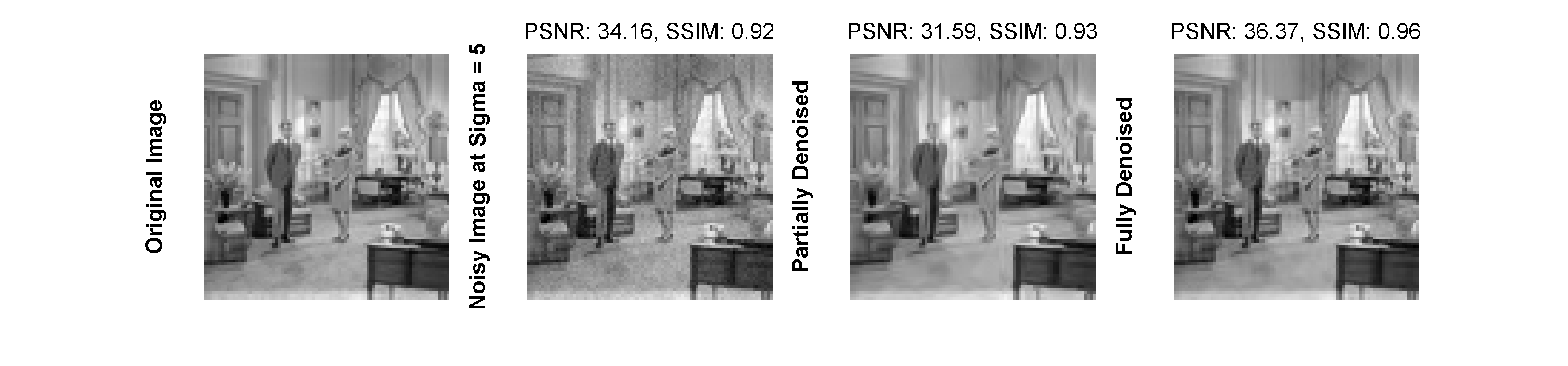}
	\includegraphics[width=1\linewidth]{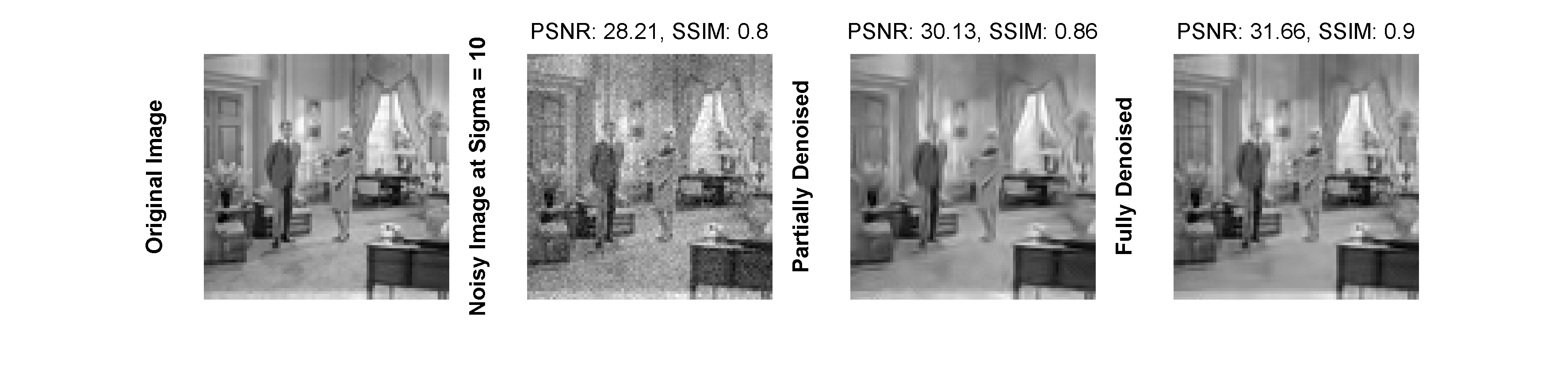}
	\includegraphics[width=1\linewidth]{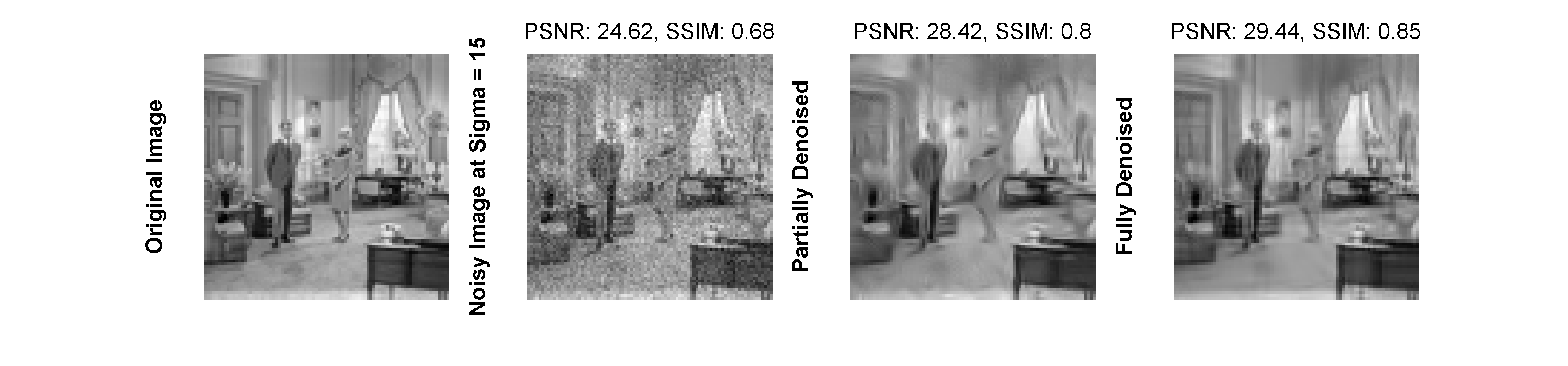}
	\includegraphics[width=1\linewidth]{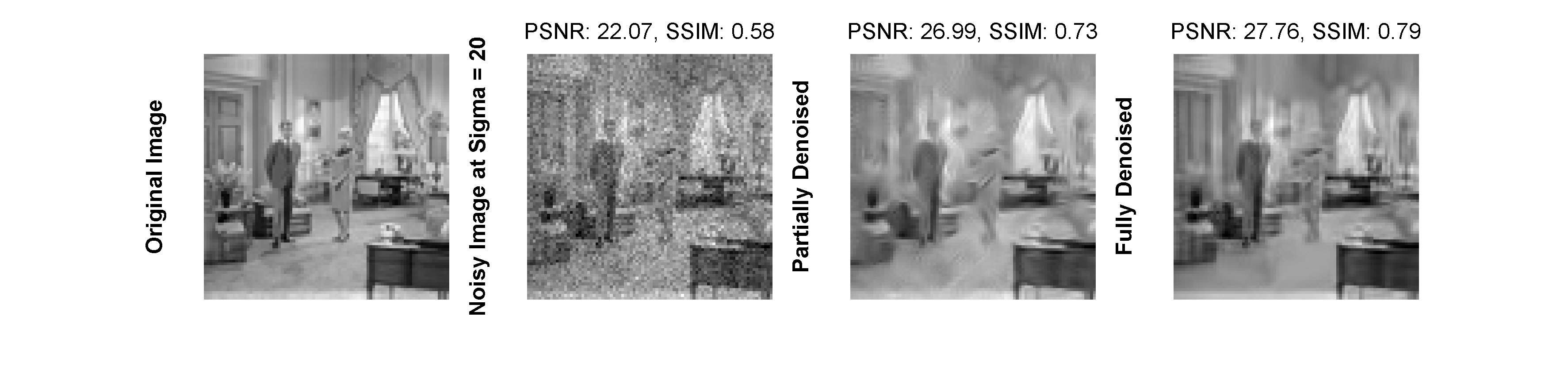}
	\includegraphics[width=1\linewidth]{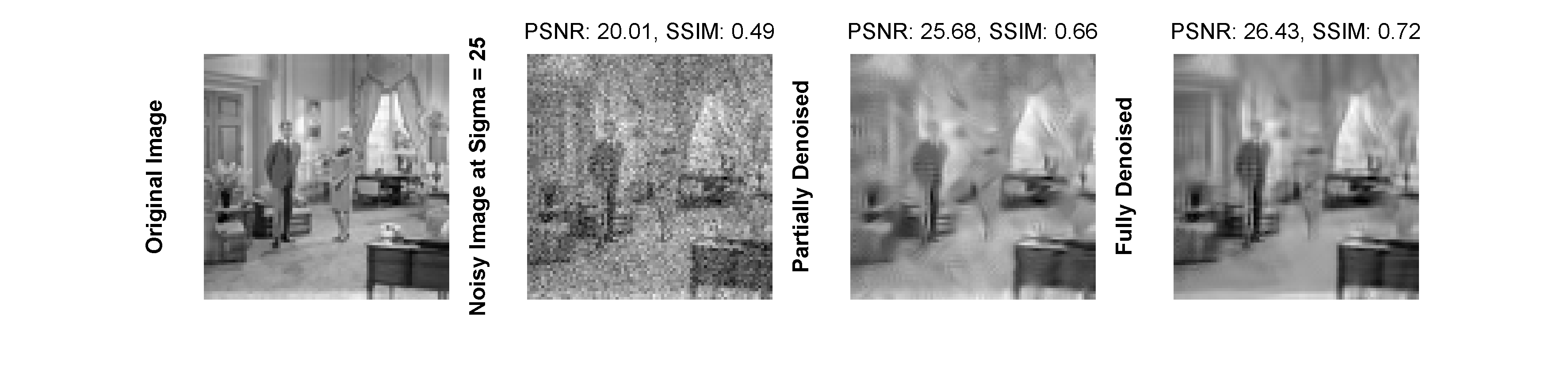}
\end{figure*}

\newpage
\begin{figure*}[t]
	\centering
	\includegraphics[width=1\linewidth]{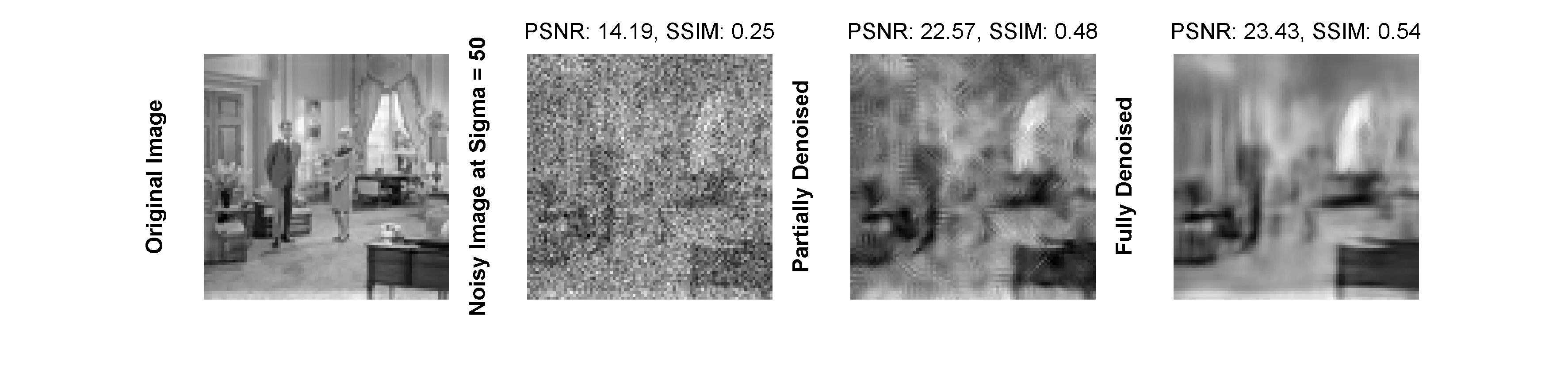}
	\includegraphics[width=1\linewidth]{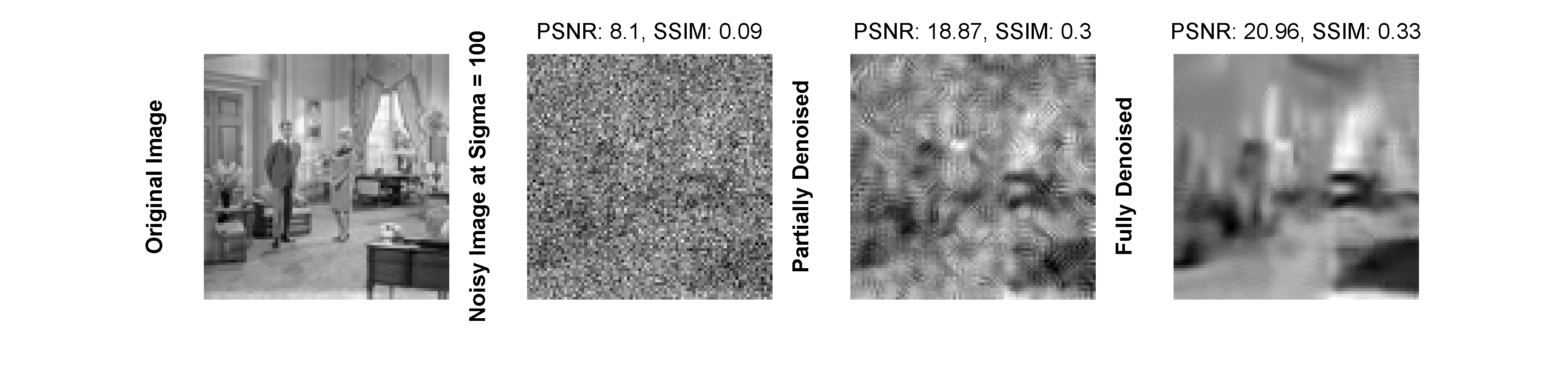}
	\includegraphics[width=1\linewidth]{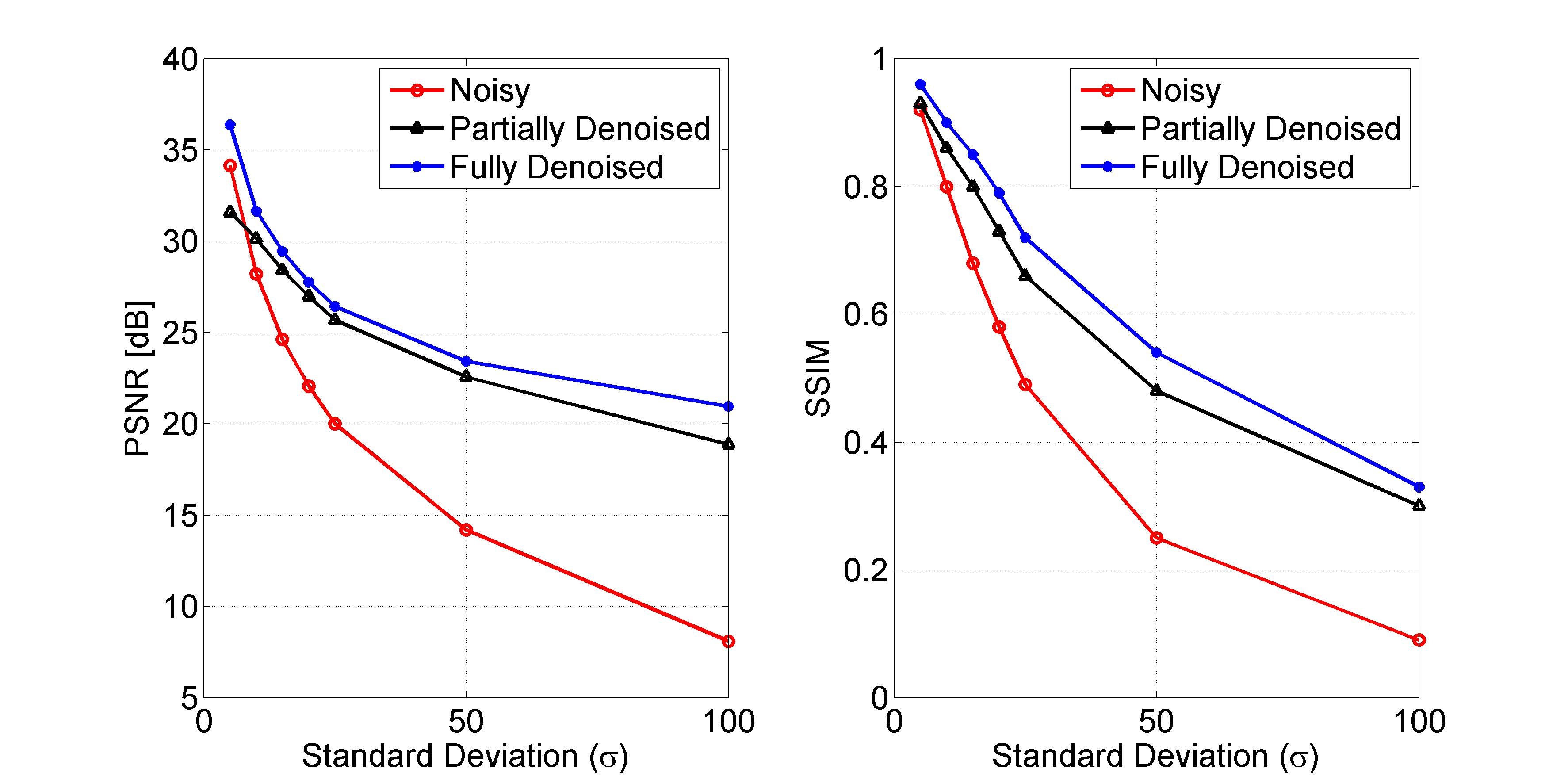}
	\caption{Denoising $256 \times256$ grayscale \textit{Living Room} standard test data images over noise $\sigma =  [5,10,15,20,25,50,100]$ when received at a node $\mu_\alpha$. Each row represent an original image, a noisy image, a partially denoised, and a fully denoised image, respectively, corrupted by a specific level of additive white Gaussian noise (AWGN). The graphical results in the end show PSNR [dB] and SSIM results in the form of graphs.}
\end{figure*}

\newpage
\begin{figure*}[t]
	\centering
	\includegraphics[width=1\linewidth]{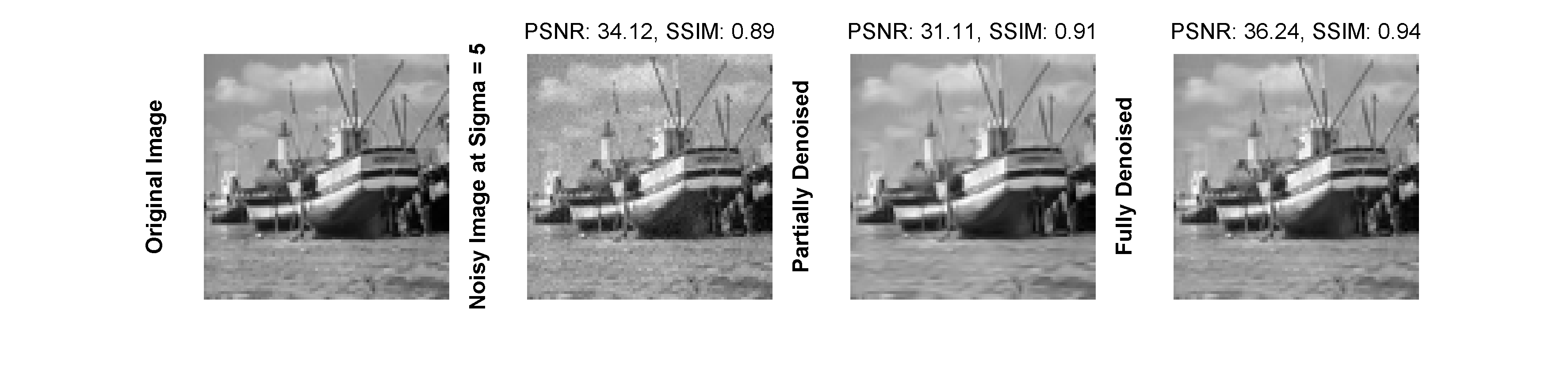}
	\includegraphics[width=1\linewidth]{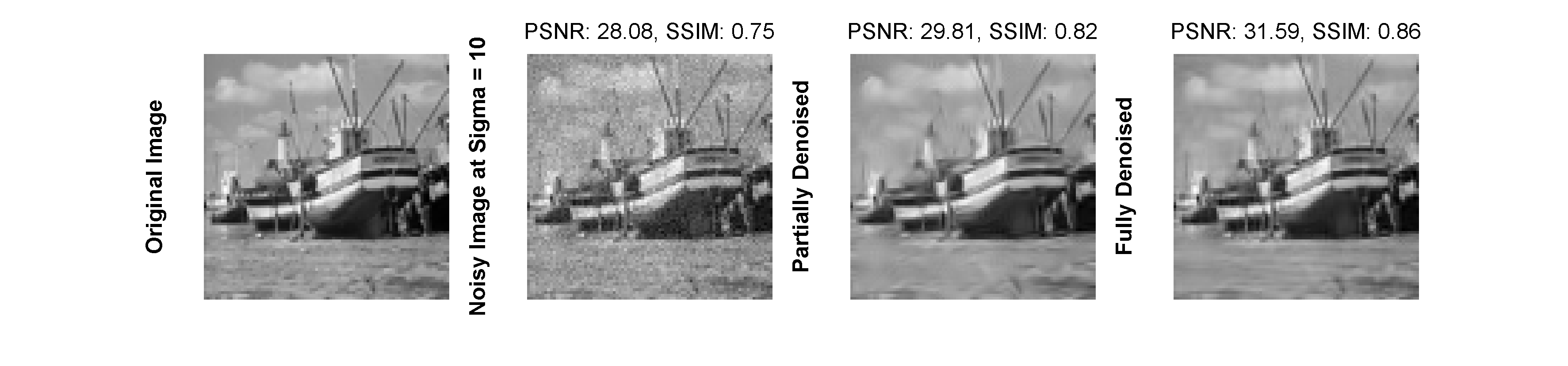}
	\includegraphics[width=1\linewidth]{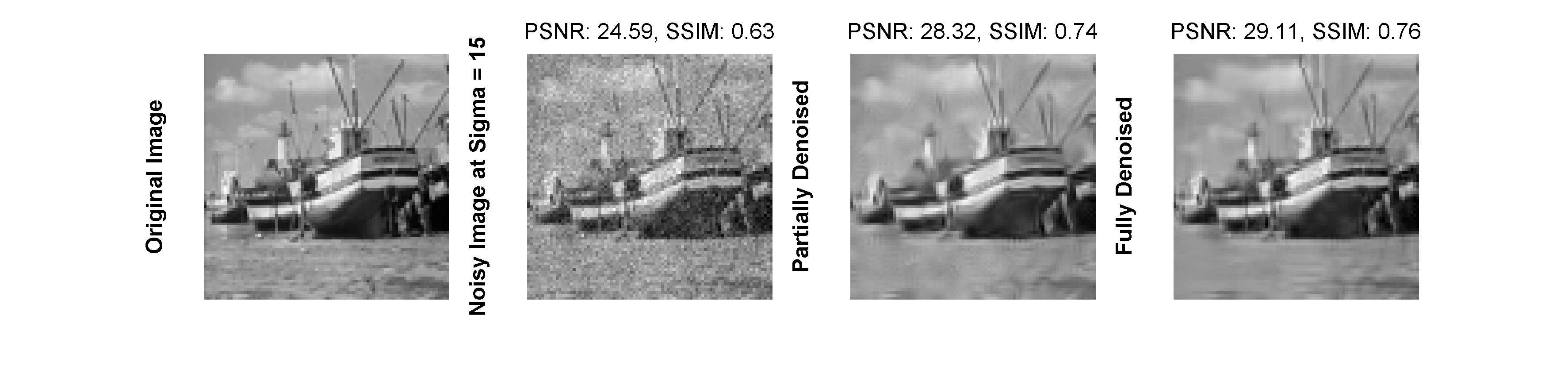}
	\includegraphics[width=1\linewidth]{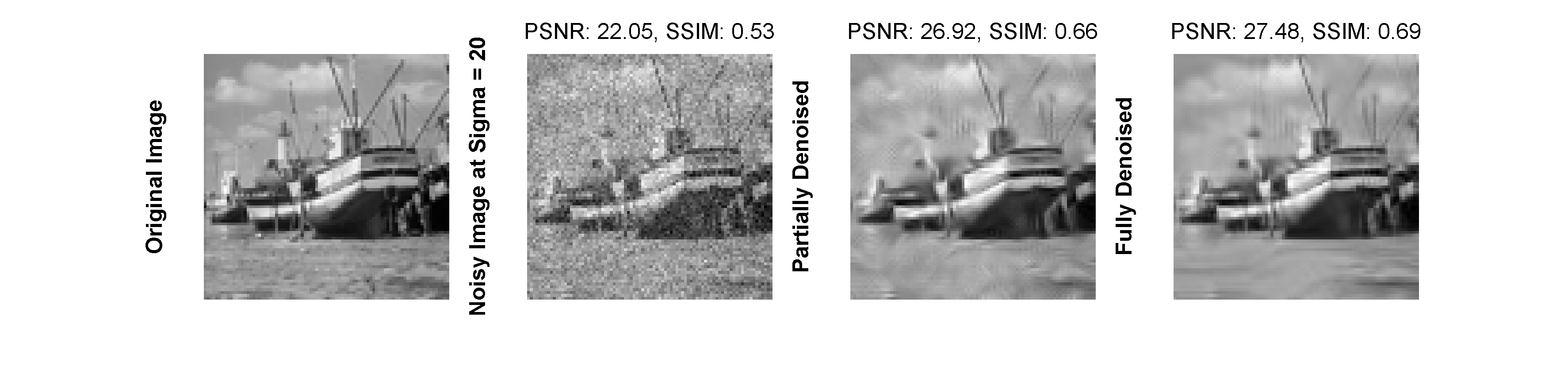}
	\includegraphics[width=1\linewidth]{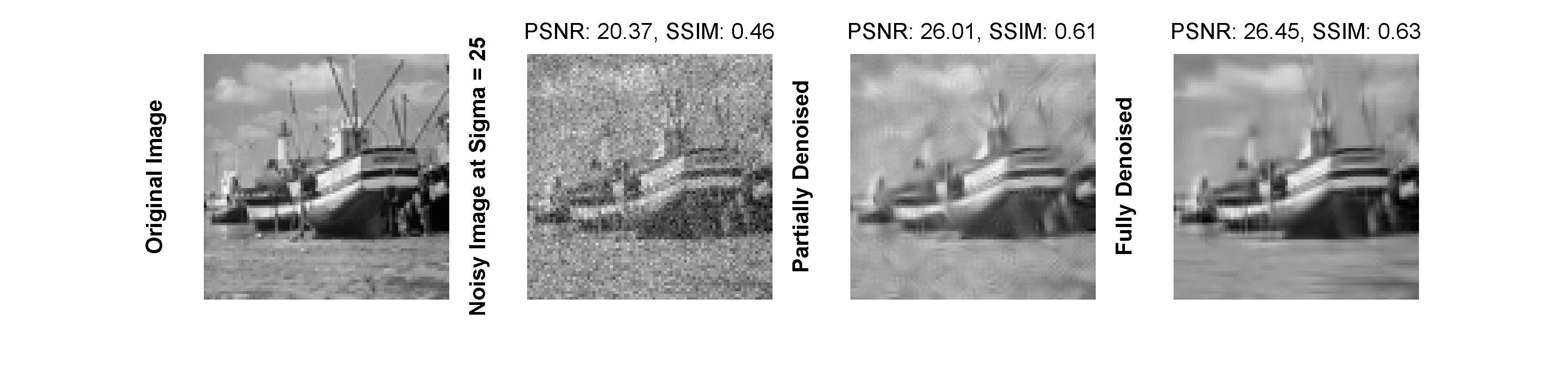}
\end{figure*}

\newpage
\begin{figure*}[t]
	\centering
	\includegraphics[width=1\linewidth]{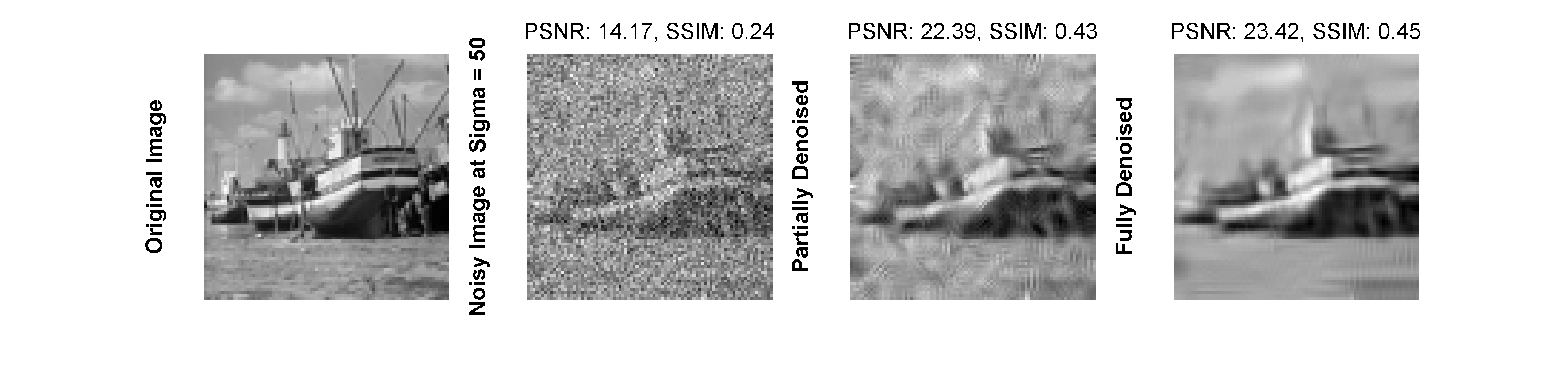}
	\includegraphics[width=1\linewidth]{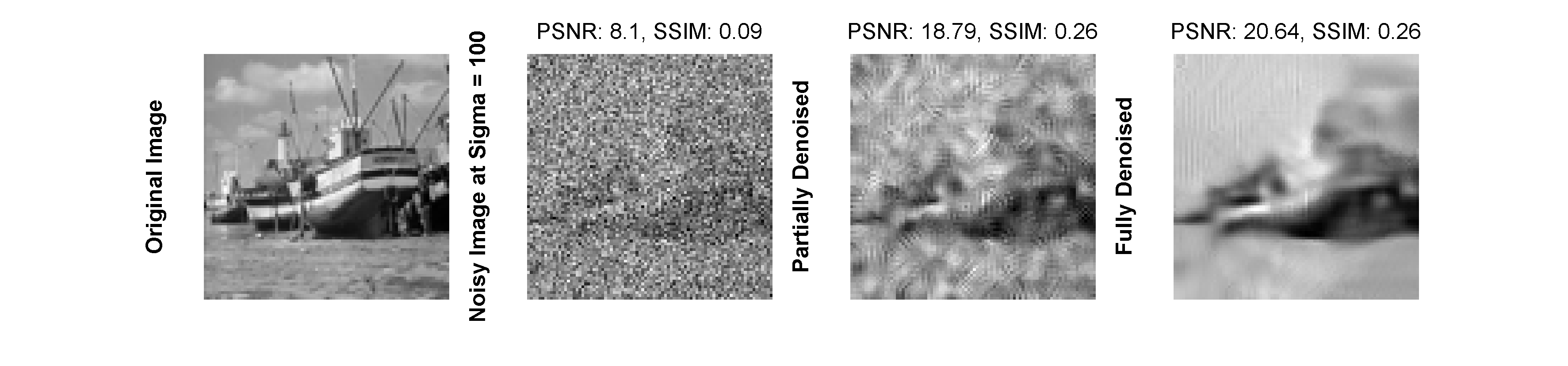}
	\includegraphics[width=1\linewidth]{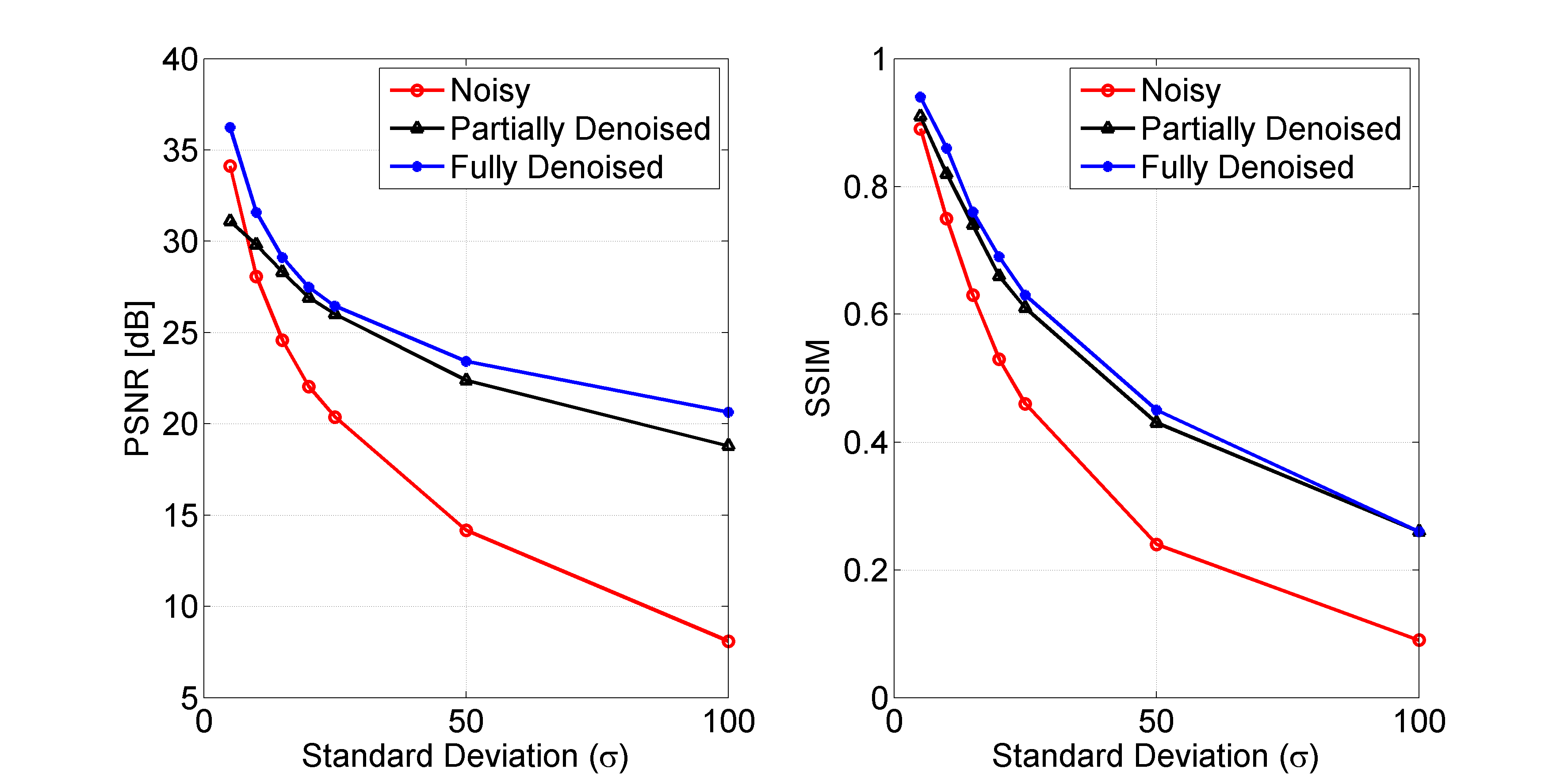}
	\caption{Denoising $256 \times256$ grayscale \textit{Boat} standard test data images over noise $\sigma =  [5,10,15,20,25,50,100]$ when received at a node $\mu_\alpha$. Each row represent an original image, a noisy image, a partially denoised, and a fully denoised image, respectively, corrupted by a specific level of additive white Gaussian noise (AWGN). The graphical results in the end show PSNR [dB] and SSIM results in the form of graphs.}
\end{figure*}

\newpage
\begin{figure*}[t]
	\centering
	\includegraphics[width=1\linewidth]{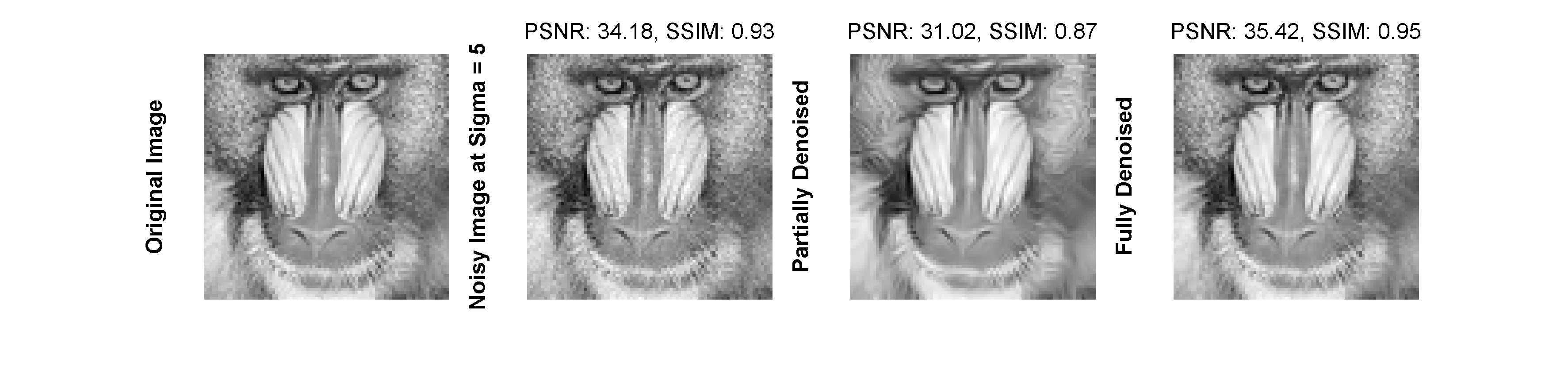}
	\includegraphics[width=1\linewidth]{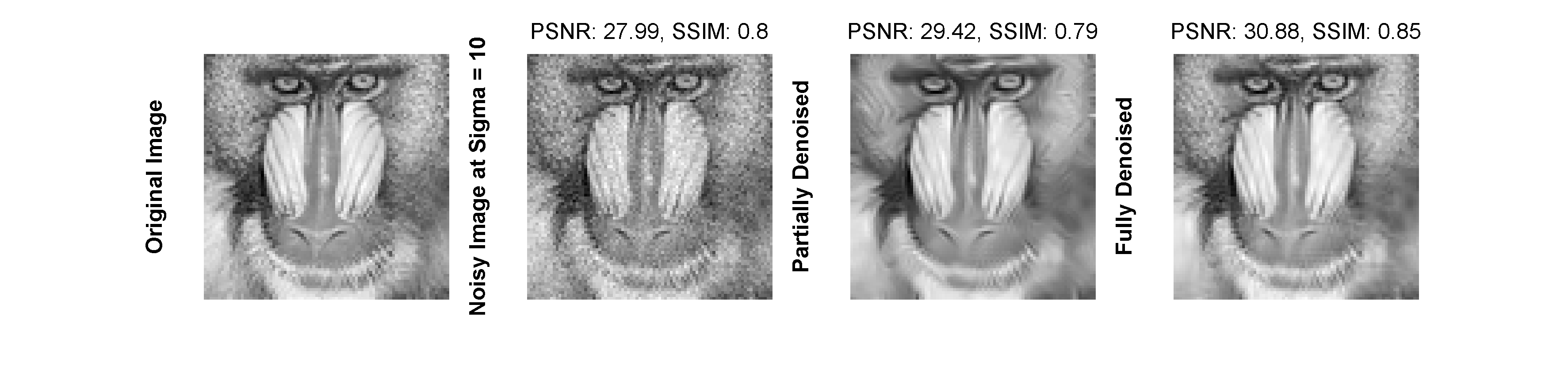}
	\includegraphics[width=1\linewidth]{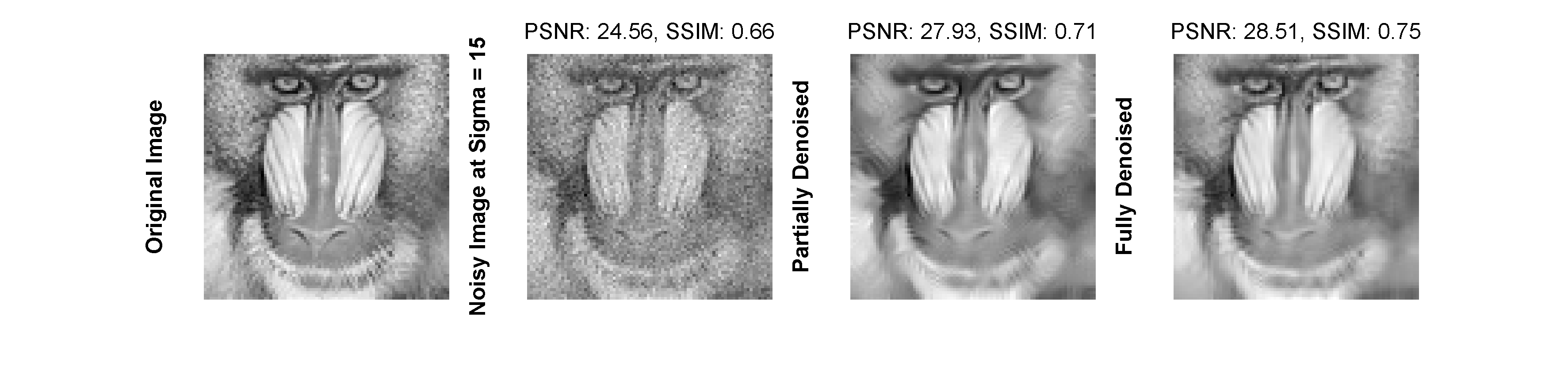}
	\includegraphics[width=1\linewidth]{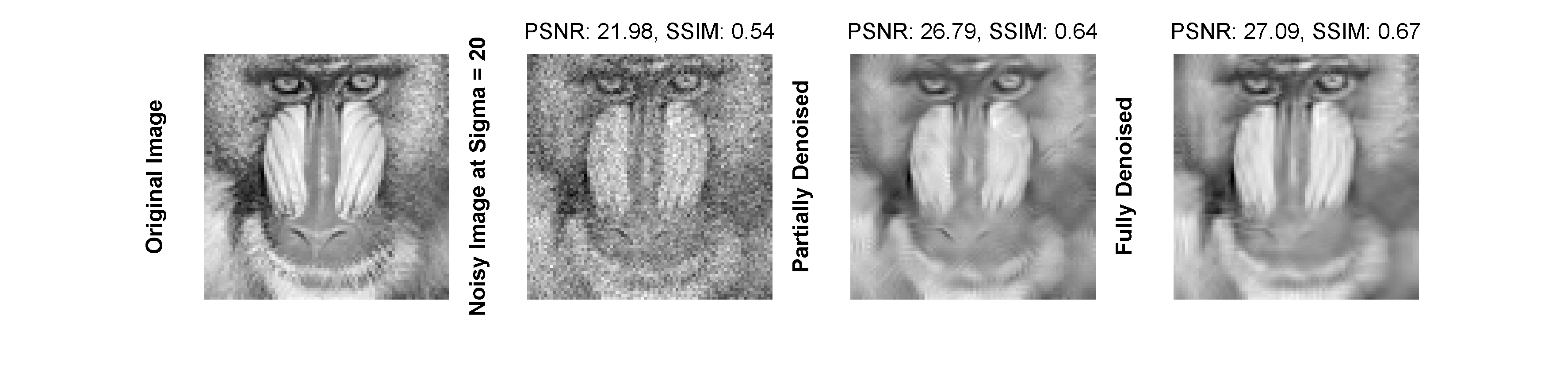}
	\includegraphics[width=1\linewidth]{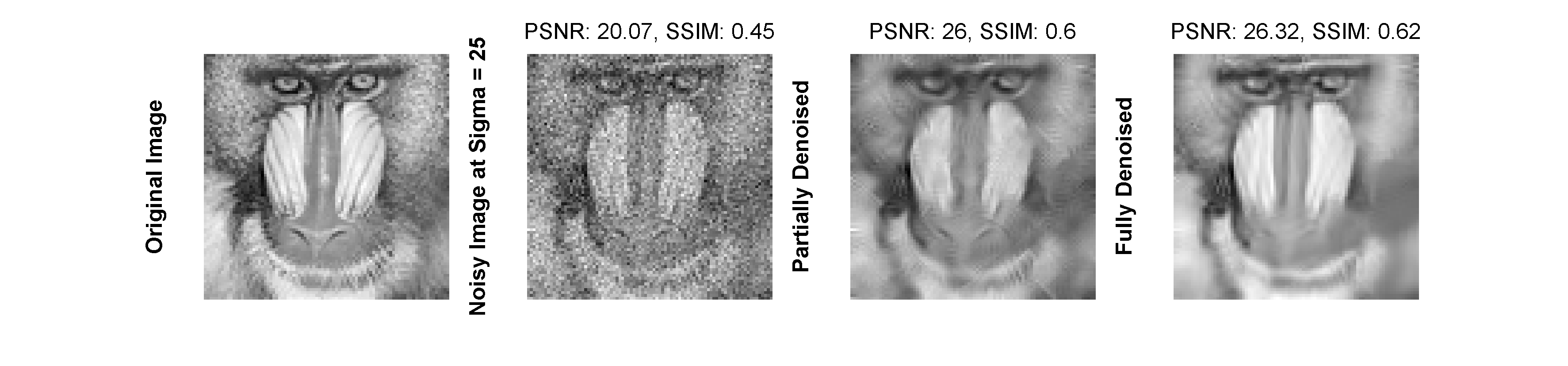}
\end{figure*}

\newpage
\begin{figure*}[t]
	\centering
	\includegraphics[width=1\linewidth]{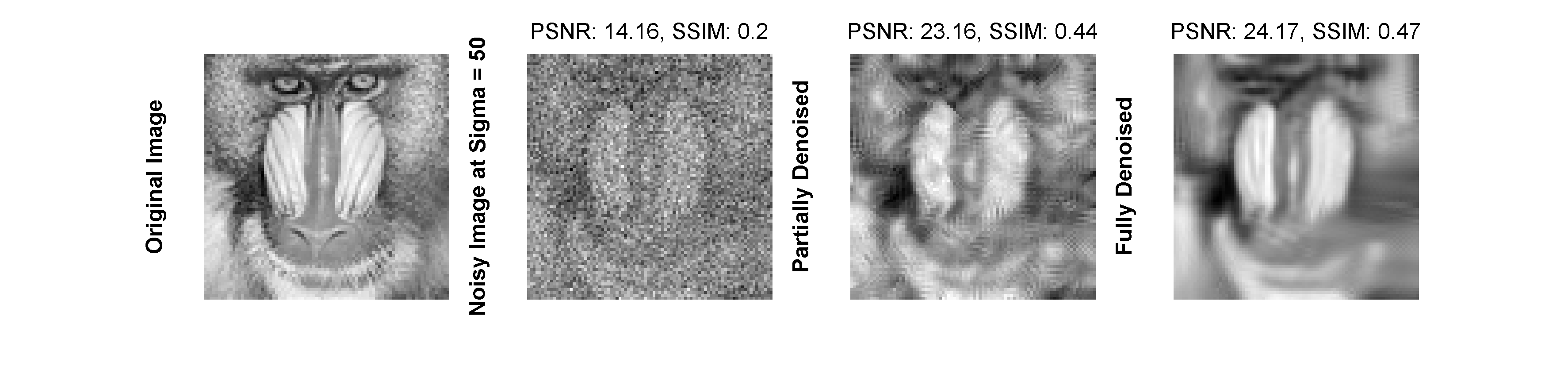}
	\includegraphics[width=1\linewidth]{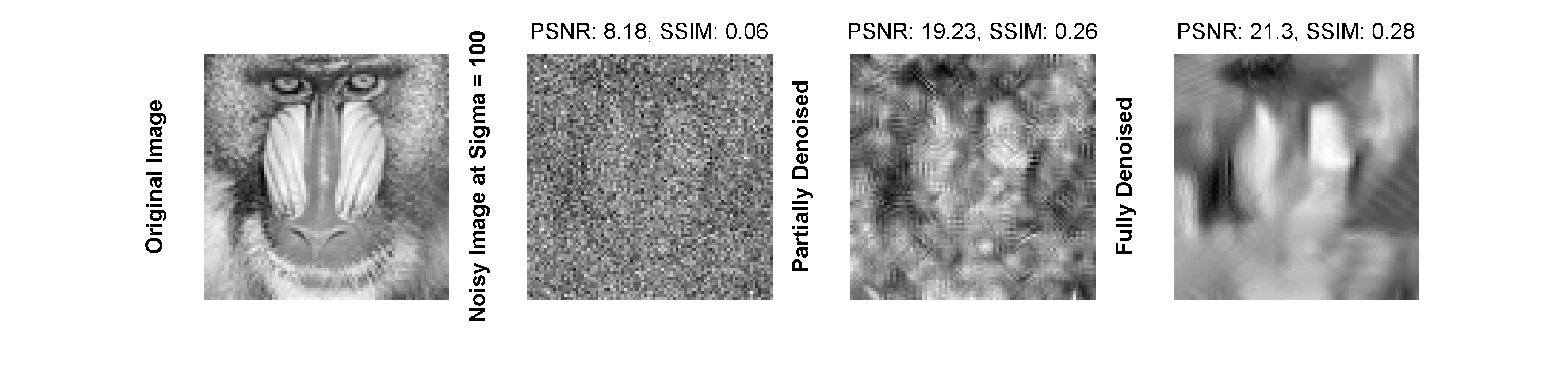}
	\includegraphics[width=1\linewidth]{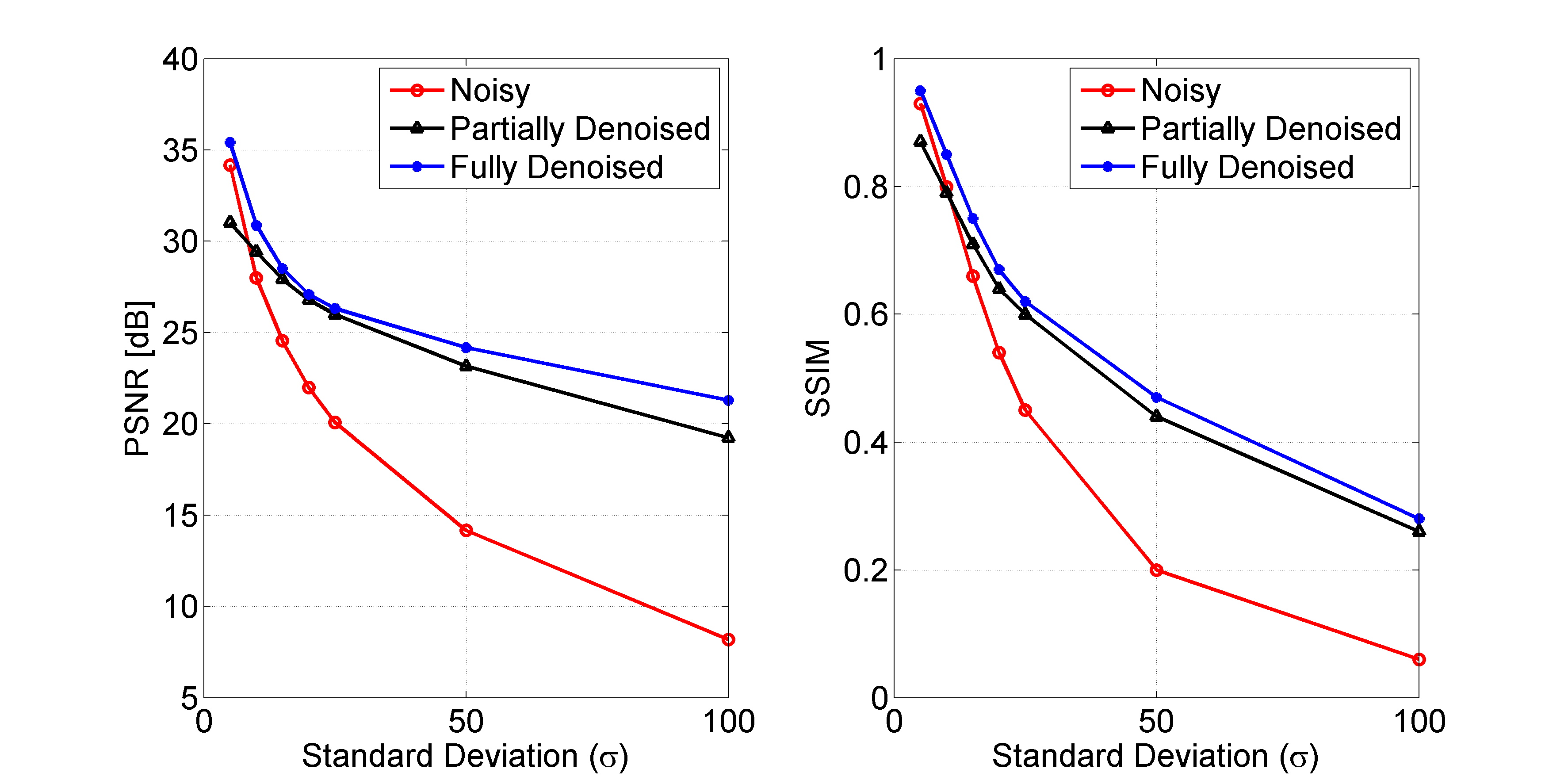}
	\caption{Denoising $256 \times256$ grayscale \textit{Mandrill} standard test data images over noise $\sigma =  [5,10,15,20,25,50,100]$ when received at a node $\mu_\alpha$. Each row represent an original image, a noisy image, a partially denoised, and a fully denoised image, respectively, corrupted by a specific level of additive white Gaussian noise (AWGN). The graphical results in the end show PSNR [dB] and SSIM results in the form of graphs.}
	\label{fig:denoised_end}
\end{figure*}

\newpage
\begin{landscape}
\begin{figure*}[t]
	\centering
	\includegraphics[width=1\linewidth]{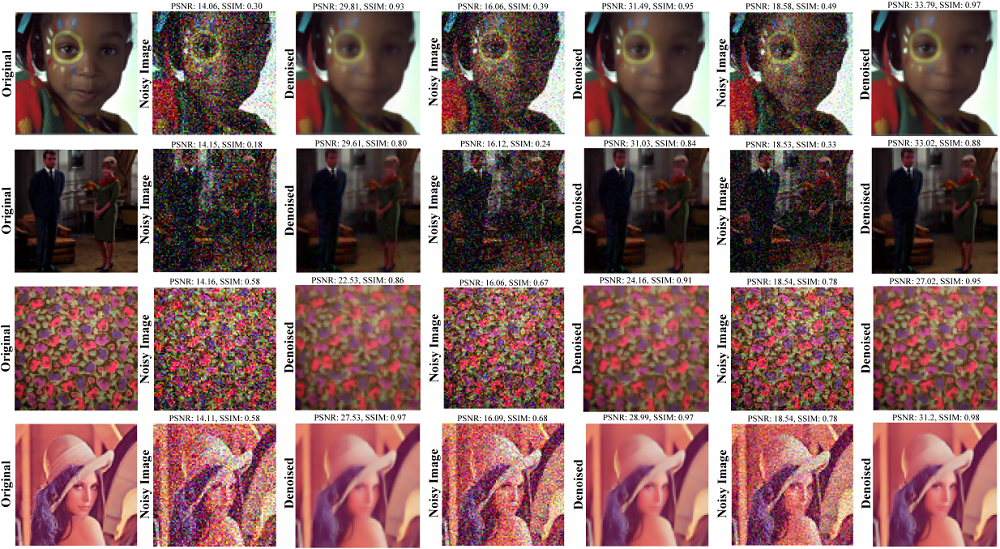}
	\caption{Denoising color images by the proposed color denoising method. 1st column: original images, 2nd and 3rd columns: noisy and denoised images at $\mathcal{N} (0,50)$, 4th and 5th columns: noisy and denoised images at $\mathcal{N} (0,40)$, 6th and 7th columns: noisy and denoised images at $\mathcal{N} (0,30)$,}
	\label{fig:color_denoised}
\end{figure*}
\end{landscape}
\clearpage

\newpage
\section{Conclusions}
\label{Conclusion}
In this work, we discussed our proposed framework that ensures energy-efficiency and data-denoising in a wireless sensor network. Our system is enriched with a layer-adaptive method that uses a 3-tier communication mechanism for effective and energy-efficient communication among the nodes ultimately minimizing the energy holes. Our presented mathematical coverage model effectively dealt with the formation of coverage holes thereby yielding a robust network. For combating noise in the data, we proposed a collaborative transform-domain based denoising algorithm to take care of the unwanted components. As shown with the help of many simulation results, our framework outperformed traditional algorithms by a significant margin, and provided a computationally-desirable algorithm for real-time applications.

\bibliographystyle{unsrt}
\bibliography{AINA_2018_Journal}

\begin{thebibliography}{10}

\bibitem{7110295}
Z.~Sheng, C.~Mahapatra, C.~Zhu, and V.~C.~M. Leung.
\newblock Recent advances in industrial wireless sensor networks toward
  efficient management in iot.
\newblock {\em IEEE Access}, 3:622--637, 2015.

\bibitem{7887698}
G.~Mois, S.~Folea, and T.~Sanislav.
\newblock Analysis of three iot-based wireless sensors for environmental
  monitoring.
\newblock {\em IEEE Transactions on Instrumentation and Measurement},
  66(8):2056--2064, Aug 2017.

\bibitem{umar2014enhancing}
A.~Umar, M.A. Hasnat, M.~Behzad, I.~Baseer, Z.A. Khan, U.~Qasim, and N.~Javaid.
\newblock On enhancing network reliability and throughput for critical-range
  based applications in uwsns.
\newblock {\em Procedia Computer Science}, 34:196--203, 2014.

\bibitem{7528266}
C.~Grumazescu, V.~A. Vlăduţă, and G.~Subaşu.
\newblock Wsn solutions for communication challenges in military live
  simulation environments.
\newblock In {\em International Conference on Communications}, pages 319--322,
  2016.

\bibitem{6883890}
O.~Salem, Y.~Liu, and A.~Mehaoua.
\newblock Anomaly detection in medical wsns using enclosing ellipse and
  chi-square distance.
\newblock In {\em IEEE International Conference on Communications (ICC)}, pages
  3658--3663, June 2014.

\bibitem{10.1007/978-3-662-48145-5_9}
S.~S. Jha and S.~B. Nair.
\newblock On a multi-agent distributed asynchronous intelligence-sharing and
  learning framework.
\newblock In Ngoc~Thanh Nguyen, editor, {\em Transactions on Computational
  Collective Intelligence XVIII}, pages 166--200, Berlin, Heidelberg, 2015.

\bibitem{behzad2014design}
M.~Behzad, A.~Sana, M.A. Khan, Z.~Walayat, U.~Qasim, Z.A. Khan, and N.~Javaid.
\newblock Design and development of a low cost ubiquitous tracking system.
\newblock {\em Procedia Computer Science}, 34:220--227, 2014.

\bibitem{sandhu2014mobility}
M.~M. Sandhu, M.~Akbar, M.~Behzad, N.~Javaid, Z.A. Khan, and U.~Qasim.
\newblock Mobility model for wbans.
\newblock In {\em Broadband and Wireless Computing, Communication and
  Applications (BWCCA), 2014 Ninth International Conference on}, pages
  155--160. IEEE, 2014.

\bibitem{sandhu2014reec}
M.M. Sandhu, M.~Akbar, M.~Behzad, N.~Javaid, Z.A. Khan, and U.~Qasim.
\newblock Reec: Reliable energy efficient critical data routing in wireless
  body area networks.
\newblock In {\em Broadband and Wireless Computing, Communication and
  Applications (BWCCA), 2014 Ninth International Conference on}, pages
  446--451. IEEE, 2014.

\bibitem{behzad2018m}
M.~Behzad.
\newblock M-behzad: Minimum distance based energy efficiency using hemisphere
  zoning with advanced divide-and-rule scheme for wireless sensor networks.
\newblock {\em arXiv preprint arXiv:1804.00898}, 2018.

\bibitem{behzad2018technology}
M.~Behzad, N.~Adnan, and S.A. Merchant.
\newblock Technology-embedded hybrid learning.
\newblock 2018.

\bibitem{7368292}
N.~Sibeko, P.~Mudali, O.~Oki, and A.~Alaba.
\newblock Performance evaluation of routing protocols in uniform and normal
  node distributions using inter-mesh wireless networks.
\newblock In {\em World Symposium on Computer Networks and Information Security
  (WSCNIS)}, pages 1--6, Sept 2015.

\bibitem{heinzelman2000application}
W.~B. Heinzelman.
\newblock {\em Application-specific protocol architectures for wireless
  networks}.
\newblock PhD thesis, Massachusetts Institute of Technology, 2000.

\bibitem{7016049}
M.~Behzad, N.~Javaid, A.~Sana, M.~A. Khan, N.~Saeed, Z.~A. Khan, and U.~Qasim.
\newblock Tsddr: Threshold sensitive density controlled divide and rule routing
  protocol for wireless sensor networks.
\newblock In {\em Ninth International Conference on Broadband and Wireless
  Computing, Communication and Applications}, pages 78--83, Nov 2014.

\bibitem{saleem2014iddr}
F.~Saleem, Y.~Moeen, M.~Behzad, M.A. Hasnat, Z.A. Khan, U.~Qasim, and
  N.~Javaid.
\newblock Iddr: improved density controlled divide-and-rule scheme for energy
  efficient routing in wireless sensor networks.
\newblock {\em Procedia Computer Science}, 34:212--219, 2014.

\bibitem{7920983}
M.~Behzad and Y.~Ge.
\newblock Performance optimization in wireless sensor networks: A novel
  collaborative compressed sensing approach.
\newblock In {\em IEEE 31st International Conference on Advanced Information
  Networking and Applications (AINA)}, pages 749--756, March 2017.

\bibitem{behzad2017distributed}
M.~Behzad, M.S. Javaid, M.A. Parahca, and S.~Khan.
\newblock Distributed pca and consensus based energy efficient routing protocol
  for wsns.
\newblock {\em Journal of Information Science \& Engineering}, 33(5), 2017.

\bibitem{behzad2018AINA}
M.~Behzad, M.~Abdullah, M.T. Hassan, Y.~Ge, and M.A. Khan.
\newblock Layer-adaptive communication and collaborative transformed-domain
  representations to optimize performance in next-generation wsns.
\newblock In {\em IEEE 32nd International Conference on Advanced Information
  Networking and Applications}, pages 101--108, 2018.

\bibitem{8109448}
A.~Jurenoks and L.~Novickis.
\newblock Analysis of wireless sensor network structure and life time affecting
  factors.
\newblock In {\em Communication and Information Technologies (KIT)}, pages
  1--6, Oct 2017.

\bibitem{behzad2017layer}
M.~Behzad, M.~Abdullah, M.T. Hassan, Y.~Ge, and M.A. Khan.
\newblock Layer-adaptive communication and collaborative transformed-domain
  representations for performance optimization in wsns.
\newblock {\em arXiv preprint arXiv:1712.04259}, 2017.

\bibitem{926982}
W.~B. Heinzelman, A.~Chandrakasan, and H.~Balakrishnan.
\newblock Energy-efficient communication protocol for wireless microsensor
  networks.
\newblock In {\em Proceedings of the 33rd Annual Hawaii International
  Conference on System Sciences}, Jan 2000.

\bibitem{925197}
A.~Manjeshwar and D.~P. Agrawal.
\newblock Teen: a routing protocol for enhanced efficiency in wireless sensor
  networks.
\newblock In {\em Proceedings 15th International Parallel and Distributed
  Processing Symposium (IPDPS)}, pages 2009--2015, April 2001.

\bibitem{smaragdakis2004sep}
G.~Smaragdakis, I.~Matta, and A.~Bestavros.
\newblock Sep: A stable election protocol for clustered heterogeneous wireless
  sensor networks.
\newblock Technical report, Boston University Computer Science Department,
  2004.

\bibitem{QING20062230}
L.~Qing, Q.~Zhu, and M.~Wang.
\newblock Design of a distributed energy-efficient clustering algorithm for
  heterogeneous wireless sensor networks.
\newblock {\em Computer Communications}, 29(12):2230--2237, 2006.

\bibitem{1696382}
A.~Khadivi and M.~Shiva.
\newblock Ftpasc: A fault tolerant power aware protocol with static clustering
  for wireless sensor networks.
\newblock In {\em IEEE International Conference on Wireless and Mobile
  Computing, Networking and Communications}, pages 397--401, June 2006.

\bibitem{SWQ28921546.OW12N}
I.~Azam, A.~Majid, I.~Ahmad, U.~Shakeel, H.~Maqsood, Z.~A. Khan, U.~Qasim, and
  N.~Javaid.
\newblock Seec: Sparsity-aware energy efficient clustering protocol for
  underwater wireless sensor networks.
\newblock In {\em IEEE 30th International Conference on Advanced Information
  Networking and Applications (AINA)}, pages 352--361, March 2016.

\bibitem{5136647}
F.~Bajaber and I.~Awan.
\newblock Centralized dynamic clustering for wireless sensor network.
\newblock In {\em International Conference on Advanced Information Networking
  and Applications (AINA) Workshops}, pages 193--198, 2009.

\bibitem{4809826}
G.~S. Tomar and S.~Verma.
\newblock Dynamic multi-level hierarchal clustering approach for wireless
  sensor networks.
\newblock In {\em 11th International Conference on Computer Modelling and
  Simulation}, pages 563--567, March 2009.

\bibitem{7822956}
M.~K. Naeem, M.~Patwary, and M.~Abdel-Maguid.
\newblock Universal and dynamic clustering scheme for energy constrained
  cooperative wireless sensor networks.
\newblock {\em IEEE Access}, 5:12318--12337, 2017.

\bibitem{7365420}
D.~Jia, H.~Zhu, S.~Zou, and P.~Hu.
\newblock Dynamic cluster head selection method for wireless sensor network.
\newblock {\em IEEE Sensors Journal}, 16(8):2746--2754, April 2016.

\bibitem{S2RT4387429.OW12N}
A.~Ahmad, K.~Latif, N.~Javaidl, Z.~A. Khan, and U.~Qasim.
\newblock Density controlled divide-and-rule scheme for energy efficient
  routing in wireless sensor networks.
\newblock In {\em 26th IEEE Canadian Conference on Electrical and Computer
  Engineering (CCECE)}, pages 1--4, May 2013.

\bibitem{6531736}
A.~S.~K. Mammu, A.~Sharma, U.~Hernandez-Jayo, and N.~Sainz.
\newblock A novel cluster-based energy efficient routing in wireless sensor
  networks.
\newblock In {\em IEEE 27th International Conference on Advanced Information
  Networking and Applications (AINA)}, pages 41--47, March 2013.

\bibitem{7952375}
M.~Behzad, M.~Masood, T.~Ballal, M.~Shadaydeh, and T.~Y. Al-Naffouri.
\newblock Image denoising via collaborative support-agnostic recovery.
\newblock In {\em 2017 IEEE International Conference on Acoustics, Speech and
  Signal Processing (ICASSP)}, pages 1343--1347, March 2017.

\bibitem{krause2015shape}
A.F. Krause, N.~Harischandra, and V.~D{\"u}rr.
\newblock Shape recognition through tactile contour tracing.
\newblock In {\em Transactions on Computational Collective Intelligence XX},
  pages 54--77. Springer, 2015.

\bibitem{6853394}
K.~He and J.~Sun.
\newblock Image completion approaches using the statistics of similar patches.
\newblock {\em IEEE Transactions on Pattern Analysis and Machine Intelligence},
  36(12):2423--2435, Dec 2014.

\bibitem{7051524}
H.~Liu, R.~Xiong, S.~Ma, X.~Fan, and W.~Gao.
\newblock Gradient based image/video softcast with grouped-patch collaborative
  reconstruction.
\newblock In {\em IEEE Visual Communications and Image Processing Conference},
  pages 141--144, Dec 2014.

\bibitem{7351075}
M.~Wang, J.~Yu, and W.~Sun.
\newblock Group-based hyperspectral image denoising using low rank
  representation.
\newblock In {\em IEEE International Conference on Image Processing (ICIP)},
  pages 1623--1627, Sept 2015.

\bibitem{7457822}
W.~Yang, J.~Liu, S.~Yang, and Z.~Quo.
\newblock Image super-resolution via nonlocal similarity and group structured
  sparse representation.
\newblock In {\em IEEE Visual Communications and Image Processing}, pages 1--4,
  Dec 2015.

\bibitem{bahrami2016reconstruction}
K.~Bahrami, F.~Shi, X.~Zong, H.W. Shin, H.~An, and D.~Shen.
\newblock Reconstruction of 7t-like images from 3t mri.
\newblock {\em IEEE transactions on medical imaging}, 35(9):2085--2097, 2016.

\bibitem{behzad2018image}
M.~Behzad.
\newblock Image denoising via collaborative dual-domain patch filtering.
\newblock {\em arXiv preprint arXiv:1805.00472}, 2018.

\bibitem{VanHulle2012}
Marc~M. Van~Hulle.
\newblock {\em Self-organizing Maps}, pages 585--622.
\newblock Springer Berlin Heidelberg, Berlin, Heidelberg, 2012.

\bibitem{1056457}
A.~Gersho.
\newblock On the structure of vector quantizers.
\newblock {\em IEEE Transactions on Information Theory}, 28(2):157--166, Mar
  1982.

\bibitem{hoppner1999fuzzy}
F.~H{\"o}ppner.
\newblock {\em Fuzzy cluster analysis: methods for classification, data
  analysis and image recognition}.
\newblock John Wiley \& Sons, 1999.

\bibitem{jain1999data}
A.K. Jain, M.N. Murty, and P.J. Flynn.
\newblock Data clustering: a review.
\newblock {\em ACM computing surveys (CSUR)}, 31(3):264--323, 1999.

\bibitem{1467423}
A.~Buades, B.~Coll, and J.-M. Morel.
\newblock A non-local algorithm for image denoising.
\newblock In {\em IEEE Computer Society Conference on Computer Vision and
  Pattern Recognition}, volume~2, pages 60--65, June 2005.

\bibitem{4271520}
K.~Dabov, A.~Foi, V.~Katkovnik, and K.~Egiazarian.
\newblock Image denoising by sparse 3-d transform-domain collaborative
  filtering.
\newblock {\em IEEE Transactions on Image Processing}, 16(8):2080--2095, Aug
  2007.

\bibitem{7005524}
G.~Liu, H.~Zhong, and L.~Jiao.
\newblock Comparing noisy patches for image denoising: A double noise
  similarity model.
\newblock {\em IEEE Transactions on Image Processing}, 24(3):862--872, March
  2015.

\bibitem{7444121}
K.~Panetta, L.~Bao, and S.~Agaian.
\newblock Sequence-to-sequence similarity-based filter for image denoising.
\newblock {\em IEEE Sensors Journal}, 16(11):4380--4388, June 2016.

\end{thebibliography}

\begin{wrapfigure}{l}{25mm}
	\includegraphics[width=1in,height=1.25in,clip]{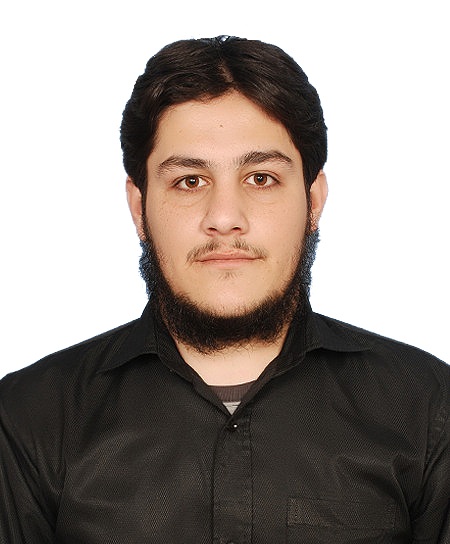}
\end{wrapfigure}\textbf{Muzammil Behzad} received his partially-funded B.S. degree with distinctions (double medalist and valedictorian) from COMSATS University Islamabad (CUI), Pakistan, and his fully-funded M.S. degree from King Fahd University of Petroleum and Minerals (KFUPM), Saudi Arabia, both in Electrical Engineering. Presently, he is serving as Research Associate at Pukyong National University (PKNU), South Korea since May 2018. Before joining PKNU, he served as Research Associate as well as Pioneer Member of Office of Hybrid Learning at CUI since 2013. His research interests are oriented around signal and image processing, wireless sensor networks, and their applications. He is a lifetime member of Pakistan Engineering Council (PEC), student member of Institute of Electrical and Electronics Engineers (IEEE), and member of Society for Industrial and Applied Mathematics (SIAM). He currently holds more than 6 years of experience in teaching and research. He is also the recipient of employee of the year and the best supervised project award from CUI in the fiscal year 2013-14.\\

\begin{wrapfigure}{l}{25mm} 
	\includegraphics[width=1in,height=1.25in,clip]{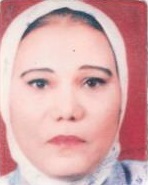}
\end{wrapfigure}\textbf{Manal Abdullah} received her Ph.D. in Computers and Systems Engineering from Ain Shams University, Egypt in 2002. She has vast experience in industrial computer networks and embedded systems. Currently, she is serving as associate professor in the faculty of computing and information technology at King Abdulaziz University, Saudi Arabia. Dr. Manal has published more than 75 research papers in various international journals and conferences. Her research interests include computer networks, performance evaluation and analysis of wireless sensor networks, network management, artificial intelligence, big data analysis, and pattern recognition.\\

\begin{wrapfigure}{l}{25mm}
	\includegraphics[width=1in,height=1.25in,clip]{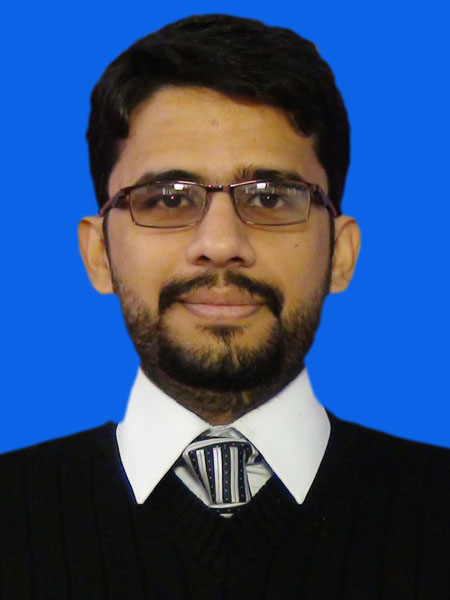}
\end{wrapfigure}\textbf{Muhammad Talal Hassan} received his B.S. degree from Riphah International University, Pakistan, and his M.S. degree from National University of Sciences and Technology (NUST), Pakistan, both in Electrical Engineering, in 2009 and 2013, respectively. Currently, he is working as Lecturer in COMSATS University Islamabad (CUI), Pakistan. He has worked as Research Assistant on the projects ``Efficient Routing in Multi-radio Dynamic-spectrum Access" and ``Network Performance Monitoring for PERN" funded by Higher Education Commission (HEC) of Pakistan. His current research interests include wireless networks, cognitive radio networks, computer networks and wireless communication.\vspace{0.5cm}

\begin{wrapfigure}{l}{25mm} 
	\includegraphics[width=1in,height=1.25in,clip]{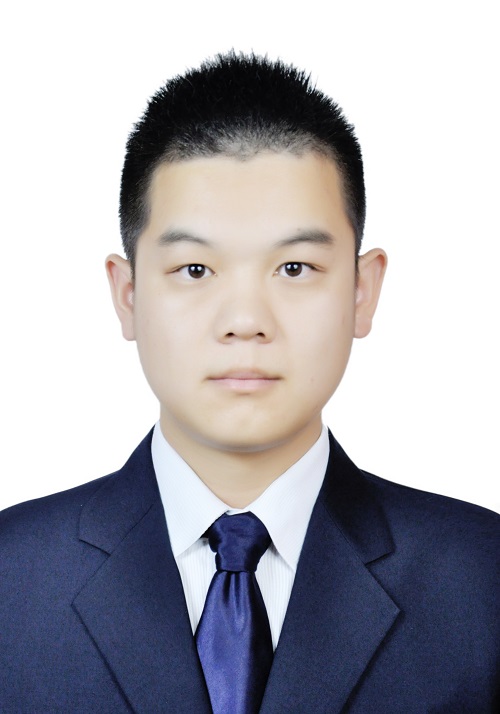}
\end{wrapfigure}\textbf{Yao Ge} received his B.Eng. degree in Electronics and Information Engineering, and his M.Eng. degree (research) in Communication and Information System both from Northwestern Polytechnical University, China, in 2013 and 2016, respectively. He is currently working towards his Ph.D. degree in Electrical Engineering at the Chinese University of Hong Kong (CUHK), Hong Kong. From October 2015 to March 2016, he was a visiting student with the Department of Electrical and Computer Systems Engineering, Monash University, Melbourne, Australia. From April 2016 to August 2016, he was a visiting student with Computer, Electrical and Mathematical Sciences \& Engineering (CEMSE) Division, King Abdullah University of Science and Technology, Saudi Arabia. His current research interests include wireless communication, information theory and statistical signal processing. Mr. Ge was awarded the 2013 Outstanding Bachelor Thesis, the 2016 Outstanding Master Thesis of Northwestern Polytechnical University.\vspace{0.5cm}

\begin{wrapfigure}{l}{25mm} 
	\includegraphics[width=1in,height=1.25in,clip]{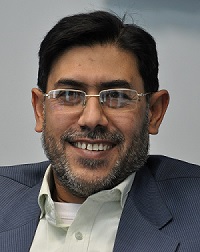}
\end{wrapfigure}\textbf{Mahmood Ashraf Khan} is working as Director in the Center of Advanced Studies in Telecommunication (CAST) at COMSATS University Islamabad (CUI), Pakistan. He has received his B.S. degree in Electrical Engineering from University of Engineering and Technology (UET), Pakistan, and his Ph.D. from Aston University, United Kingdom (UK) in 1992. He has several publications in international conferences and journal of international repute. He is also author of several books. He has executed several funded research projects in collaboration of Japanese researchers, and was also a visiting researcher at Osaka University, Japan. He has completed several community projects in the field of information and communications technology under the umbrella of ITU and APT. His areas of interest are broadband networks, wireless networks and multimedia network architecture.

\end{document}